\nonstopmode

\documentclass[11pt]{book}

\usepackage{hyperref}
\hypersetup{
    colorlinks=true,
    linkcolor=blue,
    filecolor=magenta,      
    urlcolor=blue,
}

\usepackage{amsmath,amssymb,amsthm}
\usepackage{graphics,epsfig,calc}
\usepackage{amsmath}
\usepackage{latexsym,epsfig,bm,amssymb}
\usepackage{color}
\usepackage{amsthm,mathrsfs}

\usepackage{mathptmx}

\DeclareSymbolFont{AMSb}{U}{msb}{m}{n}
\DeclareSymbolFontAlphabet{\mathbb}{AMSb}

\newcommand{\Ran}{{\rm Ran\3}}
\newcommand{\spec}{\mathop{\rm Spec\,}\nolimits}
\newcommand\nab{{\bf \nabla}}
\newcommand{\rot}{{\rm rot}}
\def\supp{\operatorname{supp}\nolimits}
\newcommand{\mes}{{\rm mes}}

\newcommand{\cC}{{\cal C}} 
\newcommand{\cD}{{\cal D}} 

\newcommand{\E}{{\cal E}} 
\newcommand{\cE}{{\cal E}}
\newcommand{\cF}{{\cal F}}
\newcommand{\cH}{{\cal H}}

\newcommand{\cK}{{\cal K}}
\newcommand{\cM}{{\cal M}}
\newcommand{\cO}{{\cal O}}
\newcommand{\cR}{{\cal R}}
\newcommand{\cS}{{\cal S}}
\newcommand{\cV}{{\cal V}}
\newcommand{\cW}{{\cal W}}
\newcommand{\cX}{{\cal X}}
\newcommand{\cY}{{\cal Y}}

\newcommand{\al}{\alpha}
\newcommand{\ccT}{{\cal T}} 

\newcommand{\bba}{{\bf a}} 
 
\newcommand{\n}{{\bf n}} 
\newcommand{\bb}{{\bf b}} 
\newcommand{\e}{{\bf e}}

\newcommand{\x}{{\bf x}} 
\newcommand{\y}{{\bf y}}

\newcommand{\bk}{{\bf k}} 
 
\newcommand{\cm}{{\rm m}}

\newcommand{\ci}{\cite}
\newcommand{\ti}{\tilde}
\newcommand{\de}{\delta}
\newcommand{\De}{\Delta}
\newcommand{\Ga}{\Gamma}

\newcommand{\ds}{\displaystyle}
\newcommand{\fr}{\frac}

\newcommand{\la}{\label}
\newcommand{\lam}{\lambda}
\newcommand{\Lam}{\Lambda}
\newcommand{\na}{\nabla}
\newcommand{\om}{\omega}
\newcommand{\Om}{\Omega}
\newcommand{\vp}{\varphi}

\newcommand{\ov}{\overline}
\newcommand{\pa}{\partial}
\newcommand{\re}{\ref}

\newcommand{\ga}{\gamma}
\newcommand{\h}{\hbar}
\newcommand{\Si}{\Sigma}
\newcommand{\si}{\sigma}

\newcommand{\dist}{\rm dist\5}
\newcommand{\const}{\rm const}
\newcommand{\ve}{\varepsilon}

\newcommand\C{{\mathbb C}}
\newcommand\R{{\mathbb R}}
\newcommand\N{{\mathbb N}}
\newcommand\Z{{\mathbb Z}}
\newcommand\T{{\mathbb T}}

\newcommand{\vka}{\varkappa}

\newcommand{\5}{{\hspace{0.5mm}}}
\newcommand{\3}{{\hspace{0.2mm}}}

\newcommand{\rRe}{{\rm Re\5}}
\newcommand{\rIm}{{\rm Im\5}}
 \newcommand{\st}{\stackrel} 
 \newcommand{\Spec}{\mathop{\rm Spec\,}\nolimits}
 
 \newcommand{\toLpR}{\st{L^p_R}{\longrightarrow}} 
\newcommand{\toLLR}{\st{L^2_R}{\longrightarrow}}

 \newcommand{\tow}{\st{L^2_w}{\longrightarrow}} 
 
  \newcommand{\toLtt}{\st{L^2(\ov\T)}{-\!\!\!-\!\!\!-\!\!\!\!\longrightarrow}}
 \newcommand{\toLp}{\st{L^p}\longrightarrow} 
 \newcommand{\toHd}{\st{H^2}\longrightarrow} 
 \newcommand{\toLd}{\st{L^2}\longrightarrow}

\newcommand{\toYs}{\stackrel{{\cY^0}_{-\al}}{-\!\!-\!\!\!\longrightarrow}}

\newcommand{\toLt}{\st{L^2({\T})}{-\!\!\!-\!\!\!-\!\!\!\!\longrightarrow}} 
\newcommand{\toLwt}{\st{L^2_w({\T})}{-\!\!\!-\!\!\!-\!\!\!\!\longrightarrow}} 
 
\newcommand{\toLC}{\st{C({\T})}{-\!\!-\!\!\!\!\longrightarrow}} 
\newcommand{\tocX}{\st{\cX}{-\!\!\!\!\longrightarrow}}

\newcommand{\Ker}{{\rm Ker\5}}

\newtheorem{theorem}{Theorem}[section]

\newtheorem{definition}[theorem]{Definition}

\newtheorem{lemma}[theorem]{Lemma}
\newtheorem{example}[theorem]{Example}
\newtheorem{remark}[theorem]{Remark}
\newtheorem{remarks}[theorem]{Remarks}
\newtheorem{cor}[theorem]{Corollary}
\newtheorem{proposition}[theorem]{Proposition}

\newcommand{\be}{\begin{equation}}
 \newcommand{\ee}{\end{equation}}
 \newcommand{\ba}{\begin{array}}
 \newcommand{\ea}{\end{array}}
 
\newcommand{\beqn}{\begin{eqnarray}}
 \newcommand{\eeqn}{\end{eqnarray}}

\newcommand{\bd}{\begin{definition}}
 \newcommand{\ed}{\end{definition}}
\newcommand{\bt}{\begin{theorem}}
 \newcommand{\et}{\end{theorem}}
\newcommand{\bp}{\begin{proposition}}
 \newcommand{\ep}{\end{proposition}}
\newcommand{\bl}{\begin{lemma}}
 \newcommand{\el}{\end{lemma}}
\newcommand{\bc}{\begin{cor}}
 \newcommand{\ec}{\end{cor}}
\newcommand{\br}{\begin{remark} }
 \newcommand{\er}{\end{remark}}
\newcommand{\brs}{\begin{remarks} }
 \newcommand{\ers}{\end{remarks}}

 \newcommand{\bex}{\begin{example} }
 \newcommand{\eex}{\end{example}}

\begin{document}

\begin{titlepage}
\vspace{2cm}

\begin{center}
{\Large\bf 
On   stability of solid state  in the 
\medskip\\
Schr\"odinger--Poisson--Newton model
}
\end{center}
\bigskip\bigskip

 \begin{center}
{\large A. Komech}
\\
{\it Faculty of Mathematics of Vienna University\\
and Institute for Information Transmission Problems RAS } \\
e-mail:~alexander.komech@univie.ac.at
\bigskip\\
{\large E. Kopylova}
\\
{\it Faculty of Mathematics of Vienna University\\
and Institute for Information Transmission Problems RAS} \\
 e-mail:~elena.kopylova@univie.ac.at
\end{center}
\vspace{1cm}

 \centerline{Resume}
 \medskip
 
We survey our recent results on stability of 3D crystals in the Schr\"odinger--Poisson--Newton 
model.
We establish orbital stability for the ground state in the 
case of finite periodic crystal and linear stability in the case of  infinite 
crystals under novel Jellium and Wiener conditions on the 
charge density of ions. The corresponding examples are given.

For the finite crystals, the  electron field   is 
described a) by  one-particle Schr\"odinger equation, 
or b) by  $N$-particle Schr\"odinger equation
in the space of antisymmetric wave functions respecting the  Pauli exclusion principle.
The proofs rely on positivity of the Hessian
of Hamiltonian functional in the directions orthogonal to the manifold  of ground states, 
and the Pauli exclusion principle plays the key role in the proof of this positivity.

The problem of spatial periodicity of the ground states is 
reduced to a generalisation  of  Newton-Girard formulas,
and  examples of non-periodic ground states  are constructed.

In the case of  infinite crystals the proofs rely on our novel 
spectral theory of Hamiltonian operators,
which is a special 
version of the Gohberg--Krein--Langer theory of selfadjoint 
operators in the Hilbert spaces with indefinite metric.
We  establish the existence
of the ground states and the dispersive decay for the linearised dynamics.

\bigskip

\bigskip


\end{titlepage}
\tableofcontents
\section{Preface}
The
first mathematical results on the stability of matter were obtained 
by Dyson and Lenard 
in \ci{D1967, DL1968} 
where the energy bound from below
was established.
The thermodynamic limit
for the Coulomb systems  
was studied first by Lebowitz and Lieb 
\ci{LL1969,LL1973}, see the survey and further development in \ci{LS2010}.
These results were extended  by Catto, Le Bris, Lions 
and others to Thomas-Fermi and Hartree-Fock models
\ci{CBL1998,CBL2001,CBL2002}. 
Further results in this direction were established 
by Canc\'es,  Lahbabi, Lewin, Sabin, Stoltz, and others
 \ci{CLL2013,CS2012, BL2005, LS2014-1, LS2014-2}.
All these results concern  either
the convergence of
the ground state of finite particle systems
in the thermodynamic limit or 
the existence of the ground state
for infinite particle systems. 

\smallskip

In the Hartree-Fock model,
the crystal ground state 
 was constructed for the first time by Catto, Le Bris, and  Lions  \ci{CBL2001,CBL2002}.
For the Thomas-Fermi model, see \ci{CBL1998}.

In \ci{CS2012}, Canc\'es and Stoltz have established the well-posedness  for 
the dynamics of  
local perturbations of the crystal ground state 
in the  {\it random phase approximation}
for the reduced  Hartree-Fock equations
with the Coulomb  pairwise interaction potential $w(x-y)=1/|x-y|$.
The  space-periodic nuclear potential
in the equation \ci[(3)]{CS2012}
does not depend on time that corresponds to 
the fixed nuclei positions. 
The nonlinear Hartree-Fock dynamics
for crystals
with the Coulomb potential and
without the  random phase approximation
was not studied previously,
see the discussion in 
\ci{BL2005} and in the introductions of \ci{CLL2013,CS2012}.

In \ci{CLL2013},
E. Canc\`es, S. Lahbabi, and M. Lewin have considered
 the random reduced HF model of crystal  when 
the ions charge density and the electron density matrix are random processes,
and the action of the lattice translations on the probability space is ergodic.
The authors obtain suitable generalizations of the Hoffmann-Ostenhof 
and Lieb-Thirring inequalities  for ergodic density matrices, 
and
construct a random potential which is a solution  to 
 the Poisson equation 
with the corresponding stationary stochastic  charge density. 
The main result is the  coincidence of this model with the thermodynamic limit in  
the case of the short range Yukawa interaction.

In \ci{LS2014-1}, Lewin and Sabin have established the well-posedness for the 
reduced von Neumann equation, describing the Fermi gas,
 with density matrices of infinite trace 
and pair-wise interaction potentials $w\in L^1(\R^3)$. Moreover, the authors  
prove the asymptotic stability of translation-invariant stationary states 
for 2D Fermi gas \ci{LS2014-2}.
\medskip

Traditional {\it one-electron} Bethe--Bloch--Sommerfeld
mathematical model of crystals  reduces to the linear
Schr\"o\-din\-ger
equation with a
space-periodic static potential, which corresponds to the standing ions.
The corresponding spectral theory
is well developed, see \ci{RS4} and the references therein.
The scattering theory for short-range and long-range perturbations of
such `periodic operators' was  constructed in \ci{GN1,GN2}.

In
\ci{ADK2016}, 
Anikin, Dobrokhotov and  Katsnelson 
studied the semiclassical asymptotic approximation of the spectrum of the two-dimensional Schr\"odinger operator with a potential periodic in x and increasing at infinity in y. 
The authors
showed that the lower part of the spectrum has a band structure (where bands can overlap) and calculate their widths and dispersion relations between energy and quasimomenta. The key role in the obtained asymptotic approximation is played by librations, i.e., unstable periodic trajectories of the Hamiltonian system with an inverted potential. 
An effective numerical algorithm for computing the widths of bands 
was presented.
An applications to quantum dimers
was discussed.

\medskip

However,
the dynamical stability of crystals 
with {\it moving ions}
was never considered previously. 
This stability is 
necessary for a rigorous analysis 
of fundamental quantum phenomena in the solid state physics: 
heat conductivity, electric conductivity, thermoelectronic emission, photoelectric effect, 
Compton effect, 
etc., see \ci{BLR}.
Here we survey our recent results on stability
of 3D crystals in the 
Schr\"odinger--Poisson--Newton 
model.
\smallskip

In the present book, we establish orbital stability for the ground state in the 
case of finite periodic crystal and linear stability for infinite 
crystals under novel   Jellium and Wiener conditions on the 
charge density of ions. The corresponding examples are given.

We establish orbital stability for the ground state in the 
case of finite periodic crystal and linear stability in the case of infinite 
crystals under novel Jellium and Wiener conditions on the 
charge density of ions. The corresponding examples are given.

For the finite crystals, the  electron field   is 
described a) by  one-particle Schr\"odinger equation, 
or b) by  $N$-particle Schr\"odinger equation
in the space of antisymmetric wave functions respecting the Pauli exclusion principle.
The proofs rely on positivity of the Hessian
of Hamiltonian functional in the directions orthogonal to the manifold  of ground states, 
and the Pauli exclusion principle plays the key role in the proof of this positivity (see Remark \ref{rPa}).

The problem of spatial periodicity of the ground states is 
reduced to a generalisation  of  Newton-Girard formulas,
and  examples of non-periodic ground states  are constructed (see Sections \ref{snp} and \ref{spp}).

In the case of  infinite crystals the proofs rely on our novel 
spectral theory of Hamiltonian operators 
\ci{KK2014a, KK2014b},
which is a special 
version of the Gohberg--Krein theory of selfadjoint 
operators in the Hilbert spaces with indefinite metric \ci{GK,KL1963,L1981}.
We  establish the existence
of the ground states and the well-posedness and dispersive decay for the linearised dynamics.
\medskip\\
{\bf Acknowledgments.} The authors are grateful to Herbert Spohn for helpful 
discussions and remarks.

The authors are
 indebted to
the Faculty of Mathematics of Vienna University, and the Institute
for the Information Transmission Problems of the Russian Academy
of Sciences for providing congenial facilities for the work.

The work was supported in part by the Department of Mechanics and
Mathematics of  Moscow State University, and 
by the Austrian Science Fund (FWF) (project
nos. P28152 and P27492).
\medskip\bigskip

Alexander Komech and Elena Kopylova

\medskip\bigskip

Vienna,\qquad 15.01.2021

\part{Orbital stability of finite crystals}

In the first part of present book, we consider  finite  crystals
under periodic boundary conditions with one ion per cell of a lattice.

The electron cloud is described by one-particle or $N$-particle
Schr\"odinger equations
while the ions are described as classical charged particles 
moving in an electrostatic potential created by 
the charge densities of electrons and ions.

We
construct the global dynamics    and prove the conservation of energy and charge.
Our main result
is  the  orbital stability of every 
 ground state with periodic arrangement of ions
under novel `Jellium' and `Wiener'  conditions on the ion charge density.  

The presentation mainly relies on our papers \ci{KKjmp2016, KK-SIAM_2017} with suitable extensions.

\chapter{One-particle Schr\"odinger theory }


\centerline{Abstract}
\medskip

In this chapter we describe 
the electron cloud in the crystal  by  one-particle 
Schr\"odinger equation and moving ions.
The ions are described as classical particles   
that corresponds to the  
Born-Oppenheimer  approximation.
The ions interact with the electron cloud via 
the scalar potential, which  is a solution to the corresponding Poisson equation.

We
construct global dynamics, prove the conservation of energy and charge, and 
 give the description of all ground states.
Our main result
is  the  orbital stability of every 
 ground state with periodic arrangement of ions
under novel `Jellium' and `Wiener'  conditions on the ion charge density (\re{Wai}) and 
(\re{W1}).

Moreover, we construct  examples of  non-periodic 
ground states. In these examples  the Wiener condition fails.
This suggests that the periodicity should hold under the  Wiener condition, but this is still an open problem.  
\smallskip

The model with the one-particle Schr\"odinger equation
does not respect the Pauli exclusion principle for electrons.
In the next chapter we extend the orbital stability to the model with $N$-particle  Schr\"odinger equation
respecting the Pauli exclusion principle.
\medskip

\section{Introduction}

We consider crystals which occupy the finite torus $\T:=\R^3/N\Z^3$
and have one ion per cell of the cubic lattice ${\Ga}:=\Z^3/N\Z^3$, where $N\in\N$.
The cubic lattice  is chosen for the simplicity of notations.
We denote by  $\sigma(x)$ the charge density of one ion,
\begin{equation}\la{ro+}
\si\in C^2({\T}),\qquad
\int_{{\T}} \sigma(x)dx=eZ>0, 
\end{equation} 
where $e>0$ is the {\it elementary charge}.
Let $\psi(x,t)$ be the wave function of the electron field, 
$q(n,t)$ denotes the ion displacement  from the reference position $n\in\Ga$,
and $\phi(x)$ be the electrostatic  potential generated by the ions and electrons.
We assume $\hbar=c=\cm=1$, where $c$ is the speed of light and $\cm$ is the electron mass.
Then the considered coupled equations  read
\begin{eqnarray}\la{LPS1}
i\pa_t\psi(x,t)\!\!&=&\!\!-\fr12\De\psi(x,t)-e\phi(x,t)\psi(x,t),\qquad x\in{\T},
\\
\nonumber\\
-\De\phi(x,t)\!\!&=&\!\!\rho(x,t):=\sum_{n\in{\Ga}}
\sigma(x-n-q(n,t))-e|\psi(x,t)|^2,\qquad x\in{\T},
\la{LPS2}
\\
\nonumber\\
M\ddot q(n,t)
\!\!&=&\!\!-(\na\phi(x,t),\sigma(x-n-q(n,t))), 
\qquad n\in{\Ga}.
\la{LPS3}
\end{eqnarray}
Here the 
brackets $(\cdot,\cdot)$
 stand for the  scalar product on the real Hilbert
space $L^2({\T})$ and for its different extensions,  and $M>0$ is the mass of one ion.
All derivatives here and below are understood in the sense of distributions.
Similar finite periodic approximations of crystals are treated in all textbooks on 
quantum theory of solid state \ci{Born, Kit, Zim}. 
However,  the stability of ground states in this model was never discussed.

Obviously, 
\begin{equation}\la{r0}
\ds\int_{\T}\rho(x,t)dx=0 
\end{equation}
by the  Poisson equation (\re{LPS2}).
Hence, the potential $\phi(x,t)$ can be eliminated from  the system (\re{LPS1})--(\re{LPS3})
using the operator $G:=(-\De)^{-1}$, see (\re{fs}) for a more precise definition. 
Substituting $\phi(\cdot,t)=G\rho(\cdot,t)$
into equations (\ref{LPS1}) and (\ref{LPS3}), we can write the system as
\begin{equation}\la{vf}
\dot X(t)=F(X(t)),\qquad t\in\R,
\end{equation}
where $X(t)=(\psi(\cdot,t), q(\cdot,t), p(\cdot,t))$ with $p(\cdot,t):=\dot q(\cdot,t)$.
The system (\ref{LPS1})--(\ref{LPS3}) is equivalent,  up to a  gauge transform (see the next section), 
to equation (\ref{vfN}) with the normalization
\begin{equation}\la{rQ}
\Vert\psi(\cdot,t)\Vert_{L^2({\T})}^2=ZN^3,\qquad t\in\R,
\end{equation}
which follows from (\re{r0}). If the integral (\ref{ro+}) vanishes, 
we have $Z=0$ and  $\psi(x,t)\equiv 0$.

We will identify the complex  functions $\psi(x)$ with 
two real functions $\psi_1(x):=\rRe\psi(x) $ 
and  $\psi_2(x):=\rIm\psi(x)$.
Now
equation (\ref{vf}) 
is equivalent to the Hamiltonian system
\begin{equation}\la{HSi}
\pa_t \psi_1(x,t)=\fr12 \pa_{\psi_2(x)}E,~~\pa_t \psi_2(x,t)=-\fr12\pa_{\psi_1(x)}E,~~
\pa_t q(n,t)= \pa_{p(n)}E,~~\pa_t p(n,t)=-\pa_{q(n)} E
\end{equation}
together with the normalisation condition (\ref{rQ}).
Here  the Hamiltonian functional (energy) reads 
\begin{equation}\la{Hfor}
  E(\psi, q, p)=\fr12\int_{{\T}}|\na\psi(x)|^2dx+\fr12(\rho,G\rho)+\sum_{n\in{\Ga}} \fr{p^2(n)}{2M},
\end{equation}
where  $ q:=(q(n): ~~n\in{\Ga})\in[{\T}]^{\ov N}$, $ p:=(p(n):~~n\in{\Ga})\in\R^{3\ov N}$ with $\ov N:=N^3$, and 
\begin{equation}\la{Hfor2}
\rho(x):=
\sum_{n\in{\Ga}}\si(x-n-q(n))-e|\psi(x)|^2,\qquad x\in{\T}.
\end{equation}
We prove 
the global well-posedness of the dynamics
and
the energy and charge conservations (\ref{EQ}).
Our main goal is the stability of 
ground states, i.e.,
solutions to (\re{LPS1})--(\re{LPS3})
with minimal (zero) energy  (\re{Hfor}).
We will consider only  ${\Ga}$-periodic ground states. 
 Nonperiodic ground states exist for some degenerate 
densities $\si$, 
see Remark \re{r1} ii) and Section \re{snp}.

We will see that  ${\Ga}$-periodic ground states  can be stable  
depending on the choice of the ion density $\sigma$.
We study special densities 
$\sigma$ satisfying some conditions below. Namely,
we will assume
 the following condition on the ion charge density,
\begin{equation}\la{Wai}
 \mbox{\bf The Jellium Condition:}~~~~~ \hat\si(\xi)
  :=\int_{\T} e^{i\xi x}\si(x)dx
 =0,\quad \xi\in {\Gamma^*_1}\setminus 0,
~~~~~~~~~~~~ ~~~~~~~~~~~~~~~~~~ ~
\end{equation}
where ${\Gamma^*_1}:=2\pi\Z^3$.
This condition immediately implies that the periodized ion charge density is a positive 
constant everywhere on the torus:
\begin{equation}\la{sipi}
	\rho^i(x):=\sum_{n\in{\Ga}}\si(x-n)\equiv eZ,\qquad x\in{\T}.
\end{equation}
 The simplest example of such a 
$\sigma$ is a constant over the unit cell of a given lattice, which is what physicists 
usually call {\it Jellium} \cite{GV2005}. We give further examples in Section \re{sex}.
Here we study this model in the rigorous context of the Schr\"odinger-Poisson equations.

Furthermore, we will assume a spectral property of the Wiener type
\begin{equation}\la{W1}
\mbox{\bf The Wiener Condition:}~~~\Si(\theta):=\sum_{m\in\Z^3}\Big[
 \fr{\xi\otimes\xi}{|\xi|^2}|\hat\si(\xi)|^2\Big]_{\xi=\theta+2\pi m}>0,\
\quad \theta\in \Pi^*_N\setminus {\Gamma^*_1},
\end{equation}
where the Brillouin zone $\Pi^*_N$ is defined by
\begin{equation}\la{PPG}
 \Pi^*_N:= \{\xi=(\xi^1,\xi^2,\xi^3)\in{\Gamma^*_N}:0\le \xi^j\le 
  2\pi,~~j=1,2,3\},\quad{\Gamma^*_N}:=\fr{2\pi}N\Z^3.
\end{equation}
This condition is
 an analog of the Fermi Golden Rule
for crystals. 
It is independent of (\re{Wai}).
 We have introduced  
conditions of type (\re{Wai}) and  (\re{W1})
for the first time
 in \ci{KKpl2015} 
 in the framework of infinite crystals.
\br\la{rW}
{\rm
i) The series \eqref{W1} converges for $\theta\in{\Gamma^*_N}\setminus{\Gamma^*_1}$
by the Parseval identity since $\sigma\in L^2(\T)$ by \eqref{ro+}.
\\
ii) The matrix $\Si(\theta)$ is $\Gamma^*_1$-periodic outside  $\Gamma^*_1$.
Thus, (\re{W1}) means that $\Si(\theta)$ is a positive matrix
for $\theta\in \Pi^*_N\setminus 0$, where $\Pi^*_N$
is the `discrete Brillouin zone' $\Ga_N^*/\Gamma^*_1$.
}
\er
The series \eqref{W1} is a nonnegative matrix.
Hence,  the Wiener condition holds `generically'. 
\bex
{\rm
 \eqref{W1} holds if
\begin{equation}\la{W1s}
\hat\si(\xi)\ne 0,\qquad \xi\in {\Gamma^*_N}\setminus{\Gamma^*_1},
\end{equation}
i.e., (\re{Wai}) are the only zeros of $\hat\si(\xi)$. 
}
\eex

However,  (\re{W1}) does not hold 
for the simplest Jellium model, when $\sigma$ is constant on the unit cell, see (\re{sic}) and (\re{sJM}). 

The energy (\re{Hfor}) is nonnegative, and its minimum is zero.
We show in Lemma \re{Jgs} that under Jellium condition (\re{Wai}) all ${\Ga}$-periodic
ground states are zero energy stationary solutions of the form
\begin{equation}\la{gr}
S_{\al,r}=(\psi^\al, \ov r,0),\qquad \al\in [0,2\pi], \quad r\in{\T},
\end{equation}
where $\psi^\al(x)\equiv e^{i\al}\sqrt{Z}$ and $\ov r\in [{\T}]^{\ov N}$  is defined by
\begin{equation}\la{gr2}
 \ov r(n)=r,\qquad n\in{\Ga}.
\end{equation}
The corresponding electronic charge  density reads
\begin{equation}\label{roZ}
 \rho^e(x):=-e|\psi^\al(x)|^2\equiv -eZ,\qquad x\in {\T}.
 \end{equation}
Hence,
the corresponding total charge density (\re{Hfor2}) identically vanishes by (\re{sipi}). 
Let us emphasize that 
 both ionic and  electronic charge densities  are uniform for the ground state under the Jellium condition.

\smallskip

Our main result (Theorem \re{tm}) is the stability  
of the real 4-dimensional `solitary manifold'
\begin{equation}\la{cS}
\cS=\{S_{\al ,r}: 
~\al \in [0,2\pi],~r\in{\T} \}.
\end{equation}
The stability means that any solution $X(t)=(\psi(\cdot,t), q(\cdot,t), p(\cdot,t))$ 
to (\re{vf})
 with initial data, 
lying in the vicinity of
the manifold $\cS$, is close to it uniformly in time. 
This is the `orbital stability' in the sense of \ci{GSS87}, since the manifold
$\cS=S^1\times{\T}\times \{0\}$ is the orbit of the symmetry group $U(1)\times{\T}$.
 Obviously, 
\begin{equation}\la{ES}
E(S)= 0,\qquad S\in\cS.
\end{equation}

Let us comment on our approach.
We prove the local well-posedness for the system  (\re{HSi})
by
the contraction mapping principle.
The global  well-posedness for the equation (\re{vf})
and the charge and energy conservation
follow by the Galerkin approximations and the uniqueness of solutions. 
We need  the charge 
conservation to return back 
from 
 (\re{HSi})
to the system  (\re{LPS1})--(\re{LPS3}).

The orbital stability of the solitary manifold $\cS$ 
is deduced
from  the lower energy estimate
\begin{equation}\la{BLi}
E(X)\ge \nu\,d^2(X,\cS)\qquad{\rm if}\qquad d(X,\cS)\le \delta,\quad X\in\cM,
\end{equation}
where 
$\cM$ is the manifold defined by  the normalization  (\re{rQ}) (see Definition \re{dM});
 $\nu,\delta>0$ and `$d$' is the distance in the `energy norm'.
This estimate obviously implies  the stability of the solitary manifold $\cS$.
We deduce  (\re{BLi}) from the positivity of the Hessian $E''(S)$
for $S\in \cS$ in the orthogonal directions to $\cS$ on the manifold $\cM$.
The Jellium and Wiener conditions are sufficient 
for this positivity.
We expect that these conditions are also necessary;
however, this is still an open  problem.
Anyway,
the positivity can break
down when these conditions  fail. We have shown this in \ci[Lemma 10.1]{KKpl2015}
in the context of infinite crystals, however the proof extends directly to the finite
crystals.
\begin{remarks}\label{r1}
{\rm
 i) In the case of infinite crystal, corresponding to $N=\infty$, the orbital stability seems  impossible. 
 Namely, for  $N=\infty$ the estimates (\ref{eq}), (\ref{GP2}) and (\ref{fp}) break down,
 as well as the estimate of type (\ref{BLi}) which is due to the discrete spectrum of the energy
 Hessian $E''(S)$ on the compact torus.
\smallskip\\
ii) We  show that the identity of type 
(\re{sipi}) holds  for a wide set of  
arrangements of ions  which are not $\Gamma_1$-periodic,
if $\si$ satisfy additional spectral conditions.
The corresponding examples are given, but in all our examples 
 the Wiener condition breaks down. We suppose that the Wiener condition provides
the periodicity (\re{gr2}), however this is a challenging open problem, 
see Section \re{spp}.  We prove the orbital stability only  for $\Gamma_1$-periodic ground states.

}
\end{remarks}

 This chapter is organized as follows.
 In Section 2 we eliminate the potential and reduce the dynamics  to the integral equation.
In Sections 3 
we prove the well-posedness.
In Section 4  we prove 
the stability of the solitary manifold $\cS$ establishing  
the lower estimate for the energy. 
In Appendices we 
prove the conservation
of the energy and charge, 
describe all ground states and give some examples.


\setcounter{equation}{0}
\section{Reduction to the integral equation}
The operator $G:=(-\De)^{-1}$   is well defined in the Fourier series:
\begin{equation}\la{fs}
\rho(x)=\sum_{\xi\in{\Gamma^*_N}}\hat\rho(\xi)e^{i\xi x},
\qquad G\rho:=\sum_{\xi\in{\Gamma^*_N}\setminus 0}\fr{\hat\rho(\xi)}{\xi^2}e^{i\xi x}.
\end{equation}
 The Poisson equation (\ref{LPS2})
implies that $\hat\rho(0,t)=\ds\int\rho(x,t)\,dx=0$, which is  equivalent to (\re{rQ}). 
 Hence,
$\phi(\cdot,t)=G\rho(\cdot,t)$ up to an additive constant $C(t)$ which can be compensated by a gauge transform 
\be\la{GT}
\psi(x,t)\mapsto\psi(x,t)\exp(-ie\ds\int_0^t C(s)ds).
\ee 
The system (\re{HSi}) 
can be written as 
\begin{equation}\la{HS}
\dot X(t)=JE'(X(t)),
\qquad X(t):=(\psi_1(t), \psi_2(t),  q(t),p(t)),
\end{equation}
where 
\begin{equation}\la{HS2}
J=\left(
\begin{array}{rrrr}
0&1/2 &0&0\\
-1/2&0&0&0\\
0&0   &0&1\\
0&0   &-1&0
\end{array}
\right).
\end{equation}
For $\psi,\vp\in L^2(\T)$ denote
\be\la{sp}
(\psi,\vp):=\int_{\T}\psi(\ov x)\cdot\vp(\ov x)d\ov x,
\ee
where $\cdot$ is the inner product of the corresponding vectors in $\R^2$.
In particular,
\begin{equation}\la{1i}
(1,i)=0.
\end{equation}

\bd
i) Denote the  real Hilbert spaces
\be\la{XW}
\cX:=L^2({\T})\oplus\R^{3\ov N}\oplus \R^{3\ov N},
\qquad \cW:=H^1({\T})\oplus\R^{3\ov N}\oplus \R^{3\ov N}.
\ee
ii) $\cV:=
H^1({\T})\times[{\T}]^{\ov N}\times \R^{3\ov N}$ is 
the Hilbert manifold endowed with the  metric 
\begin{equation}\la{dVs}
d_{\cV}(X,X'):=\Vert\psi-\psi'\Vert_{H^1({\T})}+|q-q'|+|p-p'|,\qquad X=(\psi,q,p),
\quad
 X'=(\psi',q',p')
\end{equation}
and with the `quasinorm'
\begin{equation}\la{cVs}
|X|_{\cV}:=\Vert\psi\Vert_{H^1({\T})}+|p|,\qquad X=(\psi,q,p).
\end{equation}

\ed

The linear space $\cW$ is  isomorphic to the tangent space to 
the Hilbert manifold $\cV$ at each point $X\in\cV$. 
Denote 
by the brackets
$\langle\cdot,\cdot\rangle$
the scalar product in $\cX$:
\begin{equation}\la{dWW}
\langle Y,Y'\rangle:=(\vp,\vp')+\vka\vka'+\pi\pi',\qquad Y=(\vp,\vka,\pi), 
\quad Y'=(\vp',\vka',\pi').
\end{equation}
The total charge of electrons is defined (up to a factor) by
\begin{equation}\la{Q}
Q(X):= \int|\psi(x)|^2dx,\qquad X=(\psi,q,p)\in\cX.
\end{equation}
Obviously,
\begin{equation}\la{EQV}
|X|_\cV^2\le C[E(X)+Q(X)],\qquad X\in\cV,
\end{equation}
The system 
(\re{HS})
is a  nonlinear finite-dimensional perturbation of the free Schr\"odinger equation.
We will prove  that a 
solution $X\in C(\R,\cV)$ exists and is unique for any initial state $X(0)\in\cV$, and the
energy and the electronic charge are conserved,
\begin{equation}\la{EQ}
E(X(t))=E(X(0)),\quad Q(X(t))=Q(X(0)),\qquad t\in\R.
\end{equation}
The energy (\re{Hfor}) and the charge are well defined and continuous 
on $\cV$ in the metric $d_\cV$ by the estimate (\re{Th}) below.
The charge conservation holds 
by the Noether theory \ci{A, GSS87, KQ} due to the $U(1)$-invariance of the Hamiltonian functional:
\begin{equation}\la{U1}
E(e^{i\al }\psi,q,p)=E(\psi,q,p),\qquad (\psi,q,p)\in\cV,\quad \al \in\R.
\end{equation}
We rewrite the system (\ref{HS}) in the integral  form
\begin{equation}\la{LPSi}
\left\{\begin{array}{lll}
\psi(t)&=&e^{-iH_0t} \psi(0)+ie\ds\int_0^t e^{-iH_0(t-s)} [ \phi(s)\psi(s) ]ds,\\
q(n,t)&=&q(n,0)+\frac 1M\ds\int_0^t p(n,s)ds\mod N\Z^3,\\
p(n,t)&=&p(n,0)-\ds\int_0^t (\nabla \phi(s),\sigma(\cdot-n-q(n,s))) ds,
\end{array}\right|
\end{equation}
where $H_0:=-\fr12\De$ and $\phi(s):=G\rho(s)$.
In the vector form (\ref{LPSi}) reads
\begin{equation}\la{LPSiv}
X(t)=e^{-At}X(0)+\int_0^t  e^{-A(t-s)} N(X(s)) ds,
\qquad 
A=\left(\begin{array}{ccc}
iH_0 & 0 & 0\\
0&0&0\\
0&0&0
\end{array}\right),
\end{equation}
where
\begin{equation}\la{HN}
N(X)=(ie \phi\psi ~, p,~f),\qquad 
f(n):=-(\nabla \phi,\sigma(\cdot-n-q(n))),\qquad\phi:=G\rho,
\end{equation}
and $\rho$, $G\rho$ are defined by (\re{Hfor2}) and 
(\ref{fs}) respectively.

\setcounter{equation}{0}
\section{Global dynamics}\la{Gd}
In this section we prove the well-posedness of the dynamics.

\bt\label{TLWP1}(Global well-posedness).
Let   (\re{ro+}) hold and $X(0)\in\cV$. Then 
\medskip\\
i) 
 Equation (\ref{HS}) admits  a unique  solution $X\in C(\R,{\cal V})$,
 and the maps $U(t):X(0)\mapsto X(t)$ are continuous in $\cV$ for $t\in\R$.
 \medskip\\
ii) The conservation laws (\re{EQ}) hold.
\medskip\\
iii) $X$ is the solution to (\ref{LPS1})--(\ref{LPS3}) if 
\begin{equation}\la{rQ2}
Q(X(0))=Z\ov N^3.
\end{equation}
\et

First,  let us  prove the local well-posedness.

\bp\label{TLWP}(Local well-posedness).
Let   (\re{ro+}) hold and $|X(0)|_\cV\le R$. Then 
there exists $\tau=\tau(R)>0$ such that
  equation (\ref{HS}) has  a unique  solution $X\in C([-\tau,\tau],{\cal V})$,
 and the maps $U(t):X(0)\mapsto X(t)$ are continuous in $\cV$ for $t\in [-\tau,\tau]$.

 \ep
In the next two lemmas
we prove
the boundedness and the local Lipschitz continuity of the nonlinearity $N:\cV\to\cW$.
With this proviso Proposition \re{TLWP} follows from the integral form 
(\re{LPSiv}) of the equation (\ref{HS})
 by the contraction mapping principle, since $e^{-At}$ is an isometry of $\cW$.
\begin{lemma}\label{p1}
For any $R>0$ and $|X|_\cV\le R$
\begin{equation}\label{bN}
\Vert N(X)\Vert_\cW\le C(R)
\end{equation}
\end{lemma}
\begin{proof}
First, we estimate the potential $\phi(\cdot,t):=G\rho(\cdot,t)$.
Namely, 
we split $\rho(x,t)$ as
\begin{equation}\label{split}
  \rho(x,t)=\rho^i(x,t)+\rho^e(x,t),
\end{equation}
where $\rho^i$ and $\rho^e$ are the ion and electron charge densities respectively,
\[
\rho^i(x,t)=\sum_{n\in{\Ga}}\sigma(x-n-q(n,t)),\qquad \rho^e(x,t)=-e|\psi(x,t)|^2.
\]
Applying the Cauchy-Schwarz inequality to the second formula (\ref{fs}), 
we obtain that
\begin{equation}\label{Gpsi}
\Vert \phi\Vert_{C({\T})}\le C\Vert\hat \rho\Vert_{L^2({\Gamma^*_N})}=C\Vert \rho\Vert_{L^2({\T})}
\le C(\Vert \rho^i\Vert_{L^2({\T})}+e\Vert \psi\Vert_{L^4({\T})}^2)\le
 C_1 (1+\Vert\psi\Vert_{H^1({\T})}^2)
\end{equation}
since $H^1({\T})\subset L^6({\T})$ by the Sobolev
embedding theorem.
On the other hand, the H\"older inequality implies that
\begin{equation}\label{np}
\Vert\nabla \rho^e\Vert_{L^{3/2}({\T})}\le e
\Vert\na|\psi|^2\Vert_{L^{3/2}({\T})}
\le C_1\Vert\psi\Vert_{L^6({\T})}\Vert\nabla\psi\Vert_{L^2({\T})}
\le C_2\Vert\psi\Vert_{H^1({\T})}^2.
\end{equation}
Therefore, we get by the Hausdorff-Young and the H\"older inequalities  \ci{Her1}
\begin{eqnarray}\nonumber
\Vert \nabla \phi\Vert_{L^3({\T})}\!\!&\le&\!\! C\Vert\widehat{\nabla \phi}\Vert_{L^{3/2}({\Gamma^*_N})}
\le C_1\Vert \xi\hat \rho\Vert_{L^3({\Gamma^*_N})}
\Big[\sum_{\xi\in{\Gamma^*_N}\setminus 0}|\xi|^{-6}\Big]^{1/3}
\le C_2\Vert\nabla \rho\Vert_{L^{3/2}({\T})}\\
\label{eq}
\!\!&\le&\!\! C_2(\Vert\nabla \rho^{i}\Vert_{L^{3/2}({\T})}+\Vert\nabla \rho^{e}\Vert_{L^{3/2}({\T})})
\le C_3(1+\Vert\psi\Vert_{H^1({\T})}^2).
\end{eqnarray}
Now  (\ref{Gpsi})  and  (\ref{eq}) imply by the H\"older inequality 
\begin{eqnarray}
\Vert\psi \phi\Vert_{L^2({\T})}&\le& \Vert \phi\Vert_{C({\T})}\cdot\Vert\psi \Vert_{L^2({\T})}
\le C(1+\Vert\psi\Vert_{H^1({\T})}^3)
\nonumber\\
\nonumber\\
\Vert \nabla\psi\phi\Vert_{L^2({\T})}&\le& \Vert \phi\Vert_{C({\T})}\Vert\nabla\psi \Vert_{L^2({\T})}\le C(1+\Vert\psi\Vert_{H^1({\T})}^3)
\nonumber\\
\nonumber\\
\Vert \psi\nabla \phi \Vert_{L^2({\T})}
&\le& C\Vert\psi\Vert_{L^6({\T})}\cdot \Vert\nabla \phi\Vert_{L^3({\T})}
\le C_1(1+\Vert\psi\Vert_{H^1({\T})}^3).
\end{eqnarray}
Hence,
\begin{equation}\label{eq2}
\Vert\phi\psi \Vert_{H^1({\T})}\le C(1+\Vert\psi\Vert_{H^1({\T})}^3).
\end{equation}
Finally, (\ref{eq}) and (\ref{ro+}) imply that
\begin{equation}\label{eq3}
|f(n)|\le \Vert \phi\Vert_{C({\T})}\Vert\na\si\Vert_{L^1({\T})}
\le C(1+\Vert\psi\Vert_{H^1({\T})}^2),\qquad n\in{\Ga}.
\end{equation}
At last, (\ref{bN}) holds by  (\ref{eq2}) and  (\ref{eq3}).
\end{proof}
It remains
to prove  that the nonlinearity is  locally Lipschitz.

\begin{lemma}\label{p2}
For any $R>0$ and $X_1,X_2\in\cV$ 
\begin{equation}\label{lN}
\Vert N(X_1)-N(X_2)\Vert_\cW\le C'(R) d_\cV(X_1, X_2)\qquad{\rm if}\quad
| X_1|_\cV, |X_2|_\cV\le R.
\end{equation}
\end{lemma}

\begin{proof}
Writing  $X_k=(\psi_k,q_k,p_k)$ and $\phi_k=G\rho_k$, we obtain that
\begin{equation}\label{NN1}
\Vert\phi_1\psi_1-\phi_2\psi_2\Vert_{H^1({\T})}
\le
\Vert (\phi_1-\phi_2)\psi_1\Vert_{H^1({\T})}
+\Vert \phi_2(\psi_1-\psi_2)\Vert_{H^1({\T})}.
\end{equation}
Similarly to (\ref{Gpsi})\,-\,(\ref{eq}) we obtain
\begin{eqnarray}\nonumber
\Vert \phi_2(\psi_1-\psi_2)\Vert_{H^1({\T})}&\le&
\Vert \phi_2\Vert_{C({\T})}
\Vert\psi_1-\psi_2\Vert_{H^1({\T})}+\Vert \nabla \phi_2\Vert_{L^3({\T})}
\Vert\psi_1-\psi_2\Vert_{L^6({\T})}\\
\label{NN2}
&\le&C(1+R^2)\Vert\psi_1-\psi_2\Vert_{H^1({\T})}\le C(R)d_\cV(X_1,X_2)
\end{eqnarray}
Further, similarly to (\ref{np}),
\begin{equation}\label{np2}
\Vert\nabla (\rho^e_1-\rho^e_2)\Vert_{L^{3/2}({\T})}
\le C\Vert\psi_1-\psi_2\Vert_{H^1({\T})}[\Vert\psi_1\Vert_{H^1({\T})}+\Vert\psi_2\Vert_{H^1({\T})}].
\end{equation}
Moreover, $|\si(x)-\si(x-a)|\le C|a|$,
where $|a|:=\min_{r\in a} |r|$ for $a\in \T$.
Hence, similarly to (\ref{eq}),
\begin{align*}
\Vert(\phi_1-\phi_2)\psi_1\Vert_{H^1({\T})}\le
\Vert \phi_1-\phi_2\Vert_{C({\T})}
\Vert\psi_1\Vert_{H^1({\T})}+
\Vert \nabla (\phi_1-\phi_2)\Vert_{L^3({\T})}
\Vert\psi_1\Vert_{L^6({\T})}
\\
\le 
C R\Big[\Vert \rho_1^i-\rho_2^i\Vert_{L^2({\T})}+\Vert \rho_1^e-\rho_2^e\Vert_{L^2({\T})}
+\Vert\nabla(\rho_1^i-\rho_2^i)\Vert_{L^{3/2}({\T})}+\Vert\nabla(\rho_1^e-\rho_2^e)\Vert_{L^{3/2}({\T})}\Big]
\\
\le C_1 R(|q_1-q_2|+R\Vert \psi_1-\psi_2\Vert_{H^1({\T})})
\le C(R)d_\cV(X_1,X_2).
\end{align*}
Now (\re{NN1}) and  (\re{NN2}) give
\begin{equation}\la{Th}
\Vert\phi_1\psi_1-\phi_2\psi_2\Vert_{H^1({\T})}\le C(R)
d_\cV(X_1,X_2).
\end{equation}
Similarly,
\begin{eqnarray}
&&| (\nabla \phi_1,\sigma(\cdot-n-q_1(n)))- (\nabla \phi_2,\sigma(\cdot-n-q_2(n)))|
\nonumber\\
\nonumber\\
&&\le |(\nabla \phi_1-\nabla\phi_2,\sigma(\cdot-n-q_1(n)))|
+ |(\nabla \phi_2,\sigma(\cdot-n-q_1(n))-\sigma(\cdot-n-q_2(n)))|
 \nonumber\\
\nonumber\\
&&\le C(\Vert \phi_1-\phi_2\Vert_{C({\T})}+\Vert \phi_2\Vert_{C({\T})}| q_1-q_2|\le C(R)
d_\cV(X_1,X_2).
\end{eqnarray}
This estimate and  (\re{Th}) imply (\re{lN}).
\end{proof}
Now Proposition \re{TLWP} follows from Lemmas \ref{p1} and \ref{p2}. 
\medskip\\
{\bf Proof of Theorem \re{TLWP1}.}
The local solution  $X\in C([-\tau,\tau],\cV)$  
of (\ref{HS})
exists and is unique by Proposition \re{TLWP}.
On the other hand, the conservation laws (\ref{EQ}) (proved in Proposition \re{lgal} iii)) 
together with   (\re{EQV}) imply a priori bound
\begin{equation}\la{ab}
| X(t)|_\cV^2\le C[E(X(0))+ Q(X(0))],\qquad t\in [-\tau,\tau].
\end{equation}
Hence, the local solution admits an extension to the global one $X\in C(\R,\cV)$. 
Further, the normalisation (\ref{rQ2}) implies that $Q(X(t))=Z\ov N^3$ for all $t\in\R$ by the charge conservation 
(\re{EQ}). Hence, (\ref{HS}) implies (\ref{LPS1})--(\ref{LPS3}).
\hfill$\Box$

\setcounter{equation}{0}
\section{Conservation laws}
We deduce the conservation  laws (\ref{EQ}) by the Galerkin approximations \ci{Lions}. 

\bd
i)$\cV_m$ with $m\in\N$  denotes finite dimensional  submanifold of $\cV$ formed by
\begin{equation}\la{Vm}
(\sum_{k\in{\Gamma^*_N}(m)} C_k e^{ikx},q,p),  \qquad q\in{\T}^{\ov N}, \quad p\in\R^{3\ov N}.
\end{equation}
where ${\Gamma^*_N}(m):=\{k\in{\Gamma^*_N}:  k^2\le m\}$.
\medskip\\
ii) $\cW_m$ with $m\in\N$  denotes the finite dimensional linear subspace of $\cW$
spanned by 
\begin{equation}\la{Wm}
(\sum_{k\in{\Gamma^*_N}(m)} C_k e^{ikx},\vka,v),  \qquad \vka\in\R^{3\ov N}, \quad v\in\R^{3\ov N}.
\end{equation}
\ed
 Obviously, $\cV_1\subset\cV_2\subset...$, 
the union $\cup_m\cV_m$ is dense in $\cV$, and 
 $\cW_m$ are  invariant with respect 
to $A$ and $J$.
Let us denote by $P_m$
the orthogonal projector $\cX\to\cW_m$. This projector is also orthogonal in $\cW$.
Let us 
approximate the system 
(\re{HS}) by the finite-dimensional Hamiltonian systems  on the manifold $\cV_m$,
\begin{equation}\la{gal}
\dot X_m(t)=JE_m'(X_m(t)),\qquad t\in\R,
\end{equation}
where $E_m:=E|_{\cV_m}$ and $X_m(t)=(\psi_m(t),q_m(t),p_m(t))\in C(\R,\cV_m)$.
The equation (\re{gal}) can be also  written as
\begin{equation}\la{gali}
\langle \dot X_m(t),Y\rangle=-\langle E'(X_m(t)),JY\rangle,\qquad Y\in\cW_m.
\end{equation}
This form of the equation (\re{gal}) holds since $E_m:=E|_{\cV_m}$ and $\cW_m$ is  invariant 
with respect to  $J$. Equivalently,
\begin{equation}\la{gali2}
\dot X_m(t)=-A\, X_m(t) + P_m N(X_m(t)).
\end{equation}

The Hamiltonian form guarantees the energy and charge conservation
 (\ref{EQ}):
\begin{equation}\la{EQ2}
E(X_m(t))=E(X_m(0)),\quad Q(X_m(t))=Q(X_m(0)),\qquad t\in\R.
\end{equation}
Indeed, the energy conservation holds by the Hamiltonian form  (\re{gal}), 
while the charge conservation holds
by
the Noether theory \ci{A,GSS87, KQ} due to the
$U(1)$-invariance of $E_m$, see (\re{U1}).

The equation (\re{gali2}) admits a unique local solution for every initial state
 $X_m(0)\in\cV_m$ since the right hand side 
is locally bounded and Lipschitz continuous. 
The global solutions exist by   (\re{EQV})
and
the energy and charge conservation (\re{EQ2}).
\medskip

Finally, we take any $X(0)\in\cV$ and choose a sequence 
\begin{equation}\la{s0}
X_m(0)\to X(0),\qquad m\to\infty,
\end{equation}
where the convergence holds in the metric of $\cV$.
 Therefore, 
\begin{equation}\la{EQm}
E(X_m(0))\to E(X(0)),\qquad Q(X_m(0))\to Q(X(0)).
\end{equation}
Hence, (\ref{EQ2}) and (\re{EQV}) imply the basic uniform bound 
\begin{equation}\la{Vb}
R:=\sup_{m\in\N}\,\,\sup_{t\in\R}| X_m(t)|_\cV <\infty.
\end{equation}
Therefore, 
(\ref{gali2}) and Lemma \re{p1} imply  the second basic uniform bound
\begin{equation}\la{Vb2}
\sup_{m\in\N}\,\sup_{t\in\R}\,\Vert \dot X_m(t)\Vert_{\cW^{-1}} <C(R),
\end{equation}
since  the operator 
$A:\cW\to\cW^{-1}$ is bounded, and the projector
$P_m$ is also a bounded operator in $\cW\subset \cW^{-1}$.
Hence, 
the Galerkin
approximations $X_m(t)$ are uniformly Lipschitz-continuous with values in $\cV^{-1}$:
\begin{equation}\la{ecg}
\sup_{m\in\N}\, d_{\cV^{-1}}(X_m(t),X_m(s))\le C(R)|t-s|,\qquad s,t\in\R.
\end{equation}
Let us show that  
the uniform estimates   (\ref{Vb}) and (\ref{ecg}) provide a compactness of the 
Galerkin approximations and the conservation laws. 
Let us recall that $\cX:=\cV^0$ and $\cV:=\cV^1$.
\bp\la{lgal} Let   (\re{ro+}) hold and $X(0)\in\cV$. Then 
\medskip\\
i) 
There exists
a subsequence $m'\to\infty$ such that
\begin{equation}\la{ss}
X_{m'}(t)\tocX X(t),\qquad m'\to\infty,\qquad t\in\R,
\end{equation}
where $X(\cdot)\in C(\R, \cX)$.
\medskip\\
ii) Every limit function $X(\cdot)$
is a solution to  (\re{LPSiv}), and $X(\cdot)\in C(\R,\cV)$.
\medskip\\
iii) The conservation laws (\re{EQ}) hold.
\ep
\begin{proof}
i) The convergence (\re{ss}) follows from 
 (\re{Vb}) and  (\re{Vb2}) 
by the Dubinsky
 `theorem on three spaces' \ci{Dub65}  (Theorem 5.1 of
\ci{Lions}). Namely, the embedding $\cV\subset\cX$ is compact by the Sobolev theorem 
\ci{Adams},
and hence, (\re{ss}) holds by 
(\re{Vb})
for $t\in D$, where $D$ is a countable dense set. 
Finally, let us use the  interpolation inequality  and  (\re{Vb}), (\re{ecg}):
for any $\ve>0$
\begin{equation}\la{inti}
 d_\cX(X_m(t),X_m(s))\le\ve d_\cV(X_m(t),X_m(s))
+ C(\ve)  d_{\cV^{-1}}(X_m(t),X_m(s))\le 2\ve R+C(\ve,R)|t-s|.
\end{equation}
This inequality 
implies the 
equicontinuity of the Galerkin approximations with the values in $\cX$. Hence,
convergence 
(\re{ss}) holds for all $t\in\R$ since it holds for 
the dense set of $t\in D$. 
The same equicontinuity also implies  the continuity of the limit function $X\in C(\R, \cX)$.
\medskip\\
ii) 
Integrating  equation (\re{gali2}), we obtain 
\begin{equation}\la{gal2}
\int_0^t\langle \dot X_m(t),Y\rangle\,ds=-\int_0^t
\langle X_m(s), AY)\,ds + \int_0^t\langle N(X_m(s)),Y\rangle\,ds,
\qquad Y\in\cW_m,
\end{equation}
Below we will write $m$ instead of $m'$. 
To prove   (\re{LPSiv}) it suffices to check that  in the limit $m\to\infty$, we get
\begin{equation}\la{gal4}
\int_0^t\langle \dot X(t),Y\rangle\,ds
=-\int_0^t\langle X(s),AY\rangle\,ds + \int_0^t\langle N(X(s)),Y\rangle\,ds,
\qquad Y\in\cW_n,\qquad n\in\N.
\end{equation}
The convergence of  the left hand side and of the first term on the right hand side 
of (\re{gal2}) follow from (\re{ss}) and  (\re{s0}) since $ AY\in\cW_m$.

It remains to consider the last integral of (\re{gal2}).
The integrand is uniformly bounded by (\re{Vb}) and Lemma \re{p1}. 
Hence, it suffices to check the pointwise convergence
\begin{equation}\la{Nm}
\langle N(X_m(t), Y\rangle-\!\!\!-\!\!\!\!\to \langle N(X(t), Y\rangle,\quad m\to\infty,\qquad Y\in\cW_n
\end{equation}
 for any $t\in\R$. Here $N(X_m(t))=(ie\phi_m(t)\psi_m(t),p_m(t),f_m(t))$ according to 
the notations (\re{HN}), and $Y=(\vp,\vka,v)\in \cW_n$. Hence, (\re{Nm}) reads
 \begin{equation}\la{Nm2}
 ie(\phi_m(t)\psi_m(t),\vp)+p_m(t)\vka+f_m(t)v\,\to\,  ie(\phi(t)\psi(t),\vp)+p(t)\vka+f(t)v,\quad m\to\infty.
\end{equation}
 The convergence of $p_m(s)\vka$ follows from (\re{ss}) (with $m'=m$) .
 To prove the convergence of two remaining terms, we first show that
\begin{equation}\la{2rt}
\phi_m(t):=G\rho_m \toLC \phi(t):=G\rho,\quad m\to\infty.
\end{equation}
Indeed,  (\re{ss}) implies that 
\begin{equation}\la{qp}
\psi_m(t)\toLt\psi(t),\qquad q_m(t)\to q(t),\quad m\to\infty.
\end{equation}
	The sequence $\psi_m(t)$ is bounded in $H^1({\T})$ by (\re{Vb}).
	Hence, $\psi(t)\in H^1({\T})$ and 
	the sequence $\rho_m(t)$ is bounded in the Sobolev space $W^{1,3/2}({\T})$ by (\re{np}).
	Therefore, the sequence $\rho_m(t)$ is precompact in $L^2({\T})$ by  the Sobolev compactness theorem. 
	Hence,
	\begin{equation}\la{qp2}
	\rho_m\toLt\rho,\quad m\to\infty
	\end{equation}
	by (\re{qp}).
	Therefore, (\re{2rt}) holds since the operator $G:L^2({\T})\to C({\T})$ is continuous. 
	From (\re{2rt}) and (\re{qp}) it follows  that
	\begin{equation}\la{Nm3}
	\phi_m(t)\psi_m(t)\toLt\5 \phi(t)\psi(t),\quad f_m(t)\to f(t),\quad m\to\infty,
	\end{equation}
  which proves (\re{Nm2}). Now (\re{gal4}) is proved for $Y\in\cV_n$ with any $n\in\N$.
 Hence, $X(t)$ is a solution to (\re{HS}).  Finally,  
	$\Vert N(X(\cdot))\Vert_\cW$ is a bounded function  by
  (\re{Vb}) and Lemma \re{p1}.
  Hence, 
 (\re{LPSiv}) implies that  $X(\cdot)\in C(\R,\cV)$.
 \medskip\\
 iii) The conservation laws (\re{EQ2}) and the convergences (\re{s0}), (\re{ss}) imply that
 \begin{equation}\la{EQ3}
E(X(t))\le E(X(0)),\quad Q(X(t))\le Q(X(0)),\qquad t\in\R.
\end{equation}
The last inequality holds by the first convergence of (\re{qp}). The first inequality follows from the representation
\begin{equation}\la{EQ4}
E(X_m(t))=\fr12 \Vert\na\psi_m(t)\Vert_{L^2({\T})}^2+\fr12 
\Vert \sqrt{G}\rho_m(t)\Vert_{L^2({\T})}^2 +\sum_{n\in\Gamma_n}\fr{p_m^2(n,t)}{2M}.
\end{equation}
Namely, the last two terms on the right hand side converge by (\re{qp2}) and (\re{ss}). Moreover, 
the first term is bounded by (\re{Vb}). Hence,  the first convergence of  (\re{qp}) 
implies the weak convergence 
 \begin{equation}\la{EQ5}
\na\psi_{m}(t)\toLwt\na\psi(t)
 \end{equation}
 by the Banach theorem.
Now the first inequality of (\re{EQ3}) follows by the property of the weak convergence in the Hilbert space.
Finally, the opposite inequalities to (\re{EQ3}) are also true by the uniqueness of 
solutions $X(\cdot)\in C(\R,\cV)$, which is proved in Proposition \re{TLWP}.
\end{proof}

\setcounter{equation}{0}
\section{Jellium ground states}
We describe all 
 solutions to (\re{LPS1})--(\re{LPS3}) with minimal energy  (\re{Hfor}),
give some examples of ion densities illustrating the Jellium and 
the Wiener conditions, 
and show the existence of non-periodic ground states.

\subsection{Description of all jellium ground states}

The following lemma gives the description of all ground states
of the system (\re{LPS1})--(\re{LPS3}).  
\bl\la{Jgs}
Let the  Jellium condition (\re{Wai}) hold. Then 
all solutions 
to (\re{LPS1})--(\re{LPS3})
of minimal (zero) energy are $(e^{i\al}\sqrt{Z},q^*,0)$ 
with 
$\al\in[0,2\pi]$ and
$q^*\in\T^{\ov N}$ satisfying 
the identity
\begin{equation}\la{sipiq}
	\sum_{n\in{\Ga}}\si(x-q^*(n))\equiv eZ,\qquad x\in{\T},
\end{equation}

\el
\begin{proof}
First,
let us note that the ${\Ga}$-periodic solutions (\re{gr}) 
have the zero energy, and 
the identity (\re{sipiq}) holds for $q^*=\ov r$ by  (\re{sipi}).

Further, for any solution with zero energy  (\re{Hfor})
all summands on the right hand side of  (\re{Hfor}) vanish.
The first integral  vanishes only
for constant functions. Hence, the normalization condition (\re{rQ}) gives
  \begin{equation}\la{ppo}
  \psi(x,t)\equiv
\psi_{\al(t)}(x)\equiv e^{i\al(t)}\sqrt{Z},\qquad\al(t)\in \R.
\end{equation}
Then 
\begin{equation}\label{roZ2}
 \rho^e(x,t):=-e|\psi_{\al(t)}(x)|^2\equiv -eZ,\qquad x\in {\T},\,\,\,\,t\in\R,
 \end{equation}
similarly to (\re{roZ}).
Further,
the second summand of (\re{Hfor}) vanishes only for $\rho(x,t)\equiv 0$. Hence, $\rho^i(x,t)\equiv eZ$
that is equivalent to (\ref{sipiq}) 
with $q(n,t)$ instead of $q^*(n)$
by (\re{roZ2}).
However, $M\pa_t q(n,t)=p(n,t)\equiv 0$  for the zero energy  (\re{Hfor}).
Hence,
\begin{equation}\la{qnc}
q(n,t)\equiv q^*(n),\qquad t\in\R,
\end{equation}
where $q^*$ satisfies  (\re{sipiq}).
Moreover, 
the Poisson equation (\re{LPS2}) with $\rho(x,t)\equiv 0$
implies that
$\phi(x,t)\equiv 0$ 
 after a gauge transformation (\ref{GT}).
Hence,   finally,
substituting  (\re{ppo}) into  (\re{LPS1}) 
with $\phi(x,t)\equiv 0$,
we obtain that $\al(t)\equiv{\rm const}$.
 \end{proof}

This lemma implies that all  ${\Ga}$-periodic ground states are given by
(\re{gr}).

 
 \subsection{Jellium and Wiener conditions. Examples}\la{sex}

The Wiener condition (\re{W1}) for the ground states (\re{gr}) holds 
under the generic assumption
(\re{W1s}).
On the other hand,  (\re{W1}) does not hold
  for the simplest Jellium model,
 when $\si(x)$ is the characteristic function
\begin{equation}\la{sic}
\si(x)=\si_1(x):=\left\{
\begin{array}{ll}
eZ,& x\in\Pi\\ 
0,& x\in{\T}\setminus\Pi
\end{array}\right|,
\end{equation}
where $\Pi:=[-1/2,1/2]^3$.
 Indeed, 
in this case the Fourier transform 
\begin{equation}\la{sJM}
\hat\sigma_1(\xi)=eZ\hat\chi_1(\xi_1)\hat\chi_1(\xi_2)\hat\chi_1(\xi_3);\qquad
\hat\chi_1(s)=\fr {2\sin s/2}s,\quad s\in\R\setminus 0,
\end{equation}
where 
$\chi_1(s)$ is the characteristic function of the interval 
 $[-1/2,1/2]$.
In this case we have
for $\theta= (0,\theta_2,\theta_3)$,
\begin{equation}\la{DK2}
\Si(\theta)= \sum_{m\in\Z^3:\,m_1=0}\Big[
 \fr{\xi\otimes\xi}{|\xi|^2}|\hat\si(\xi)|^2\Big]_{\xi=\theta+2\pi m},
\end{equation}
which is a degenerate matrix since $\xi_1=0$ in each summand. Hence, (\re{W1}) fails.
Similarly, the Wiener condition  fails for 
$\si_k(x)=eZ\chi_k(x_1)\chi_k(x_2)\chi_k(x_3)$ where 
$\chi_k=\chi_1*...*\chi_1$ ($k$ times)
 with $k=2,3,...$,
since in this case
\begin{equation}\la{sJM2}
\hat\sigma_k(\xi)=eZ\hat\chi_k(\xi_1)\hat\chi_k(\xi_2)\hat\chi_k(\xi_3);\qquad
\hat\chi_k(s)=\Big[\fr {2\sin s/2}s\Big]^k,\quad s\in\R\setminus 0.
\end{equation}

 
 \subsection{Non-periodic ground states}\la{snp}

It is easy to construct ground states   which are not $\Gamma_1$-periodic
in the case  of characteristic function (\re{sic}). Namely, the identity (\re{sipiq}) obviously holds 
for periodic arrangement  of ions (\re{gr2}). Now let us modify this periodic arrangement as follows:
\begin{equation}\la{mod}
q^*(n)=
(r_1,r_2,r_3+\tau(n_1,n_2)),\qquad n\in{\Ga},
\end{equation}
where $\tau(n_1,n_2)$ is an arbitrary point of
the circle
$\R/N\Z$. Now (\re{sipiq}) obviously holds for any arrangement of ions (\re{mod}).

Next lemma gives a more general spectral 
assumptions on $\si$ which provide ground states
with non-periodic ion arrangements. For example, let us assume that
\begin{equation}\la{spc}
\si(\xi)=0,\quad \xi_3\in 2\pi\Z\setminus 0, \qquad{\rm and}
\qquad\si(\xi_1,\xi_2,0)=0,\quad (\xi_1,\xi_2)\in 2\pi\Z^2\setminus 0.
\end{equation}
In particular, this holds for the densities (\re{sJM2}) with all $k=1,...$

\bl\la{lnonp}
Let $\si$ satisfy the spectral condition (\re{spc}).
Then there exist  ground states  which are not $\Gamma_1$-periodic.

\el
\begin{proof}
In the Fourier transform (\re{sipiq}) reads
\begin{equation}\la{sip4}
\int_{\T} e^{i\xi x}\sum_{n\in{\Ga}}\si(x-n-q^*(n))dx=
\hat\si(\xi)\sum_{n\in{\Ga}} e^{i\xi(n+q^*(n))}=
\left\{
\begin{array}{rl}
eZ\ov N,&\xi=0\\
0,&\xi\in{\Gamma^*_N}\setminus 0.
\end{array}\right.
\end{equation}
These identities hold for any density $\si$ satisfying (\re{Wai})
if 
\begin{equation}\la{sip5}
\sum_{n\in{\Ga}} e^{i\xi(n+q^*(n))}=0,\qquad \xi\in{\Gamma^*_N}\setminus{\Gamma^*_1}.
\end{equation}
In particular, $q^*=\ov r$
satisfies the system
(\re{sip5})  since then
\begin{equation}\la{sip6}
\sum_{n\in{\Ga}} e^{i\xi(n+q^*(n))}=\sum_{n\in{\Ga}} e^{i\xi(n+r)}=
e^{i\xi r}
\sum_{n\in{\Ga}} e^{i\xi n}=0,\qquad \xi\in{\Gamma^*_N}\setminus{\Gamma^*_1}.
\end{equation}
Indeed, 
\begin{equation}\la{sip7}
\sum_{n\in{\Ga}} e^{i\xi n}=\sum_{n\in{\Ga}} e^{i(\xi_1 n_1+\xi_2 n_2+\xi_3 n_3)}
=\sum_{n_1=0}^{ N-1} e^{i\xi_1 n_1}
\sum_{n_2=0}^{ N-1} e^{i\xi_2 n_2}
\sum_{n_3=0}^{ N-1} e^{i\xi_3 n_3}
=0
\end{equation}
since at least one $\xi_k\not\in 2\pi\Z$ for $\xi\in{\Gamma^*_N}\setminus{\Gamma^*_1}$.
Now we modify $\ov r$ as follows:
\begin{equation}\la{modr}
q^*(n):=( a_1(n_1,n_2)),a_2(n_1,n_2), r_3),\qquad n\in{\Ga},
\end{equation}
where $a_1(n_1,n_2))$ and $a_2(n_1,n_2)$ are {\it arbitrary points} of the circle $\R/N\Z$.
Obviously, $q^*(n)$ can be non-periodic in $n_1$ and $n_2$.

On the other hand, (\ref{sip4}) holds. Indeed, for $\xi_3\in2\pi\Z$
this follows from (\ref{spc}), while for $\xi_3\not\in2\pi\Z$
we have
\begin{equation}\la{sip62}
\sum_{n\in{\Ga}} e^{i\xi(n+q^*(n))}=\sum_{n_1,n_2=0}^{N-1}
e^{i(\xi_1(n_1+a_1(n_1,n_2)+\xi_2(n_2+a_2(n_1,n_2)}
\sum_{n_3=0}^{N-1} e^{i\xi_3 n_3}=0.
\end{equation}
\end{proof}

\subsection{On the problem of periodicity}\la{spp}
In the non-periodic examples above 
the Wiener condition (\ref{W1}) breaks down for the case
 $\theta_3=0$ by (\ref{spc}). We suppose that the Wiener condition provides
only periodic ground states, however this is an open problem. 
For densities $\si$ satisfying a more strong condition (\re{W1s}), the  
identity (\re{sipiq}) is equivalent  to the system (\re{sip5}) by (\re{sip4}).
The system 
 (\re{sip5})
can be written as an `algebraic system'
\begin{equation}
\sum_{n\in{\Ga}} w_1^{m_1}(n) w_2^{m_2}(n) w_3^{m_3}(n)=0,\qquad m=(m_1,m_2,m_3)\in\Z^3\setminus N\Z^3
\end{equation}
for 
\begin{equation}
w_j(n_1,n_2,n_3):=e^{i\ds\fr{2\pi}N [n_j+q^*_j(n)]},
\qquad n\in{\Ga},\quad j=1,2,3.
\end{equation}
The ${\Ga}$-periodicity of $q^*$ 
is equivalent to the fact that
 only solutions are 
\begin{equation}
w_j(n_1,n_2,n_3)=C_j \lam^{n_j},\qquad j=1,2,3,
\end{equation}
where 
$\lam:=e^{i\ds\fr{2\pi}N}$.

\br
{\rm
For the corresponding 1D analog
\begin{equation}
\sum_{n_1=0}^{N-1} w_1^{m_1}(n_1)=0,\qquad m=1,...,N-1,
\end{equation}
the only solutions are $w_1(n_1):=C_1 \lam^{n_1}$ that follows easily 
from the 
Newton-Girard formulas guessed by Girard in 1629 
and rediscovered (without a proof) by Newton in 1666. The formulas were 
proved by Euler in 1747, see \ci{Euler1747,T2001}.
}
\er

\setcounter{equation}{0}
\section{The orbital stability of the ground state}
In this section  we expand the energy into the Taylor series and prove the orbital stability 
checking the  positivity of the energy Hessian.
\subsection{The Taylor expansion of the Hamiltonian functional}
\la{sT}

We will deduce the lower estimate (\re{BLi}) using
 the  Taylor expansion of $E(S+Y)$ for 
$S=S_{\al,r}:=(\psi^\al,\ov r,0)\in \cS$ and $Y=(\vp,\vka,\pi)\in\cW$:
\begin{equation}\la{te} 
E(S+Y)=E(S)+\langle E'(S),Y\rangle +\fr12 \langle Y,E''(S)Y\rangle + R(S,Y)=
\fr12 \langle Y,E''(S)Y\rangle + R(S,Y)
\end{equation}
since $E(S)=0$ and $E'(S)=0$.
First,
we expand 
 the  charge density (\re{Hfor2})
corresponding to 
$S+Y=(\psi^\al+\vp,\ov r+\vka,\pi)$:
\begin{equation}\la{ro}
\rho(x)=\rho^{(0)}(x)+\rho^{(1)}(x)+\rho^{(2)}(x),\qquad x\in{\T},
\end{equation}
where $\rho^{(0)}$ and  $\rho^{(1)}$ are respectively 
the  terms of zero and first  order in 
$Y$, while  $\rho^{(2)}$ is the remainder.
However, $\rho^{(0)}(x)$ is the total charge density of the ground state
which is identically zero by (\ref{sipi}) and (\re{roZ}):
\begin{equation}\la{ro0}
\rho^{(0)}(x)=\rho^i_0(x)-e|\psi^\al(x)|^2\equiv 0,\qquad x\in{\T}.
\end{equation}
Thus, $\rho=\rho^{(1)}+\rho^{(2)}$.
Expanding (\re{Hfor2}), we obtain 
\begin{eqnarray}
\la{ro11}\rho^{(1)}(x)\!\!\!\!&\!\!\!\!=\!\!\!\!&\!\!
\si^{(1)}(x)-2e\psi^\al (x)\cdot\vp(x),\quad \si^{(1)}(x)=-\sum_{n\in{\Ga}} \vka(n)\cdot\na\si(x-n-r),
\smallskip\\
\la{ro13}\rho^{(2)}(x)\!\!\!\!&\!\!=\!\!\!\!&\!\!\si^{(2)}(x)-e|\vp(x)|^2,~~ \si^{(2)}(x)=\fr12\sum_{n\in{\Ga}}
\int_0^1\!\!(1\!-\!s)
[\vka(n)\cdot\na]^2\si(x-n-r-s\vka(n))ds,
\end{eqnarray}
where $\cdot$ means the inner product of real vectors, and in particular,
$\psi^\al (x)\cdot\vp(x)$ is the inner product of the corresponding vectors in $\R^2$.
Substituting  $\psi=\psi^\al +\vp$ and $\rho=\rho^{(1)}+\rho^{(2)}$  
into (\ref{Hfor}), we obtain that
the  quadratic part  of (\re{te}) reads 
\begin{equation}\la{B2}
\fr12\langle Y,E''(S) Y \rangle=\fr12\int_{{\T}}|\na \vp(x)|^2dx+
\fr12 (\rho^{(1)},G\rho^{(1)})+K(\pi),~~~~~~ K(\pi):=\ds\sum_n\fr{\pi^2(n)}{2M}~~
\end{equation}
and the remainder equals
\begin{equation}\la{B3}
R(S,Y)=\fr 12(2\rho^{(1)}+\rho^{(2)},G\rho^{(2)}). 
\end{equation}

\subsection{The null space of the energy Hessian}

In this section we calculate the null space 
 \begin{equation}\la{KYd}
  \cK(S):=\Ker\, E''(S)\Big|_\cW,\qquad S\in \cS
  \end{equation}
 under the Wiener condition.

  \bl
  Let the Jellium and the Wiener conditions (\re{Wai}),  
(\re{W1}) hold and $S\in\cS$. Then
  \begin{equation}\la{KY}
  \cK(S)=\{(C,\ov s,0):~~ C\in\C,~~s\in\R^3 \},
  \end{equation}
  where $\ov s\in\R^{3\ov N}$ is defined similarly to (\re{gr2}): 
  $\ov s(n)\equiv s$.
  \el
  \begin{proof}
  All summands  of
the
energy (\re{B2}) are nonnegative. Hence,  this expression is zero if and only if
 all the summands vanish:
\begin{equation}\la{rb}
\vp(x)\equiv C,\quad
(\rho^{(1)}, G\rho^{(1)})=\Vert\sqrt{G}[\si^{(1)}-2e\psi^\al \cdot\vp]\Vert_{L^2({\T})}^2
=0,\quad \pi=0.
\end{equation}
Note that $\sqrt{G}\psi^\al \cdot\vp=\sqrt{G}\psi^\al \cdot C=0$ since 
the operator $G$ annihilates the constant functions 
by (\re{fs}). Hence, (\re{rb}) implies that
\begin{equation}\la{rb2}
\sqrt{G}\si^{(1)}=0.
\end{equation}
On the other hand, (\re{ro11}) gives
in the Fourier transform 
\begin{equation}\la{B314}
\hat\si^{(1)}(\xi)=\hat\si(\xi)\xi\cdot\sum_{n\in{\Ga}} ie^{i\xi (n+r)}\vka(n)
=i\hat\si(\xi)\xi\cdot e^{i\xi r}\hat \vka(\xi),\qquad\xi\in {\Gamma^*_N},
\end{equation}
where $\hat \vka(\xi):=\sum_{n\in{\Ga}} ie^{i\xi n}\vka(n)$ is a $2\pi\Z^3$-periodic function on
${\Gamma^*_N}$.
Hence, definition (\re{fs}) and the Jellium condition (\re{Wai}) imply that
\begin{eqnarray}\la{B315}
0=\Vert\sqrt{G}\si^{(1)}\Vert_{L^2({\T})}^2&=&
 N^{-3}\sum_{{\Gamma^*_N}\setminus {\Gamma^*_1}} |\hat\si(\xi)\fr{\xi\hat \vka(\xi)}{|\xi|}|^2
\nonumber\\
\nonumber\\
&=&N^{-3}\sum_{\theta\in\Pi^*_N\setminus {\Gamma^*_1}} 
\langle\hat \vka(\theta),
 \sum_{m\in\Z^3}\Big[\fr{\xi\otimes\xi}{|\xi|^2}|\hat\si(\xi)|^2\Big]_{\xi=\theta+2\pi m}\hat \vka(\theta)\rangle
 \nonumber\\
\nonumber\\
&=&N^{-3}\sum_{\theta\in\Pi^*_N\setminus {\Gamma^*_1}} 
\langle\hat \vka(\theta),
 \Si(\theta)
 \hat \vka(\theta)\rangle.
\end{eqnarray}
As a result, 
\begin{equation}\la{B316}
\hat \vka(\theta)=0,\qquad \theta\in \Pi^*_N\setminus {\Gamma^*_1}
\end{equation}
by the Wiener condition
(\re{W1}).   
On the other hand, $\hat \vka(0)\in\R^3$ remains arbitrary.
	Respectively, 
	$\vka=\ov s$ with 
	an arbitrary $s\in\R^3$.
\end{proof}
\br
{\rm 
The key point of the proof is 
the explicit calculation (\re{B314}) in the Fourier transform.
This calculation relies on the invariance of
the Hessian $E''(S)$ with respect to ${\Ga}$-translations which is 
due to the periodicity of the ions arrangement of the ground state.
}
\er
\br ({\it Beyond the Wiener condition.})
{\rm
 If
the Wiener condition (\re{W1})  fails, the dimension of the space
\begin{equation}\la{V}
V:=\{v\in \R^{3\ov N}:~~ v(n)=\sum_{\theta\in\Pi^*_N\setminus{\Gamma^*_1}}
e^{-i\theta n} \hat v(\theta),
\qquad \hat v(\theta)\in\C^3,~~ \Si(\theta)\hat v(\theta)  \equiv\,  0\}
\end{equation}
is positive.
The above calculations show that in this case 
 \begin{equation}\la{KYg}
  \cK(S)=\{(C, \ov s+v,0):~~ C\in\C,~~s\in{\T},~~v\in V \}.
  \end{equation}
 The subspace $V\subset \R^{3\ov N}$ is orthogonal to the $3D$ subspace 
$\{\ov s:s\in\R^3\}\subset \R^{3\ov N}$ by the Parseval theorem.
Hence,  $\dim \cK(S)=5+d$,  where $d:=\dim V>0$. 
Thus, $\dim \cK(S)>5$.
Under the Wiener condition $V=0$, and (\re{KYg}) coincides with (\re{KY}).
}
 \er

\subsection{The positivity  of the Hessian}\la{spoH}
Denote by $N_S\cS$ the normal subspace to $\cS$ at a point $S$:
\begin{equation}\la{L0N}
N_S\cS:=\{Y\in\cW: \langle Y,\tau\rangle=0,~~\tau\in T_S\cS\},
\end{equation}
where $T_S\cS$ is the
tangent space to $\cS$ at the point $S$  and $\langle\cdot,\cdot\rangle$
stands for the scalar product (\re{dWW}).

\bd\la{dM}
Denote by $\cM$ the Hilbert manifold 
\begin{equation}\la{cM}
\cM:=\{X\in\cV: Q(X)=Z N^3\}.
\end{equation}
\ed
Obviously, $\cS\subset\cM$, and a tangent space to $\cM$ at a point
$S=(\psi^\al,\ov r, 0)$
is given by 
 \begin{equation}\la{TSM}
T_S\cM=\{(\vp, \vka,\pi)\in\cW:\vp\bot\psi^\al \},
\end{equation}
since $DQ(\psi^\al,\ov r,0)=(\psi^\al,0,0)$.
\bl
Let the Jellium condition (\re{Wai}) hold and $S=S_{\al ,r}\in \cS$. Then
the Wiener condition (\re{W1}) is necessary and sufficient for the positivity
of the Hessian $E''(S)$
in the orthogonal directions to $\cS$ on $\cM$, i.e., 
\begin{equation}\la{cM2}
E''(S)\Big|_{N_S\cS\cap T_S\cM}>0.
\end{equation}
\el
\begin{proof}
i) Sufficiency.
Differentiating $S_{\al,r}=(e^{i\al}\psi_0,\ov r,0)\in \cS$ in the parameters $\al \in [0,2\pi]$ and $r\in{\T}$, 
we obtain 
\begin{equation}\la{tv}
T_S\cS=\{(iC\psi^\al ,\ov s,0): ~~C\in\R,~~s\in\R^3\}.
\end{equation}
Hence, (\re{KY})  implies that 
\begin{equation}\la{tv2}
K(S):=\cK(S)\cap N_S\cS=
\{(C\psi^\al ,0,0): ~~C\in\R\}
\end{equation}
by (\re{1i}) and (\re{dWW}).
Therefore,
	\begin{equation}\la{tv22}
	\cK(S)\cap  N_S\cS\cap  T_S\cM=K(S)\cap  T_S\cM=
	(0 ,0,0),
	\end{equation}
since the vector $(\psi^\al ,0,0)$
 is orthogonal to $T_S\cM$ by (\re{TSM}). Now (\re{cM2}) follows since 
$E''(S)\ge 0$ by (\re{B2}).
\medskip\\
ii)	Necessity.
If the Wiener condition (\re{W1}) fails,
the null space $\cK(S)$ is given by (\re{KYg}). Hence, 
(\re{tv}) implies that now
\begin{equation}\la{tv3}
K(S)=
\{(C\psi^\al ,v,0): ~~C\in\R, ~~ v\in V\}.
\end{equation}
However,  $(\psi^\al,\psi^\al)>0$.
Hence, (\re{TSM}) implies that $(\psi^\al ,v,0)\not\in T_S\cM$ 
and
the intersection
\begin{equation}\la{tv4}
K(S)\cap T_S \cM=
\{0,v,0): ~~ v\in V\}
\end{equation}
is the nontrivial subspace of the dimension  $d>0$. 
Thus, the Hessian $E''(S)$ vanishes on this nontrivial subspace
of  $N_S\cS\cap  T_S\cM$.
\end{proof}

\br\la{rS} {\rm The positivity of type (\re{cM2}) breaks down for the 
submanifold 
$\cS(r):=\{S_{\al ,r}:\al \in[0,2\pi]\}$ 
with a fixed $r\in{\T}$ instead of  the solitary manifold $\cS$.
Indeed, in this case the corresponding tangent space  is smaller, 
\begin{equation}\la{tvr}
T_S\cS(r)=\{(iC\psi^\al ,0,0): ~~C\in\R\},
\end{equation}
and hence, the normal subspace $N_S\cS(r)$ is larger, containing all 
vectors $(0,\ov s,0)$
generating  the shifts of the torus. However, all these vectors also 
belong 
to the null space (\re{KY}) and to $T_S\cM$.
Respectively, the null space of the Hessian $E''(S)$ in $T_S\cM\cap N_S\cS(r)$ 
is  3-dimensional.
}
\er



\subsection{The orbital stability}
Here we prove our   main result.
\bt\la{tm}
Let the conditions (\re{Wai}),  (\re{W1}) and  (\re{ro+}) hold, 
and $\cS$ is the solitary 
manifold (\re{cS}). 
Then
for any 
$\ve>0$  there exists $\de=\de(\ve)>0$ such that for 
$X(0)\in\cM$ with
$d_\cV(X(0),\cS)<\de$ we have
\begin{equation}\la{m}
d_\cV(X(t),\cS)<\ve,\qquad t\in\R
\end{equation}
for the corresponding solution
$X(t)\in C(\R,\cV)$  to (\re{LPS1})--(\re{LPS3}).
\et
For the proof is suffices to check
the lower energy estimate (\re{BLi}): 
\begin{equation}\la{BL}
E(X)\ge \nu\,d^2_\cV(X,\cS)\quad{\rm if}\quad d_\cV(X,\cS)\le\de,\quad X\in\cM
\end{equation}
with some $\nu,\de>0$. 
This estimate implies Theorem \re{tm}, since the energy is conserved
along all trajectories.
First, we prove similar lower bound for the energy Hessian. 
\bl
Let all conditions of Theorem \re{tm}  hold. Then  for each $S\in \cS$
\begin{equation}\la{L02}
 \langle Y, E''(S)Y\rangle > \nu\Vert Y\Vert_\cW^2,  \qquad Y\in  N_S\cS \cap T_S\cM,
\end{equation}
where $\nu>0$.
\el
\begin{proof}
It suffices to prove (\ref{L02}) for $S=(\psi_0,0,0)$. Note that $E''(S)$ is not complex linear due to the integral 
in (\re{Hfor}). Hence, we express the action of $E''(S)$ in $\vp_1(x):=\rRe\vp(x)$ and $\vp_1(x):=\rIm\vp(x)$:
formulas (\ref{ro11}) and (\ref{B2}) imply that 
\begin{equation}\la{E''}
E''(S)Y=\left(\begin{array}{cccl}
 2H_0+4e^2\psi_0 G\psi_0 & 0 & 2L & 0
 \medskip\\
 0 & 2H_0  &0 & 0\medskip\\
 2L^{\5*}  &    0  &   T    & 0  \\
      0      &    0            &   0    &  M^{-1} \\
\end{array}\right)Y
\qquad{\rm for}\quad
Y=\left(\begin{array}{c}\vp_1 \\ \vp_2 \\ \vka \\ \pi \end{array}\right),
\end{equation}
where $H_0=-\fr12\De$ ,
 the operator $L$  corresponds to the matrix
\begin{equation}\la{S}
 L(x,n):=e\psi_0(x)G\na\si(x-n):~~ ~~x\in\R^3,~n\in\Ga,
\end{equation}
 and  $T$ corresponds to the real matrix with the entries
\begin{equation}\la{T}
T(n-n'):=-\ds\langle  G\na\otimes\na\si(x-n'),  \si(x-n) \rangle,\quad n, n'\in\Ga.
\end{equation}

Thus, $E''(S)$  is the selfadjoint operator in $\cX$ with the discrete spectrum, and
  (\re{cM2}) implies that the minimal eigenvalue of
$E''(S)$  in the invariant space $N_S\cS\cap   T_S\cM$ is  positive. Therefore, (\re{L02}) follows.
\end{proof}
\medskip

The positivity (\re{L02}) implies
the lower energy estimate (\re{BL}),  since the 
higher-order terms in (\re{te}) 
are negligible by the following lemma.

\bl\la{lre}
Let $\si(x)$ satisfy  (\re{ro+}).
Then the  remainder (\re{B3}) admits the estimate
\begin{equation}\la{B31}
|R(S,Y)|\le C
\Vert Y\Vert_\cW^3\quad\,\,\,{\rm for}\quad  \,\,\,\Vert Y\Vert_\cW \le 1.
\end{equation}
\el
\begin{proof} 
It suffices to prove the 
estimates
\begin{equation}\la{B312}
\Vert\sqrt{G}\rho^{(1)}\Vert_{L^2({\T})}\le C_1\Vert Y\Vert_\cW,\quad  
\Vert\sqrt{G}\rho^{(2)}\Vert_{L^2({\T})}\le C_2\Vert Y\Vert_\cW^2
\quad   \,\,\,{\rm for}\,\,\, \quad    \Vert Y\Vert_\cW  \le 1.
\end{equation}
Then  (\re{B31}) will follow from (\re{B3}).
\medskip\\
i)
By (\re{ro11}) we have
for $Y=(\vp,\vka,\pi)$ 
\begin{equation}\la{B313}
\sqrt{G}\rho^{(1)}=\sqrt{G}\si^{(1)}-2e\sqrt{G}\psi^\al (x)\cdot\vp(x).
\end{equation}
The operator $\sqrt{G}$ is bounded in $L^2(\R^3)$ by the definition (\re{fs}). 
Hence,  
\begin{equation}\la{GP1}
\Vert\sqrt{G}\si^{(1)}\Vert_{L^2({\T})}\le C |\vka|
\end{equation}
by (\re{ro11}). Applying to the second term the 
Cauchy-Schwarz and Hausdorff-Young inequalities,
we obtain 
\begin{equation}\la{GP2}
\Vert\sqrt{G}\psi^\al (x)\cdot\vp\Vert_{L^2({\T})}\le 
C\Big[\sum_{\xi\in{\Gamma^*_N}}\fr{|\hat\vp(\xi)|^2}{|\xi|^2}\Big]^{1/2}\le 
C\Vert\hat\vp\Vert_{L^4({\Gamma^*_N})}\Big[\sum_{\xi\in{\Gamma^*_N}}|\xi|^{-4}\Big]^{1/2}
\le
C\Vert\vp\Vert_{L^{4/3}({\T})}.
\end{equation}
Hence, the first inequality (\re{B312}) is proved.
\medskip\\
ii) Now we prove the second 
 inequality  (\re{B312}). By (\re{ro13}) we have
for $Y=(\vp,\vka,\pi)$ 
\begin{equation}\la{B317}
\sqrt{G}\rho^{(2)}(x)=\sqrt{G}\si^{(2)}(x)-e\sqrt{G}|\vp(x)|^2.
\end{equation}
Similarly to  (\re{GP1})
\begin{equation}\la{GP3}
\Vert\sqrt{G}\si^{(2)}\Vert_{L^2({\T})}\le C|\vka|^2.
\end{equation}
 Finally, 
denoting $\beta(x):=|\vp(x)|^2$, we obtain similarly to (\re{GP2}) 
\begin{eqnarray}\la{fp}
\Vert\sqrt{G}|\vp(x)|^2\Vert_{L^2({\T})}
&\le&
C\Big[\sum_{\xi\in{\Gamma^*_N}}\fr{|\hat \beta(\xi)|^2}{|\xi|^2}\Big]^{1/2}\le 
C\Vert \hat\beta\Vert_{L^4({\Gamma^*_N})}\Big[\sum_{\xi\in{\Gamma^*_N}}|\xi|^{-4}\Big]^{1/2}
\nonumber\\
&\le&
C_1\Vert \beta\Vert_{L^{4/3}({\T})}
=
C_1\Vert \vp\Vert_{L^{8/3}({\T})}^2
\le C_2\Vert\vp\Vert_{H^1({\T})}^2
\end{eqnarray}
by the Sobolev embedding theorem \ci{Adams}.
Now the lemma is proved.
\end{proof}



\chapter{N-particle Schr\"odinger theory}

\centerline{Abstract}
\medskip

In this chapter, we
 consider the coupled
Schr\"odinger--Poisson--Newton equations for finite  crystals
with moving ions and 
 antisymmetric $N$-particle
wave functions. This antisymmetry expresses the Pauli exclusion principle. 
The dynamics of electron field is described by  the many-particle 
Schr\"odinger equation.

We
construct  global dynamics,
prove the conservation of energy and charge, and
  give the description of all ground states.
Our main result
is  the  orbital stability of every 
 ground state with periodic arrangement of ions
under the same `Jellium' and  Wiener-type  conditions on the ion charge density (\ref{Wai}), (\ref{W1}).
The Pauli exclusion principle
plays the key role in the proof of stability, see Remark \ref{rPa}.

\setcounter{equation}{0}
\section{Introduction}

As in Chapter I, we consider crystals which occupy the finite torus $\T:=\R^3/N\Z^3$
and have one ion per cell of the cubic lattice $\Ga:=\Z^3/N\Z^3$, where $N\in\N$ and $N>1$. We also assume the conditions (\ref{ro+}).
Now we denote 
\be\la{FvN}
\ov\T:=
\T^{\ov N}
:=\{\ov x=(x_1,...,x_{\ov N}): x_j\in\T,\quad j=1,...,\ov N\},\qquad \ov N:=N^3.
\ee
\bd $\cF$ is the `fermionic' Hilbert space of  complex antisymmetric functions
 $\psi(x_1,...,x_{\ov N})$ on $\ov\T$
with the norm
\be\la{FN}
\Vert\psi\Vert_\cF^2:= \Vert\na^\otimes\psi\Vert_{L^2(\ov\T)}^2+\Vert\psi\Vert_{L^2(\ov\T)}^2,
\ee
where $\na^\otimes$ denotes the gradient with respect to $\ov x\in\ov\T$.
\ed
Let 
$\psi(\cdot,t)\in\cF$ for $t\in\R$  be the antisymmetric  
wave function of the fermionic electron field, 
$q(n,t)$ denotes the ion displacement  from the reference position $n\in\Ga$,
and
$\phi(x,t)$ be the electrostatic  potential generated by the ions and electrons.
We assume $\hbar=c=\cm=1$, where $c$ is the speed of light and $\cm$ is the electron mass.
Let us denote the `second quantized' operators on $\cF$,
\be\la{2De}
\De^\otimes:=\sum_{j=1}^{\ov N} \De_{x_j};\qquad \phi^\otimes(\ov x,t):=\sum_{j=1}^{\ov N} \phi(x_j,t).
\ee
The coupled Schr\"odinger-Poisson-Newton equations  read as follows
\beqn\la{LPS1N}
i\dot\psi(\ov x,t)\!\!&=&\!\!-\fr12\De^\otimes\psi(\ov x,t)-e\phi^\otimes(\ov x,t)\psi(\ov x,t),\qquad \ov x\in\ov\T,
\\
\nonumber\\
-\De\phi(x,t)\!\!&=&\!\!\rho(x,t):=\sum_{n\in\Ga}
\sigma(x-n-q(n,t))+\rho^e(x,t),\qquad x\in\T,
\la{LPS2N}
\\
\nonumber\\
M\ddot q(n,t)
\!\!&=&\!\!-(\nab\phi(x,t),\sigma(x-n-q(n,t))), 
\qquad n\in\Ga.
\la{LPS3N}
\eeqn
Here
the 
brackets $(\cdot,\cdot)$
 stand for the  scalar product on the real Hilbert
space $L^2(\T)$ and for its different extensions,   
$M>0$ is the mass of one ion,
and the electronic charge density is defined by
\be\la{reN}
\rho^e(x,t)
:=-e\int_{\ov\T}\sum_{j=1}^{\ov N} \de(x-y_j)|\psi(\ov y,t)|^2\,d\ov y,~~ x\in\T.
\ee
 Similar finite periodic approximations of crystals are treated in all textbooks on 
quantum theory of solid state
\ci{Born, Kit, Zim}. 
However,   the  stability 
of ground states in
this model was newer discussed.
 \medskip
 
The normalization (\re{rQ}) now reads,
\be\la{rQN}
\Vert\psi(\cdot,t)\Vert_{L^2(\ov\T)}^2=Z,\qquad t\in\R.
\ee
Hence, as in Chapter I, the potential $\phi(x,t)$ can be eliminated
from  the system (\re{LPS1N})--(\re{LPS3N})
using 
the operator $G:=(-\De)^{-1}$ defined in (\ref{fs}).
Substituting $\phi(\cdot,t)=G\rho(\cdot,t)$
into the remaining equations (\ref{LPS1N}) and (\ref{LPS3N}) we can write these equations as
\be\la{vfN}
\dot X(t)=F(X(t)),\qquad t\in\R,
\ee
where $X(t)=(\psi(\cdot,t), q(\cdot,t), p(\cdot,t))$ with $p(\cdot,t):=\dot q(\cdot,t)$.
We will identify  complex  functions $\psi(\ov x)$ with 
two real functions $\psi_1(\ov x):=\rRe\psi(\ov x) $ 
and  $\psi_2(\ov x):=\rIm\psi(\ov x)$.
With these notations,
equation (\ref{vfN}) 
is equivalent to the Hamiltonian system
\begin{equation}\la{HSiN}
\pa_t \psi_1(\ov x,t)=\fr12 \pa_{\psi_2(\ov x)}E,~~\pa_t \psi_2(\ov x,t)=-\fr12\pa_{\psi_1(\ov x)}E,~~
\pa_t q(n,t)= \pa_{p(n)}E,~~\pa_t p(n,t)=-\pa_{q(n)} E
\end{equation}
together with the normalisation condition (\ref{rQN}).
Now
the Hamiltonian functional (energy)
reads 
\be\la{HforN}
  E(\psi, q, p)=\fr12\int_{\ov\T} |\na^\otimes\5\psi(\ov x)|^2\,d\ov x+
  \fr12 (\rho,G\rho)+
  \sum_{n\in\Ga} \fr{p^2(n)}{2M}.
\ee
As in (\ref{Hfor2}),  the total charge density
 $\rho(x)$ is the sum of
the ion and electronic charge densities,
\be\la{Hfor2N}
\rho(x):=\rho^i(x)+\rho^e(x),\quad
\rho^i(x):=\sum\limits_{n\in\Ga}\si(x-n-q(n)),\qquad x\in\T.
\ee 
Similarly to notations of  Chapter I, we denote the Hilbert manifolds
\be\la{cMN}
 \cV:=\cF\otimes \ov\T\otimes \R^{3\ov N},\qquad
\cM:=\{X\in\cV: Q(X)=Z\}.
\ee
We prove 
the global well-posedness of the dynamics:
for any $X(0)\in\cM$
there exists a unique
solution 
$X(t)\in C(\R,\cV)$ 
to (\re{HSiN}), and
the energy and charge conservations  hold:
\be\la{EQN}
E(X(t))=E(X(0)),\quad Q(X(t))=Q(X(0)),\qquad t\in\R.
\ee

Our main goal is the stability of 
ground states, i.e.,
solutions to (\re{HSiN})
with minimal energy  (\re{HforN}) and 
with  $\Ga$-periodic arrangement of ions
under  the same 
 Jellium and the Wiener conditions (\re{Wai}) and  (\re{W1}) 
onto the  ion densities $\si(x)$. 
Nonperiodic arrangements 
exist for some degenerate 
densities $\si$ as in Section \ref{snp}.

\medskip

The energy (\re{HforN}) is nonnegative, and its minimum is zero.
We show  that under the Jellium condition 
 all
ground states  with $\Ga$-periodic arrangement of ions have the form 
\begin{equation}\la{grN}
S(t):=(\psi_0e^{-i\om_0t},\ov r,0), ~~~~~r\in\T.
\end{equation}
Here
\begin{equation}\la{gr2N}
 \ov r\in \ov\T:\quad \ov r(n)=r,\,\,n\in\Ga,
\end{equation}
while $\psi_0\in\cF$ is a normalised  eigenfunction 
\be\la{eig}
-\fr12 \De^\otimes\5\psi_0(\ov x)=\om_0\psi_0(\ov x),\qquad \ov x\in\ov\T;\qquad \Vert\psi_0\Vert_{L^2(\ov\T)}=1,
\ee
which corresponds to the minimal eigenvalue  
$
\om_0:=\min \5\5\Spec(-\fr12 \De^\otimes).
$
\medskip

We establish
the stability  
of the real  4-dimensional `solitary manifold'
\be\la{cSN}
\cS=\{S_{\al,r}=(\psi^\al,\ov r,0):~\psi^\al(\ov x)
\equiv e^{i\al}\psi_0(\ov x),~~\al\in [0,2\pi];~~ r\in\T\},
\ee
where $\psi_0$ is a fixed eigenfunction, satisfying  the additional restriction (\re{adrN}).
 The normalization (\re{rQN}) and the identity (\re{eig}) imply that
\be\la{ESN}
E(S)= \om_0Z,\qquad S\in\cS.
\ee

The main result of this chapter  is the following theorem.

\bt\la{tmN}
Let the Jellium and Wiener conditions  (\re{Wai}) and (\re{W1}) hold as well as  (\re{adrN}).
Then
for any 
$\ve>0$  there exists $\de=\de(\ve)>0$ such that for 
$X(0)\in\cM$ with
$d_\cV(X(0),\cS)<\de$ we have
\be\la{mN}
d_\cV(X(t),\cS)<\ve,\qquad t\in\R,
\ee
where  $X\in C(\R,{\cal V})$
is the corresponding solution
 to (\ref{HSiN}).

\et

In Chapter I,
we have proved  similar result for the system  
(\re{LPS1})--(\re{LPS3}) with the 
one-particle Schr\"odinger equation instead of (\re{LPS1N}).
The general plan of analysis in the 
present chapter is also similar. However, now
all estimates and the uniformity (\re{roZ})
require completely new arguments.
\medskip

Let us comment on our approach.
We prove the local well-posedness for the system  (\re{HSiN})
by
the contraction mapping principle.
The global  well-posedness we deduce
from the  energy conservation which follows
by the Galerkin approximations. 
\medskip

The orbital stability of the solitary manifold $\cS$ is deduced
from  the lower energy estimate
\be\la{BLiN}
E(X)-\om_0Z\ge \nu\,d^2(X,\cS)\qquad{\rm if}\qquad d(X,\cS)\le \de,\qquad X\in\cM,
\ee
where
$\nu,\de>0$ and `$d$' is the distance in the `energy norm'.
We deduce this estimate  from the positivity of the Hessian $E''(S)$
for $S\in \cS$ in the orthogonal directions to $\cS$ on the manifold $\cM$.
We show that the Wiener condition \eqref{W1} is necessary for this positivity under 
the Jellium condition \eqref{Wai}.
We expect that this  condition is also necessary for the positivity of $E''(S)$;
however, this is still an open challenging problem.


This chapter is organized as follows.
In Section 2 we introduce function spaces.
In Section 3 we collect all our assumptions.
In Section 4
we describe all fermionic jellium ground states and give basic examples.
In Section 5  we prove 
the stability of the solitary manifold $\cS$ establishing  
the positivity  the energy Hessian. 
In Appendices we construct the global dynamics.


\setcounter{equation}{0}
\section{Function spaces and integral equation}

Equation (\ref{vfN}) with the normalization (\ref{rQN}) 
is equivalent, up to a gauge transform,
to the
system (\ref{LPS1N})--(\ref{LPS3N}).
The system (\ref{HSiN}) can be written as 
\begin{equation}\la{HSN}
\dot X(t)=JE'(X(t)),
\end{equation}
where $J$ is the matrix (\ref{HS2}). 
For $\psi,\vp\in L^2(\ov\T)$ denote
\be\la{spN}
(\psi,\vp):=\rRe\int_{\ov\T}\psi(\ov x)\cdot\vp(\ov x)d\ov x,
\ee
where $\cdot$ is the inner product of the corresponding vectors in $\R^2$.

\bd
i)  Denote the  real Hilbert spaces
$$
\cX:=L^2(\ov\T)\oplus \R^{3\ov N}\oplus \R^{3\ov N},\qquad
\cW:=\cF\oplus \R^{3\ov N}\oplus \R^{3\ov N}.
$$
ii) $\cV:=\cF\otimes \ov\T\otimes \R^{3\ov N}$ is 
the real Hilbert manifold endowed with the metric 
\begin{equation}\la{dVsN}
d_{\cV}(X,X'):=\Vert\psi-\psi'\Vert_{\cF}+|q-q'|+|p-p'|,\qquad X=(\psi,q,p),\quad
X'=(\psi',q',p')
\end{equation}
and with the `quasinorm'
\be\la{cVsN}
|X|_{\cV}:=\Vert\psi\Vert_{\cF}+|p|,\qquad X=(\psi,q,p).
\ee

\ed
The linear space $\cW$ is 
isomorphic to
 the tangent space to 
the Hilbert manifold $\cV$ at each point $X\in\cV$. 
As in Chapter I, we will use notation (\ref{dWW}) and  the inequality (\ref{EQV}).
The system (\re{HSN})
 in the integral  form reads similarly to (\ref{LPSi}),
\begin{equation}\la{LPSiN}
\left\{\begin{array}{lll}
\psi(t)&=&e^{-iH_0t } \psi(0)+ie\ds\int_0^t e^{-iH_0(t-s)} [ \phi^\otimes(s)\psi(s) ]ds,\\
\\
q(n,t)&=&q(n,0)+\frac 1M\ds\int_0^t p(n,s)ds\mod N\Z^3,\\
\\
p(n,t)&=&p(n,0)-\ds\int_0^t (\nabla \phi(s),\sigma(\cdot-n-q(n,s))) ds,
\end{array}\right|
\end{equation}
where 
$H_0=-\fr 12\De^\otimes$
and $\phi(s):=G\rho(s)$.
In the vector form (\ref{LPSiN}) reads
\begin{equation}\la{LPSivN}
X(t)=e^{-tA}X(0)+\int_0^t  e^{-(t-s)A} N(X(s)) ds,
\qquad 
A=\left(\begin{array}{ccc}
	 iH_0  & 0 & 0\\
	0&0&0\\
	0&0&0
\end{array}\right),
\end{equation}
where
\be\la{HNN}
N(X)=(ie \phi^\otimes\5\psi ~, p,~f),\,\,\, f(n):=-(\nabla \phi,\sigma(\cdot-n-q(n))),\,\,\,\phi:=G\rho,
\ee
and $\rho$ is now defined by (\re{Hfor2N}), (\ref{reN}).

\setcounter{equation}{0}
\section{Main assumptions}

We assume the same 
 Jellium and the Wiener conditions (\re{Wai}) and  (\re{W1}) 
onto the  ion densities $\si(x)$. 
However, now we need an additional condition for the orbital stability 
of the ground state (\re{grN}).
Namely, recall that
the exterior product of functions 
$f_j\in L^2(\T)$
is defined by
\be\la{LamN}
[\Lam_{j=1}^{\ov N} f_j](\ov x):=\fr1{\sqrt{\ov N!}} \sum_{\pi\in S_{\ov N}} 
(-1)^{|\pi|} 
\prod\limits_{j=1}^{\ov N}f_j(x_{\pi(j)}),
\qquad \ov x=(x_1,...,x_{\ov N})\in\ov\T,
\ee
where $S_{\ov N}$ is the symmetric group and $|\pi|$ denotes the sign (or parity) of a 
transposition $\pi$.
Every eigenfunction  (\re{eig}) admits an expansion
in the exterior products
\be\la{gexN}
\psi_0(\ov x)=\sum_{\ov k} C(\ov k)\wedge_{j=1}^{\ov N} e^{ik_j x_j}.
\qquad k_j\in\Ga^*_N:=\fr{2\pi}N\Z^3.
\ee
Here  $\ov k:=\{k_1,...,k_{\ov N}\}$, where $k_j$ are different for 
distinct $j$, and $\fr12 \sum_{j=1}^{\ov N}k_j^2=\om_0$.
We will consider the eigenfunctions 
(\re{gexN}) with the additional restriction
\be\la{adrN}
\#(\ov k\setminus\ov k')\ge 2\quad{\rm if}
\quad \ov k\ne \ov k'. 
\ee
This condition implies that
the corresponding electronic charge density 
is uniform
(see Lemma \re{lre}),
\begin{equation}\label{roZN}
\rho^e(x)\equiv -eZ,\qquad x\in \T.
\end{equation}
This identity plays a crucial role in our approach. It implies that
the corresponding total charge density (\re{Hfor2N}) identically vanishes by (\re{sipi}).
Let us emphasize
that both ionic and electronic charge densities are uniform for the ground state under 
the Jellium
condition together with (\re{adrN}).
\medskip

\setcounter{equation}{0}
\section{Global dynamics}

Here we prove the global well-posedness of the 
system (\ref{HSN}).

\bt\label{TLWP1N}
Let   (\re{ro+}) hold and $X(0)\in\cM$.
Then 
\medskip\\
i) 
There exists a unique
solution 
$X(t)\in C(\R,\cV)$ 
to (\re{HSN}).
\medskip\\
ii) The energy and charge conservations (\re{EQN}) hold.

\et

First we construct the local solutions by contraction arguments.  
To construct the global solutions we will prove
 energy conservation 
using the Galerkin approximations.

\bt\label{TLWPN}(Local well-posedness).
Let   (\re{ro+}) hold and
$X(0)=(\psi_0,q_0,p_0)\in\cV$
with
 $| X(0)|_\cV\le R$. Then 
there exists $\tau=\tau(R)>0$ such that
  equation (\ref {HSN}) has  a unique  solution $X\in C([-\tau,\tau],{\cal V})$,
 and the maps $U(t):X(0)\mapsto X(t)$ are continuous in $\cV$ for $t\in [-\tau,\tau]$.

\et
In the
 next two propositions we prove
the boundedness and the local Lipschitz continuity of the nonlinearity 
$N:\cV\to\cW$ defined in (\re{HNN}).
With  this proviso Theorem \re{TLWP} follows from the integral form 
(\re{LPSivN}) of the equation (\ref{HSN})
 by the contraction mapping principle, since $e^{-At}$ is an isometry of $\cW$.
First, we prove the boundedness of $N$.
\bp\label{p1N}
For any $R>0$ and $X=(\psi,q,p)\in\cV$
\begin{equation}\label{bNN}
\Vert N(X)\Vert_\cW\le C(R)\qquad{\rm for}\quad |X|_\cV\le R.
\end{equation}
\ep
\begin{proof}
We need appropriate bounds for the charge density   $\rho$
and for the corresponding potential  $\phi$.

\bl The charge density (\re{Hfor2N}) admits the bounds 
\be\la{re4N}
\Vert\rho\Vert_{L^3(\T)}+\Vert\na\rho\Vert_{L^{3/2}(\T)}\le C(1+\Vert\psi\Vert_\cF^2).
\ee

\el
\begin{proof}
 We have
  $\rho(x)=\rho^i(x)+\rho^e(x)$ by (\re{Hfor2N}).
The bound
(\re{re4N})
for $\rho^i$ holds by (\re{ro+}). It remains to prove the bound for $\rho^e$.
Definition (\re{reN})   implies that  
\be\la{re2N}
\rho^e(x)=-e\sum_{j=1}^{\ov N} \int_{\T^{\ov N-1}} |\psi(\ov y)|^2\Big|_{y_j=x}\,dx_1...\widehat{dy_j}...dy_{\ov N},
\quad x\in\T,
\ee
where the hat means that this differential is omitted. 
Differentiating, we obtain that
\be\la{re22N}
\na\rho^e(x)=-e\sum_{j=1}^{\ov N} \int_{\T^{\ov N-1}} \na_{y_j}|\psi(\ov y)|^2\Big|_{y_j=x}\,dy_1...\widehat{dy_j}
...dy_{\ov N},\quad x\in\T.
\ee
Applying the Cauchy--Schwarz inequality to (\re{re2N}), we get 
\beqn\la{re23N}
\Vert\rho^e(x)\Vert_{L^3(\T)}&\le& C\sum_{j=1}^{\ov N} \int_{\T^{\ov N-1}} [\int_\T|\psi(\ov y)|^6dy_j]^{1/3}\,dy_1...
\widehat{dy_j}
...dy_{\ov N}
\nonumber\\
&\le&
\int_{\T^{\ov N-1}} [\int_\T|\na_{y_j}\psi(\ov y)|^2dy_j]\,dy_1...
\widehat{dy_j}
...dy_{\ov N}
\le
 C\Vert\psi\Vert_{H^1(\ov\T)}^2
\eeqn
by the Sobolev embedding theorem \ci[Theorem 5.4, Part I]{Adams}.
Similarly, (\re{re22N}) implies that
\beqn\la{re24N}
\Vert\na\rho^e(x)\Vert_{L^{3/2}(\T)}&\le& C\sum_{j=1}^{\ov N}
\int_{\T^{\ov N-1}} [\int_\T|\psi(\ov y)\na_{y_j}\ov\psi(\ov y)|^{3/2}dy_j]^{2/3}\,dy_1...
\widehat{dy_j}
...dy_{\ov N}
 \nonumber\\
&\le&
\int_{\T^{\ov N-1}} [\int_\T|\psi(\ov y)|^6dy_j]^{1/6}[\int_\T|\na_{y_j}\psi(\ov y)|^2dy_j]^{1/2}\,dy_1...
\widehat{dy_j}
...dy_{\ov N}
 \nonumber\\
&\le&
\int_{\T^{\ov N-1}} [\int_\T|\na_{y_j}\psi(\ov y)|^2dy_j]\,dy_1...
\widehat{dy_j}
...dy_{\ov N}
\le
 C\Vert\psi\Vert_{H^1(\ov\T)}^2
\eeqn
by the H\"older inequality and  the Sobolev embedding theorem.
\end{proof}

\bl
 The potential $\phi:=G\rho$ admits the bound
\be\la{PbN}
\Vert \phi\Vert_{C(\T)}+\Vert \na\phi\Vert_{L^3(\T)}\le C(1+\Vert\psi\Vert_\cF^2).
\ee
\el
\begin{proof}
 Applying   the H\"older and Hausdorff--Young inequalities to 
(\re{fs}), we obtain that
\be\la{s1N}
\Vert \phi\Vert_{C(\T)}\le C\Vert\fr{\ti\rho(\xi)}{\xi^2}\Vert_{L^1(\Ga^*_N\setminus 0)}
\le C_1\Vert \xi\tilde \rho\Vert_{L^3(\Ga^*_N)}
\Big[\sum_{\xi\in\Ga^*_N\setminus 0}|\xi|^{-9/2}\Big]^{2/3}
\le C_2\Vert\nabla \rho\Vert_{L^{3/2}(\T)}.
\ee
Similarly, 
\be\la{s2N}
\Vert \na\phi\Vert_{L^3(\T)}\le C\Vert\fr{\ti\rho(\xi)}{|\xi|}\Vert_{L^{3/2}(\Ga^*_N\setminus 0)}
\le C_1\Vert \xi\tilde \rho\Vert_{L^3(\Ga^*_N)}
\Big[\sum_{\xi\in\Ga^*_N\setminus 0}|\xi|^{-6}\Big]^{1/3}
\le C_2\Vert\nabla \rho\Vert_{L^{3/2}(\T)}.
\ee
Now the bound (\re{PbN}) follows from (\re{re4N}).
\end{proof}

Now we can prove the estimate (\re{bN}).
First, we will prove
\begin{equation}\label{eq2N}
\Vert\phi^\otimes\5\psi\Vert_\cF\le C(1+\Vert\psi\Vert_\cF^3)
\end{equation}
in the notation (\re{2De}).
According to definition (\re{FN})
it suffices to  check that
\begin{equation}\la{c1N}
\Vert\phi^\otimes\5\psi \Vert_{L^2(\ov\T)}+
\Vert \phi^\otimes\, 
\nabla^\otimes\5\psi\Vert_{L^2(\ov\T)}
+
\Vert \psi\nabla^\otimes\,\phi^\otimes\Vert_{L^2(\ov\T)}\le C(1+\Vert\psi\Vert_\cF^3).
\end{equation}
The first two summands   admit the needed estimate by (\re{PbN}). The third summand
requires some additional argument. Namely,
\beqn\la{aaN}
\Vert \psi 
\nabla^\otimes\,\phi^\otimes\5\Vert_{L^2(\ov\T)}^2
&=& \int_{\ov\T}|\sum_{j=1}^{\ov N} \na\phi(x_j)\psi(\ov x)|^2\,d\ov x
\le C\sum_1^{\ov N} \int_{\ov\T}|\na\phi(x_j)\psi(\ov x)|^2\,d\ov x
\nonumber\\
 &=&C\sum_{j=1}^{\ov N} \int_{\T^{\ov N-1}}
 \Big[\int_{\T}|\na\phi(x_j)\psi(\ov x)|^2\,d x_j\Big]dx_1...\widehat{dx_j}...dx_{\ov N}.
\eeqn
The inner integral is estimated as follows
\be\la{iiN}
\int_\T|\na\phi(x_j)\psi(\ov x)|^2\,d x_j\le \Vert \na\phi\Vert_{L^3(\T)}^2\Big[\int _\T|\psi(\ov x)|^6\,d x_j\Big]^{1/3}
\le C\Vert \na\phi\Vert_{L^3(\T)}^2\int_\T|\na_{x_j}\psi(\ov x)|^2\,d x_j
\ee
by the H\"older inequality and the Sobolev embedding theorem. Substituting this estimate into (\re{aaN}), we obtain
\begin{equation}\label{c3N}
\Vert \psi 
\nabla^\otimes\,\phi^\otimes\Vert_{L^2(\ov\T)}\le C\Vert \na\phi\Vert_{L^3(\T)}
\Vert\psi\Vert_\cF.
\end{equation}
This and (\re{PbN}) imply (\re{c1N}) for the third summand. 
Finally, using 
(\re{HNN}),
(\ref{PbN}) and (\ref{ro+}),
\begin{equation}\label{eq3N}
|f(n)|\le \Vert \phi\Vert_{C(\T)}\Vert\na\si\Vert_{L^1(\T)}
\le C(1+\Vert\psi\Vert_\cF^2),\qquad n\in\Ga.
\end{equation}
Hence, (\ref{eq2N}) and (\ref{eq3N}) imply (\ref{bNN}). Proposition \re{p1N} is proved.
\end{proof}

It remains to prove  that the nonlinearity is locally Lipschitz.

\begin{proposition}\label{p2N}

For any $R>0$ and $X_1,X_2\in\cV$ with $| X_1|_{\cal V},| X_2|_\cV\le R$
\begin{equation}\label{lNN}
\Vert N(X_1)-N(X_2)\Vert_\cW\le C(R) d_\cV(X_1,X_2).
\end{equation}
\end{proposition}
\begin{proof}
Writing  $X_k=(\psi_k,q_k,p_k)$ and $\phi_k=G\rho_k$, we obtain that
\begin{equation}\label{NN1N}
\Vert\phi^\otimes_1\5\psi_1-\phi^\otimes_2\5\psi_2\Vert_\cF\le\Vert(\phi^\otimes_1-\phi^\otimes_2)\5\psi_1\Vert_\cF
+\Vert\phi^\otimes_2 \5(\psi_1-\psi_2)\Vert_\cF.
\end{equation}
Using   (\ref{c3N}) and (\ref{PbN}), we obtain 
\begin{eqnarray}\nonumber
\Vert \phi^\otimes_2\5(\psi_1-\psi_2)\Vert_\cF&\le&
(\Vert \phi_2\Vert_{C(\T)}+\Vert \nabla \phi_2\Vert_{L^3(\T)})
\Vert\psi_1-\psi_2\Vert_\cF\\
\nonumber\\
&\le&C(1+R^2)\Vert\psi_1-\psi_2\Vert_\cF\le C(R)
d_\cV(X_1,X_2).\label{NN2N}
\end{eqnarray}
Further,
 (\ref{s1N}) and (\ref{s2N}) give that
\be\la{P12N}
\Vert\phi_1-\phi_2\Vert_{C(\T)}+\Vert\na\phi_1-\na\phi_2\Vert_{L^3(\T)}
\le\Vert\na\rho_1-\na\rho_2\Vert_{L^{3/2}(\T)}
\ee
However, 
 $|\si(x)-\si(x-a)|\le C|a|$,
where $|a|:=\min_{r\in a} |r|$ for $a\in \T$ (by definition, $a\subset\R^3$ is a class of equivalence mod $N\Z^3$).
Therefore,
as in (\ref{re4N}),
\be\la{re43N}
\Vert\na(\rho_1-\rho_2)\Vert_{L^{3/2}(\T)}\le CR
(| q_1-q_2|+\Vert\psi_1-\psi_2\Vert_\cF).
\ee
Hence,
\beqn
\Vert(\phi^\otimes_1-\phi^\otimes_2)\psi_1\Vert_\cF&\le&
(\Vert \phi_1-\phi_2\Vert_{C(\T)}+\Vert \nabla (\phi_1-\phi_2)\Vert_{L^3(\T)})\Vert\psi_1\Vert_\cF
\nonumber\\
\nonumber\\
&\le&
C R\Vert\na(\rho_1-\rho_2)\Vert_{L^{3/2}(\T)}\Vert\psi_1\Vert_\cF
\le C(R)d_\cV(X_1,X_2).  \label{ek}
\eeqn
Now (\re{NN1N})--(\re{ek}) give
\be\la{ThN}
\Vert\phi^\otimes_1\5\psi_1-\phi^\otimes_2\5\psi_2\Vert_{H^1(\T)}\le C(R)
d_\cV(X_1,X_2).
\ee
Similarly,  (\re{P12N}), (\re{re43N}) and (\ref{PbN})  imply 
\begin{eqnarray}\nonumber
&&\Vert\langle\na \phi_1,\sigma(\cdot-n-q_1(n))\rangle-\langle\na\phi_2,\sigma(\cdot-n-q_2(n))\rangle\Vert
\nonumber\\
\nonumber\\
&\le&\Vert\langle\na(\phi_1-\phi_2),\sigma(\cdot-n-q_1(n))\rangle\Vert
+\Vert\langle\nabla \phi_2,\sigma(\cdot-n-q_1(n))-\sigma(\cdot-n-q_2(n))\rangle\Vert
\nonumber\\
\nonumber\\
&\le&C(\Vert\phi_1-\phi_2\Vert_{C(\T)}+\Vert \phi_2\Vert_{C(\T)}) | q_1-q_2|)\le C(R)
d_\cV(X_1,X_2).
\end{eqnarray}
This estimate together with  (\re{ThN}) prove (\re{lNN}).
\end{proof}
Now 
Theorem \re{TLWPN} follows from 
Propositions \ref{p1N} and \ref{p2N}. 
\medskip\\
{\bf Proof of Theorem \re{TLWP1N}.}
The local solution  $X\in C([-\tau,\tau],\cV)$ to  (\ref {HSN})
exists and is unique by Theorem \re{TLWPN}.
On the other hand,
the conservation laws (\ref{EQN}) (proved in
Proposition \re{lgalN} iii) below)   
together with   (\re{EQV}) imply a priori bound
	\be\la{abN}
	| X(t)|_\cV^2\le C[E(X(0))+ Q(X(0))],\qquad t\in [-\tau,\tau].
	\ee
Hence, the local solution
admits an extension to the global solution
$X\in C(\R,\cV)$. 
\hfill$\Box$

\br
{\rm 
The condition $X(0)\in\cM$ implies that $X(t)\in\cM$ for all $t\in\R$ by 
  the charge conservation 
(\re{EQN}). 
	Hence,  (\ref {HSN}) implies (\ref{LPS1N})--(\ref{LPS3N}) with 
	the potential $\phi(\cdot,t)=G\rho(\cdot,t)$.
}
\er

\setcounter{equation}{0}
\section{Conservation laws}

We deduce the conservation  laws (\ref{EQN}) by the Galerkin approximations \ci{Lions}.

\bd
i)$\cV_m$ with $m\in\N$  denotes the finite-dimensional  submanifold of 
the Hilbert manifold
$\cV$ 
\be\la{VmN}
\cV_m:=\{(\sum_{\ov k} C(\ov k)\Lam_{j=1}^{\ov N} e^{ik_jx_j},q,p):k_j\in\Ga^*_N,~~
C(\ov k)\in\C,~~
\sum\limits_{j=1}^{\ov N}
k_j^2\le m,~ 
q\in\ov\T,~ p\in\R^{3\ov N}\},
\ee
where $\ov k:=(k_1,...,k_{\ov N})$.
\medskip\\
ii) $\cW_m$ with $m\in\N$  denotes the finite-dimensional linear subspace of the Hilbert space
$\cW$
\be\la{WmN}
\cW_m:=\{(\sum_{\ov k} C(\ov k) \Lam_{j=1}^{\ov N} e^{ik_jx_j},\vka,v):k_j\in\Ga^*_N,~~
C(\ov k)\in\C,~~
\sum\limits_{j=1}^{\ov N}k_j^2\le m,  ~ \vka\in\R^{3\ov N}, ~ v\in\R^{3\ov N}\}.
\ee
\ed
Obviously, $\cV_1\subset\cV_2\subset...$, 
the union $\cup_m\cV_m$ is dense in $\cV$, and 
$\cW_m$ are  invariant with respect 
to $H$ and $J$.
Let us denote by $P_m$
the orthogonal projector $\cX\to\cW_m$. This projector is also orthogonal in 
the Hilbert space $\cW$.
Let us 
approximate the system 
(\re{HSN}) by finite-dimensional Hamiltonian systems  on the manifold $\cV_m$,
\be\la{galN}
\dot X_m(t)=JE_m'(X_m(t)),\qquad t\in\R,
\ee
where $E_m:=E|_{\cV_m}$ and $X_m(t)=(\psi_m(t),q_m(t),p_m(t))\in C(\R,\cV_m)$.
The equation (\re{galN}) can be also  written as
\begin{equation}\la{galiN}
\langle \dot X_m(t),Y\rangle=-\langle E'(X_m(t)),JY\rangle,\qquad Y\in\cW_m.
\end{equation}
This form of the equation (\re{galN}) holds 
since $E_m:=E|_{\cV_m}$ and $\cW_m$ is  invariant with respect 
to  $J$.
Equivalently,
\begin{equation}\la{gali2N}
\dot X_m(t)=H X_m(t) + P_m N(X_m(t)).
\end{equation}
The Hamiltonian form guarantees the energy and charge conservation
 (\ref{EQN}):
\begin{equation}\la{EQ2N}
E(X_m(t))=E(X_m(0)),\quad Q(X_m(t))=Q(X_m(0)),\qquad t\in\R.
\end{equation}
Indeed, the energy conservation holds by the Hamiltonian form  (\re{galN}), 
while the charge conservation holds
by
the Noether theory \ci{A,GSS87, KQ} due to the
$U(1)$-invariance of $E_m$, see (\re{U1}).

The equation (\re{gali2N}) admits a unique local solution for every initial state
 $X_m(0)\in\cV_m$ since the right hand side 
is locally bounded and Lipschitz continuous. 
The global solutions exist by   (\re{EQV})
and
the energy and charge conservation (\re{EQ2N}).
\medskip

Finally, we take any $X(0)\in\cV$ and choose a sequence 
\be\la{s0N}
X_m(0)\to X(0),\qquad m\to\infty,
\ee
where the convergence holds in the metric of $\cV$.
Therefore, 
\be\la{EQmN}
E(X_m(0))\to E(X(0)),\qquad Q(X_m(0))\to Q(X(0)).
\ee
Hence, (\ref{EQ2N}) and (\re{EQV}) imply the basic uniform bound 
\be\la{VbN}
R:=\sup_{m\in\N}\,\,\sup_{t\in\R}| X_m(t)|_\cV <\infty.
\ee
Therefore, 
(\ref{gali2N})  and  Proposition \re{p1N} imply the second basic uniform bound
\be\la{Vb2N}
\sup_{m\in\N}\,\sup_{t\in\R}\,\Vert \dot X_m(t)\Vert_{\cW^{-1}} <C(R),
\ee
since  the operator 
$H:\cW\to\cW^{-1}$ is bounded, and the projector
$P_m$ is also a bounded operator in $\cW\subset \cW^{-1}$.
Hence, 
the Galerkin
approximations $X_m(t)$ are uniformly Lipschitz-continuous with values in $\cV^{-1}$:
\be\la{ecgN}
\sup_{m\in\N}\, d_{\cV^{-1}}(X_m(t),X_m(s))\le C(R)|t-s|,\qquad s,t\in\R.
\ee
Let us show that  
the uniform estimates   (\ref{VbN}) and (\ref{ecgN}) imply a compactness of the 
Galerkin approximations and the conservation laws. 
\bp\la{lgalN} Let   (\re{ro+}) hold and $X(0)\in\cV$. Then 
\medskip\\
i) 
There exists
a subsequence $m'\to\infty$ such that
\be\la{ssN}
X_{m'}(t)\tocX X(t),\qquad m'\to\infty,\qquad t\in\R,
\ee
where $X(\cdot)\in C(\R, \cX)$.
\medskip\\
ii) Every limit function $X(\cdot)$
is a solution to  (\re{HSN}), and $X(\cdot)\in C(\R,\cV)$.
\medskip\\
iii) The conservation laws (\re{EQN}) hold.
\ep
\begin{proof}
	i) The convergence (\re{ssN}) follows from 
	(\re{VbN}) and  (\re{Vb2N}) 
	by the Dubinsky
	`theorem on three spaces' \ci{Dub65}  (Theorem 5.1 of
	\ci{Lions}). Namely, the embedding $\cV\subset\cX$ is compact by the Sobolev theorem,
	and hence, (\re{ssN}) holds by 
	(\re{VbN})
	for $t\in D$, where $D$ is a countable dense set. 
	Finally, let us use the  interpolation inequality  and  (\re{VbN}), (\re{ecgN}):
	for any $\ve>0$
	\be\la{intiN}
	d_\cX(X_m(t),X_m(s))\le\ve d_\cV(X_m(t),X_m(s))
	+ C(\ve)  d_{\cV^{-1}}(X_m(t),X_m(s))\le 2\ve R+C(\ve,R)|t-s|.
	\ee
	This inequality 
	implies the 
	equicontinuity of the Galerkin approximations with values in $\cX$. Hence,
	convergence 
	(\re{ssN}) holds for all $t\in\R$ since it holds for 
	the dense set of $t\in D$. 
	The same equicontinuity also implies the continuity of the limit function $X\in C(\R, \cX)$.
	\medskip\\
	ii) 
	Integrating equation (\re{gali2N}), we obtain 
	\begin{equation}\la{gal2N}
	\int_0^t\langle \dot X_m(t),Y\rangle\,ds=\int_0^t\langle X_m(s),HY)\,ds + \int_0^t\langle N(X_m(s)),Y\rangle\,ds,
	\qquad Y\in\cW_m,
	\end{equation}
	Below we will write  $m$ instead of $m'$. 
	To prove  (\re{LPSivN}) it suffices to
	check that 
	in the limit $m\to\infty$, we get
	\begin{equation}\la{gal4N}
	\int_0^t\langle \dot X(t),Y\rangle\,ds
	=\int_0^t(X(s),HY)\,ds + \int_0^t\langle N(X(s)),Y\rangle\,ds,
	\qquad Y\in\cW_n,\qquad n\in\N.
	\end{equation}
	The convergence of  the left hand side and of the first term on the right hand side 
	of (\re{gal2N})
	follow
	from (\re{ssN}) and  (\re{s0N}) since $HY\in\cW_m$.

	It remains to consider the last integral in (\re{gal2N}).
	The integrand is uniformly bounded by (\re{VbN}) and Proposition \re{p1N}. 
	Hence, it suffices to check the pointwise
	convergence
	\be\la{NmN}
	\langle N(X_m(s), Y\rangle-\!\!\!-\!\!\!\!\to \langle N(X(s), Y\rangle,\qquad Y\in\cW_n
	\ee
	for any $s\in\R$. Here 
$N(X_m(s))=(ie\phi^\otimes_m(s)\psi_m(s),p_m(s),f_m(s))$ according to 
	the notation (\re{HNN}), and
	$Y=(\vp,\vka,v)\in \cW_n$. Hence, 
	(\re{NmN})
	reads
	\be\la{Nm2N}
	ie[\phi^\otimes_m(s)\psi_m(s),\vp]+p_m(s)\vka+f_m(s)v\,\to\,  
	ie(\phi^\otimes(s)\5\psi(s),\vp)+p(s)\vka+f(s)v,
	\ee
	where $[\cdot,\cdot]$ is the scalar product in $L^2(\ov\T)$.
	The convergence of $p_m(s)\vka$ follows from (\re{ssN}) (with $m'=m$) .
	To prove the convergence of the two remaining terms we first show that
	\be\la{2rtN}
	\phi_m(s):=G\rho_m \toLC \phi(s):=G\rho.
	\ee
	Indeed, 
	 (\re{ssN}) implies that 
	 \be\la{qpN}
	 \psi_m(s)\toLtt\psi(s),\qquad q_m(s)\to q(s)
	\ee
	Further, the sequence $\psi_m(s)$ is bounded in $H^1(\ov\T)$ by (\re{VbN}). 
	Hence, the sequence $\rho_m(s)$ is bounded in 
	the Sobolev space
	$W^{1,3/2}(\T)$ by (\re{re4N}).
	Therefore, the sequence $\rho_m(s)$ is precompact in $L^2(\T)$ by  the Sobolev compactness theorem. 
	Hence,
	\be\la{qp2N}
	\rho_m\toLt\rho
	\ee
	by (\re{qpN}).
	Therefore,
	(\re{2rtN}) holds
	since the operator $G:L^2(\T)\to C(\T)$ is continuous. 
	Finally, (\re{2rtN}) and (\re{qpN}) imply that
	\be\la{Nm3N}
	\phi^\otimes_m(s)\psi_m(s)\toLtt\5 \phi^\otimes(s)\psi(s),\qquad
	f_m(s)\to f(s),
	\ee
	which proves (\re{Nm2N}). Now (\re{gal4N}) is proved for $Y\in\cV_n$ with any $n\in\N$.
	Hence, $X(t)$ is a solution to (\re {HSN}). Finally,  
	$N(X(t))$ is  bounded in $\cW$ by
  (\re{VbN}) and Proposition \re{p1N}.
  Hence, 
 (\re{LPSivN}) implies that  $X(\cdot)\in C(\R,\cV)$.
	\medskip\\
	iii) The conservation laws (\re{EQ2N}) and the convergences (\re{s0N}), (\re{ssN}) imply that
	\be\la{EQ3N}
	E(X(t))\le E(X(0)),\quad Q(X(t))\le Q(X(0)),\qquad t\in\R.
	\ee
	The last inequality holds by the first convergence of (\re{qpN}).
The first inequality follows from the representation
\begin{equation}\la{EQ4N}
E(X_m(t))=\fr12 \Vert\na\psi_m(t)\Vert_{L^2({\T})}^2+\fr12 
\Vert \sqrt{G}\rho_m(t)\Vert_{L^2({\T})}^2 +
\sum_{n\in\Gamma_n}\fr{p_m^2(n,t)}{2M}.
\end{equation}
Namely, the last two terms on the right hand side converge 
by (\re{qp2N}) and (\re{ssN}). Moreover, 
the first term is bounded by (\re{VbN}). Hence,  the first convergence of
 (\re{qpN}) implies the weak convergence 
 \begin{equation}\la{EQ5N}
\na\psi_{m}(t)\toLwt\na\psi(t)
 \end{equation}
 by the Banach theorem.
Now the first inequality of (\re{EQ3N}) follows by the property of the weak convergence
in the Hilbert space.
Finally, the opposite inequalities 
 to (\re{EQ3N})
 are also true by the uniqueness of 
solutions $X(\cdot)\in C(\R,\cV)$, which is proved in Proposition \re{TLWPN}.
\end{proof}

\setcounter{equation}{0}
\section{Jellium fermionic ground states}
In this section we check the key identity (\re{roZN})
and
construct all solutions to (\re{HSN}) with minimal 
energy  (\re{HforN}). Furthermore we
give examples illustrating the Jellium and 
the Wiener conditions.

\subsection{Uniform electronic charge density}
Let us  establish the   identity (\re{roZN}).
\bl\la{lreN} 
Let  the condition (\re{adrN}) hold for an
 eigenfunction (\re{gexN}), and 
\be\la{normN}
\int_{\ov\T}|\psi_0(\ov x)|^2d\ov x=Z.
\ee
Then the identity  (\re{roZN}) holds. 
\el
\begin{proof}
By the  antisymmetry of  $\psi_0(x_1,...,x_{\ov N})$
it remains to prove that 
\be\la{norm2N}
\int_{\ov\T}\de(x-x_1)|\psi_0(\ov x)|^2d\ov x=Z/\ov N,\qquad x\in\T.
\ee
Substitute here the expansion (\re{gexN}).
The normalization condition (\re{normN}) gives
\be\la{ppo12N}
\sum_{\ov k}|C(\ov k)|^2 \ov N^{\ov N}=Z.
\ee
Further,
\beqn
&&
\int_{\ov\T}\de(x-x_1)|\psi_0(\ov x)|^2d\ov x
= \fr 1{\ov N!}\sum_{\ov k}\Big\{|C(\ov k)|^2
\int_{\ov\T}\de(x\!-\!x_1)
\Big[\sum_{\pi,\pi'\in S_{\ov N}} 
(-1)^{|\pi|+|\pi'|} \prod\limits_{j=1}^{\ov N}
e^{i[k_{\pi(j)}-k_{\pi'(j)}]x_j}
\Big]
d\ov x\Big\} \nonumber 
\\
&+&\fr 1{\ov N!}\rRe\sum_{\ov k\ne \ov k'}\Big\{C(\ov k)\ov C(\ov k')
\Big[\sum_{\pi,\pi'\in S_{\ov N}} 
(-1)^{|\pi|+|\pi'|} \int_{\ov\T}\de(x\!-\!x_1)
\prod\limits_{j=1}^{\ov N}
e^{i[k_{\pi(j)}-k'_{\pi'(j)}]x_j}
d\ov x\Big]\Big\}. \la{aN}
\eeqn
The integrals in the last line vanish 
since $k_{\pi(j)}-k'_{\pi'(j)}\ne 0$
at least for one $j\ne 1$ by (\re{adrN}). 
On the other hand,
the integrals in the first line do not vanish
only in the case when $k_{\pi(j)}\equiv k_{\pi'(j)}$ for $j\ne 1$, i.e.,
when $\pi=\pi'$. Hence, 
\be
\int_{\ov\T}\de(x-x_1)|\psi_0(\ov x)|^2d\ov x
= 
 \ov N^{\ov N-1}
\sum_{\ov k}|C(\ov k)|^2
\int_{\T}\de(x\!-\!x_1)
d x_1=\ov N^{\ov N-1}
\sum_{\ov k}|C(\ov k)|^2=Z/\ov N
\ee
by (\re{ppo12N}).

\end{proof}
\br\la{rPAN}
Similar calculations 
show that the uniformity (\re{norm2N}) can break down
for the wave functions (\re{gexN}) if the condition (\re{adrN}) fails.
\er

\subsection{Description of  all ground states}

The following lemma gives the  description all ground states
of the system (\re{HSN}).  

\bl\la{JgsN}
Let the  Jellium condition (\re{Wai}) hold. Then 
all solutions 
to (\re{LPS1N})--(\re{LPS3N})
of minimal energy are $(e^{i\al}\psi_0(\ov x)e^{-i\om_0t},q^*,0)$ 
with
$\al\in[0,2\pi]$ and
 $q^*\in\T^{\ov N}$ satisfying 
the identity
\begin{equation}\la{sipiqN}
	\sum_{n\in{\Ga}}\si(x-q^*(n))\equiv eZ,\qquad x\in{\T},
\end{equation}

\el
\begin{proof}
First,
let us note that the ${\Ga}$-periodic solutions (\re{grN}) 
have the minimal energy, and 
the identity (\re{sipiqN}) holds for $q^*=\ov r$ by  (\re{sipi}).

Further, a solution is of  minimal  energy  (\re{HforN}) if
the first summand
 on the right hand side of  (\re{HforN}) takes the minimal values $\om_0$ while the second and the third summands should vanish.
Hence, the normalization condition (\re{rQN}) gives
  \begin{equation}\la{ppoN}
  \psi(\ov x,t)\equiv e^{i\al(t)}\psi_0(\ov x),\qquad\al(t)\in \R,
\end{equation}
while
\be\la{rodqp}
 \rho(x)\equiv 0,\qquad M\pa_t q(n,t)=p(n,t)\equiv 0.
 \ee
Therefore, (\re{roZN}) implies 
\begin{equation}\label{roZ2N}
 \rho^i(x,t)\equiv eZ,\qquad q(n,t)\equiv q^*(n),
 \end{equation}
which is
equivalent to (\ref{sipiqN}).

Moreover, the Poisson equation (\re{LPS2N}) with $\rho(x,t)\equiv 0$
implies that $\phi(x,t)\equiv 0$ after a gauge transformation (\ref{GT}).
Hence,   finally,
substituting  (\re{ppoN}) into  (\re{LPS1N}) 
with $\phi(x,t)\equiv 0$,
we obtain that $-\dot\al(t)\equiv\om_0$.
 \end{proof}

This lemma implies that all  ${\Ga}$-periodic ground states are given by
(\re{grN}).
Recall that the examples of charge densities satisfying 
Jellium and the Wiener conditions are given in Section \ref{sex}.

\setcounter{equation}{0}
\section{The orbital stability of the ground state}

In this section we expand the energy into the Taylor series and prove the orbital stability checking the
positivity of the energy Hessian.



\subsection{The Taylor expansion of energy functional}\la{sTN}

 We will deduce the lower estimate (\re{BLiN}) using
 the  Taylor expansion of $E(S+Y)$ at
$S=S_{\al,r}=(\psi^\al,\ov r,0)\in \cS$ and $Y=(\vp,\vka,\pi)\in\cW$.
Further calculations are very close to the one of 
Section 
\ref{sT}:
 \be\la{teN} 
E(S+Y)=E(S)+\langle E'(S),Y\rangle +\fr12 \langle Y,E''(S)Y\rangle + R(S,Y)=\om_0Z+
\fr12 \langle Y,E''(S)Y\rangle + R(S,Y)
\ee
since $E(S)=\om_0Z$ by (\re{ESN}), and $E'(S)=0$.
Here  $E'(S)$ and  $E''(S)$ stand for the G\^ ateaux differentials.
Let us recall that $\psi^\al=e^{i\al}\psi_0(x)$ where $\psi_0(x)$
is given by (\re{gexN}), and the condition (\re{adrN}) holds.

First,
we expand 
 the  charge density (\re{Hfor2N}), (\re{reN})
corresponding to $S+Y=(\psi^\al+\vp,\ov r+\vka,\pi)$:
\be\la{roN}
\rho(x)=\rho^{(0)}(x)+\rho^{(1)}(x)+\rho^{(2)}(x),\qquad x\in\T,
\ee
where $\rho^{(0)}$ and  $\rho^{(1)}$ are respectively 
the  terms of zero and first  order in 
$Y$, while  $\rho^{(2)}$ is the remainder.
However, $\rho^{(0)}(x)$ is the total charge density of the ground state
which is identically zero
by
(\ref{sipi}) and (\re{roZN}):
\begin{equation}\la{ro0N}
\rho^{(0)}(x)=\rho^i_0(x)-e|\psi^\al(x)|^2\equiv 0,\qquad x\in{\T}.
\end{equation}
Thus, $\rho=\rho^{(1)}+\rho^{(2)}$.
Expanding
 (\re{Hfor2N}) and  (\re{reN}), we obtain
\beqn
\la{ro11N}
\!\!\!\!\!\!\!\!\!\!\!\!\!\!\!\!\!\!\!\!\!\!\!\!\!\!\rho^{(1)}(x)\!\!\!\!&\!\!\!\!=\!\!\!\!&\!\!
\si^{(1)}(x)
-2e\,
\Si(\psi^\al)\vp(x)
,\quad \si^{(1)}(x)=-\sum_{n\in\Ga} \vka(n)\cdot\na\si(x-n-r),\quad
\\
\la{ro13N}
\!\!\!\!\!\!\!\!\!\!\!\!\!\!\!\!\!\!\!\!\!\!\!\!\!\!\rho^{(2)}(x)\!\!\!\!&\!\!=\!\!\!\!&\!\!\si^{(2)}(x)
\!-\!e
\,
\Si(\vp)\vp(x)
,~ \si^{(2)}(x)\!=\!\fr12\sum_{n\in\Ga}
\int_0^1\!\!(1\!-\!s)
[\vka(n)\cdot\na]^2\si(x\!-\!n\!-\!r-\!s\vka(n))ds,
\eeqn
where $\Si(\psi)$ denotes the real-linear operator defined by 
\be\la{wwdN}
\Si(\psi)\vp(x):=\sum\limits_{j=1}^{\ov N}
%
\int_{\ov\T}\de(x-y_j)\psi(\ov y)\cdot\vp(\ov y)\,d\ov y,
\qquad x\in\T,
\ee
where $\cdot$ is the inner product of the corresponding vectors  $\psi(\ov y)=(\psi_1(\ov y),\psi_2(\ov y))\in\R^2$
and
$\vp(\ov y)=(\vp_1(\ov y),\vp_2(\ov y))\in~\R^2$.
The operator $\Si(\psi^\al)$ is continuous $L^2(\ov\T)\to L^2(\T)$, which follows similarly to (\ref{re23N}). 
Obviously, 
\be\la{si12}
\Si(\psi^\al)\vp(x)=\Si(\psi^\al_1)\vp_1(x)+\Si(\psi^\al_2)\vp_2(x).
\ee
Note that $\Si(\vp)\vp\in L^2(\T)$ for $\vp\in\cF$ by estimate (\ref{re4N}).
\smallskip

Substituting  $\psi=\psi^\al +\vp$ and $\rho=\rho^{(1)}+\rho^{(2)}$ 
 into (\ref{HforN}), we obtain that
the  quadratic part  of (\re{teN}) reads 
\begin{equation}\la{B2N}
\fr12\langle Y,E''(S) Y \rangle=\fr12\int_{\ov\T}|\na \vp(\ov x)|^2]d\ov x+
\fr12 (\rho^{(1)},G\rho^{(1)})+K(\pi),~~~~~~ K(\pi):=\ds\sum_n\fr{\pi^2(n)}{2M}~~
\end{equation}
and the remainder equals
\begin{equation}\la{B3N}
R(S,Y)=\fr 12 (2\rho^{(1)}+\rho^{(2)},G\rho^{(2)}). 
\end{equation}

\subsection{The null space of Hessian}
In this section
 we calculate the null space 
\be\la{KYdN}
\cK(S):=\Ker\, \Big[E''(S)\Big|_\cW\Big],\qquad S\in \cS.
\ee

\bl\la{lWN}
Let the Jellium condition (\re{Wai}) and the Wiener condition   
(\re{W1}) hold, and $S\in\cS$. Then
\be\la{KYN}
\cK(S)=\{(0,\ov s,0):~~s\in\R^3 \},
\ee
where $\ov s\in\R^{3\ov N}$ is defined similarly to (\re{gr2}): 
$\ov s(n)\equiv s$.
\el
\begin{proof}
	All the summands of
	the
	energy (\re{B2N}) are nonnegative. Hence,  this expression is zero if and only if
	all the summands vanish: in the notation (\re{ro11N})
	\be\la{rbN}
	\vp(\ov x)\equiv C,\quad
	(\rho^{(1)}, G\rho^{(1)})=\Vert\sqrt{G}[\si^{(1)}-2e
	\Si(\psi^\al)\vp
	]\Vert_{L^2(\T)}^2
	=0,\quad \pi=0.
	\ee
	Here $C=0$ by the antisymmetry of $\vp$ since $N>1$. Therefore, 
	 $\Si(\psi^\al)\vp= 0$, and hence,
 (\re{rbN}) implies that
	\be\la{rb2N}
	\sqrt{G}\si^{(1)}=0.
	\ee
	On the other hand, 
	in the Fourier transform 
	(\re{ro11N}) reads
	\be\la{B314N}
	\ti\si^{(1)}(\xi)=\ti\si(\xi)\xi\cdot\sum_{n\in\Ga} ie^{i\xi [n+r]}\vka(n)
	=i\ti\si(\xi)\xi\cdot  e^{i\xi r} \hat \vka(\xi),\qquad\xi\in \Ga^*_N,
	\ee
	where $\hat \vka(\xi):=\sum_{n\in\Ga} e^{i\xi n}\vka(n)$ is a $2\pi\Z^3$-periodic function on
	$\Ga^*_N$.
	Hence, Definition (\re{fs}) and the Jellium condition (\re{Wai}) imply that
	\beqn\la{B315N}
	0=\Vert\sqrt{G}\si^{(1)}\Vert_{L^2(\T)}^2&=&
	N^{-3}\sum_{\Ga^*_N\setminus \Ga_1^*} |\ti\si(\xi)\fr{\xi\hat \vka(\xi)}{|\xi|}|^2
	\nonumber\\
	\nonumber\\
	&=&N^{-3}\sum_{\theta\in\Pi^*\setminus \Ga_1^*} 
	\langle\hat \vka(\theta),
	\sum_{m\in\Z^3}\Big[\fr{\xi\otimes\xi}{|\xi|^2}|\ti\si(\xi)|^2\Big]_{\xi=\theta+2\pi m}\hat \vka(\theta)\rangle
	\nonumber\\
	\nonumber\\
	&=&N^{-3}\sum_{\theta\in\Pi^*\setminus \Ga_1^*} 
	\langle\hat \vka(\theta),
	\Si(\theta)
	\hat \vka(\theta)\rangle.
	\eeqn
	As a result, 
	\be\la{B316N}
	\hat \vka(\theta)=0,\qquad \theta\in \Pi^*\setminus \Ga_1^*
	\ee
	by the Wiener condition
	(\re{W1}).
	On the other hand, $\hat \vka(0)\in\R^3$ remains arbitrary.
	Respectively, 
	$\vka=\ov s$ with 
	an arbitrary $s\in\R^3$.\end{proof}
\brs\la{rPa}
{\rm
i) The annihilation of $\vp(x)\equiv C$ in  (\ref{rbN}) is the main point where
we use the antisymmetry of wave functions (`Pauli exclusion principle').
\smallskip\\
ii)
The key point of the proof is the explicit calculation (\re{B314N}) 
in the Fourier transform.
This calculation relies on the invariance of
the Hessian $E''(S)$ with respect to $\Ga$-translations which is 
due to the periodicity of the ions arrangement of the ground state.
}
\ers

\br{\it  Beyond the Wiener condition}
{\rm
If the Wiener
condition (\re{W1}) fails, the dimension of the space
\be\la{VN}
V:=\{v\in \R^{3\ov N}:~~ v(n)=\sum_{\theta\in\Pi^*\setminus\Ga_1^*} e^{-i\theta n} \hat v(\theta),
\qquad \hat v(\theta)\in\C^3,~~ \Si(\theta)\hat v(\theta)=0\}
\ee
 is positive.
The above calculations  show that in this  case
\be\la{KYgN}
\cK(S)=\{(0, \ov s+v,0):~~s\in\R^3,~~v\in V \}.
\ee
The subspace $V\subset \R^{3\ov N}$ is orthogonal to the $3D$ subspace 
$\{\ov s:s\in\R^3\}\subset \R^{3\ov N}$
by the Parseval theorem.
Hence,
$\dim \cK(S)=3+d$, 
where $d:=\dim V>0$. 
Thus, $\dim \cK(S)>3$.
Under the Wiener condition $V=0$, and (\re{KYgN}) coincides with (\re{KYN}).
}
\er

\subsection{The positivity  of  Hessian}

As in Section \ref{spoH}, denote by
$ N_S\cS$ the normal subspace to $\cS$ at a point $S$:
\be\la{L0NN}
 N_S\cS:=\{Y\in\cW: \langle Y,\tau\rangle=0,~~\tau\in  T_S\cS\},
\ee
where $ T_S\cS$ is the
tangent space to $\cS$ at the point $S$ and $\langle\cdot,\cdot\rangle$
stands for the scalar product 
 (\re{dWW}).
Obviously, $\cS\subset\cM$, and  the tangent space to $\cM$ at a point
$S=(\psi^\al,\ov r, 0)$
is given by 
 \begin{equation}\la{TSMN}
T_S\cM=\{(\vp, \vka,\pi)\in\cW:\vp\bot\psi^\al \},
\end{equation}
since $DQ(\psi^\al,\ov r,0)=2(\psi^\al,0,0)$.

\bl\la{lW2}
Let
 the 
Jellium condition (\re{Wai}) hold, and $S=S_{\al,r}\in \cS$. Then the Wiener condition (\re{W1}) is necessary and sufficient
for the positivity of the Hessian 
$E''(S)$
in the orthogonal directions to $\cS$ on $\cM$,
i.e., 
\be\la{cM2N}
E''(S)\Big|_{N_S\cS\cap  T_S\cM}>0.
\ee
\el
\begin{proof} i) Sufficiency.
Differentiating $S_{\al,r}=(e^{i\al}\psi_0,\ov r,0)\in \cS$ in the parameters $\al \in [0,2\pi]$ and $r\in\T$, 
we obtain 
\begin{equation}\la{tvN}
T_S\cS=\{(iC\psi^\al ,\ov s,0): ~~C\in\R,~~s\in\R^3\}.
\end{equation}
	Hence,
	(\re{KYN}) implies that 
	\be\la{tv2N}
	\cK(S)\cap  N_S\cS=
	(0 ,0,0).
	\ee	
Now  (\re{cM2N}) follows since $E''(S)\ge 0$ by (\re{B2N}).
	\medskip\\
 ii)	Necessity.
If the Wiener condition (\re{W1}) fails,
the space $\cK(SN)$ is given by (\re{KYgN}), and hence, 
(\re{tvN}) implies that now
\be\la{tv3N}
\cK(S)\cap  N_S\cS=
\{0 ,v,0): ~~C\in\R, ~~ v\in V\}\subset T_S\cM.
\ee
Therefore, the Hessian $E''(S)$ vanishes on the nontrivial space
$\cK(S)\cap  N_S\cS\subset   T_S\cM$ of the
dimension  $d>0$. 
Respectively, the positivity (\re{cM2N}) breaks down.
\end{proof}

\br\la{rSN} {\rm The positivity of type (\re{cM2N}) breaks down for the submanifold 
	$\cS(r):=\{S_{\al ,r}:~\al\in[0,2\pi]\}$ 
	with a fixed $r\in\T$ instead of  the solitary manifold $\cS$.
	Indeed, then the corresponding tangent space  is smaller:
	\be\la{tvrN}
	 T_S\cS(r)=\{(iC\psi^\al ,0,0):\,\,C\in\R\}.
	\ee
	Hence, the normal subspace $ N_S\cS(r)$ is larger, in particular containing all 
the vectors $(0,\ov s,0)$
	generating  the shifts of the torus. However, all these vectors also 
belong 
	to the null space (\re{KYN}) and to $ T_S\cM$.
	Respectively, the null space of the Hessian $E''(S)$ in $ T_S\cM\cap N_S\cS(r)$ is  at least 3-dimensional.
}
\er

\subsection{The orbital stability}
Here we prove  Theorem \re{tmN} which is the  main result
of this chapter.
For the proof is suffices to check
the lower energy estimate (\re{BLiN}): 
\be\la{BLN}
E(X)-\om_0Z\ge \nu\,d^2_\cV(X,\cS)\quad{\rm if}\quad d_\cV(X,\cS)\le \de,
\quad X\in\cM
\ee
with some $\nu,\de>0$. 
This estimate implies Theorem \re{tmN} since the energy is conserved
along all trajectories.
First, we prove similar lower bound for the energy Hessian. 
\bl
Let conditions of Theorem \re{tmN} hold. Then  for each $S\in \cS$
\be\la{L02N}
\langle Y, E''(S)Y\rangle > \nu\Vert Y\Vert_\cW^2,  \qquad Y\in N_S\cS\cap   T_S\cM ,
\ee
 where $\nu>0$.
\el
\begin{proof}
First, we note that $E''(S)$ is not complex linear due to the integral 
in (\re{HforN}). Hence, we should express the action of $E''(S)$ 
in $\vp_1(x):=\rRe\vp(x)$ and $\vp_1(x):=\rIm\vp(x)$:
formulas (\ref{B2N}) and
(\ref{ro11N}), (\ref{si12})
 imply that similarly to (\ref{E''}),
\begin{equation}\la{E''N}
E''(S)Y=
\left(\begin{array}{cccl}
  -\De^\otimes+4e^2\Si^*(\psi^\al_1) G\ti\Si(\psi^\al_1) & 4e^2\Si^*(\psi^\al_1) G\ti\Si(\psi^\al_2) & 2L & 0
 \medskip\\
 4e^2\ti\Si^*(\psi^\al_2) G\Si(\psi^\al_1) &   -\De^\otimes +4e^2\Si^*(\psi^\al_2) G\ti\Si(\psi^\al_2) &0 & 0\medskip\\
 2L^{\5*}  &    0  &   T    & 0  \\
      0      &    0            &   0    &  M^{-1} \\
\end{array}\right)Y
\end{equation}
for $Y=(\vp_1 , \vp_2 , \vka , \pi )$.
Similarly to (\ref{S}), 
the operator $L$  corresponds to
the matrix
\begin{equation}\la{SN}
 L(\ov x,n):=e\Si^*(\psi^\al)G\na\si(x-n):~~~x\in\R^3,~n\in\Ga,
\end{equation}
while $T$ coincides with  (\ref{T}).
\smallskip

Thus, $E''(S)$  is the selfadjoint operator in $\cX$ with the discrete spectrum, and
  (\re{cM2N}) implies that the minimal eigenvalue of
$E''(S)$  in the invariant space $N_S\cS\cap   T_S\cM$ is  positive. Therefore, (\re{L02N}) follows.
\end{proof}
\medskip

The positivity (\re{L02N}) implies
the lower energy estimate (\re{BLN})  since the 
higher-order terms in (\re{teN}) 
are negligible by the following  lemma.

\bl\la{lre2N}
Let $\si(x)$ satisfy  (\re{ro+}).
Then the  remainder  (\re{B3}) admits the bound
\be\la{B31N}
|R(S,Y)|\le C
\Vert Y\Vert_\cW^3
\quad\,\,\,{\rm for}\quad  \,\,\,
\Vert Y\Vert_\cW \le 1.
\ee
\el
\begin{proof} Due to (\re{B3N}) 
	it suffices to prove the 
	estimates
	\be\la{B312N}
	\Vert\sqrt{G}\rho^{(1)}\Vert_{L^2(\T)}\le C_1\Vert Y\Vert_\cW,\quad  
	\Vert\sqrt{G}\rho^{(2)}\Vert_{L^2(\T)}\le C_2\Vert Y\Vert_\cW^2
	\quad \,\,\,{\rm for}\,\,\,\quad     \Vert Y\Vert_\cW    \le 1.
	\ee
	i)
	By (\re{ro11N}) we have
	for $Y=(\vp,\vka,\pi)$ 
	\be\la{B313N}
	\sqrt{G}\rho^{(1)}=\sqrt{G}\si^{(1)}-2e\sqrt{G}\,
	\Si(\psi^\al)\vp
	%
	\ee
	in the notation (\re{wwdN}).
	The operator $\sqrt{G}$ is bounded in $L^2(\R^3)$ by (\re{fs}). 
	Hence,  (\re{ro11N}) and (\ref{ro+}) imply by the Cauchy-Schwarz inequality,  
	\be\la{GP1N}
	\Vert\sqrt{G}\si^{(1)}\Vert_{L^2(\T)}\le C |\vka|.
	\ee
 Finally, applying the Cauchy--Schwarz and Hausdorff--Young inequalities
to the second term on the RHS of  (\re{B313N}),
	we obtain 
	\beqn\la{GP2N}
	\Vert\sqrt{G}\,
	\Si(\psi^\al)\vp
	\Vert_{L^2(\T)}&\le& 
	C\Big[\sum_{\xi\in\Ga^*_N\setminus 0}\fr{|\ti\vp(\xi)|^2}{|\xi|^2}\Big]^{1/2}
\le C\Vert\ti\vp\Vert_{L^4(\Ga^*_N)}\Big[\sum_{\xi\in\Ga^*_N\setminus 0}|\xi|^{-4}\Big]^{1/2}
\nonumber\\
&\le& C_1\Vert\vp\Vert_{L^{4/3}(\T)}\le C_2\Vert\vp\Vert_{H^1(\T)}^2
	\eeqn
by the Sobolev embedding theorem.
	Hence, the first inequality (\re{B312N}) is proved.
	\medskip\\
	ii) Now we prove the second 
	inequality  (\re{B312N}).  According to (\re{ro13N}),
	\be\la{B317N}
	\sqrt{G}\rho^{(2)}(x)=\sqrt{G}\si^{(2)}(x)-e\sqrt{G}\,
	\Si(\vp)\vp
	\ee
	Now (\ref{ro+}) and the triangle inequality give
	\be\la{GP3N}
	\Vert\sqrt{G}\si^{(2)}\Vert_{L^2(\T)}\le C|\vka|^2.
	\ee
	At last, 
	we obtain similarly to (\re{GP2N}) 
	\beqn\la{fpN}
	\Vert\sqrt{G}\,
	\Si(\vp)\vp
	\Vert_{L^2(\T)}
	\le
	C\Big[\sum_{\xi\in\Ga^*_N\setminus 0}\fr{|\ti \beta(\xi)|^2}{|\xi|^2}\Big]^{1/2}
	\le
	C_1\Vert \beta\Vert_{L^{4/3}(\T)}.
	\eeqn
Finally,
applying the triangle inequality and the Sobolev embedding theorem, we obtain
\beqn\la{fp2N}
\Vert \beta\Vert_{L^{4/3}(\T)}
&\le&
\int_{\T^{\ov N-1}} [\int_\T|\vp(\ov x)|^{8/3} dx_j]^{3/4}\,dx_1...
\widehat{dx_j}
...dx_{\ov N}
\nonumber\\
&\le&
\int_{\T^{\ov N-1}} [\int_\T|\na_{x_j}\vp(\ov x)|^2dx_j]\,dx_1...
\widehat{dx_j}
...dx_{\ov N}
\le
 C\Vert\vp\Vert_{H^1(\T)}^2.
\eeqn
Now the lemma is proved.
\end{proof}

\part{Linear stability  of infinite crystals}

In the second part of present book, we 
consider infinite crystals with many ions per cell of a lattice.

We construct the ground state for 1D, 2D and 3D lattices 
in 3D space.
The main results are  well-posedness of the linearised
dynamics in the space of finite-energy states and
its dispersive decay.

The presentation mainly relies on our papers \ci{KK2014a, KK2014b, K2014, K2015, KKpl2015, KK_2018} with suitable extensions.

\chapter{On existence of  ground states}

\centerline{Abstract} 
\medskip

A~space-periodic ground state is shown to exist for lattices of
smeared
ions  in $\R^3$ coupled to the Schr\"odinger and scalar fields. 
The  elementary cell is necessarily neutral. 
 
The 1D, 2D and 3D lattices in $\R^3$ are considered, and 
a~ground state is constructed by minimizing the energy per cell. 
The case of a~3D lattice is rather standard,  because 
the elementary cell is compact, and the spectrum of the Laplacian is discrete. 

In the cases of  1D and 2D lattices, the energy functional 
is differentiable only on a dense set of variations, 
due to the presence of the continuous spectrum 
of the Laplacian that causes the infrared divergence of the Coulomb bond. 
Respectively,  
the construction of electrostatic potential 
and
the derivation of the Schr\"odinger equation for the minimizer 
 in these cases require an extra argument.

 The space-periodic ground states for 1D and 2D lattices 
 give the model of
 the nanostructures similar to 
 the carbon nanotubes and graphene respectively.



\setcounter{equation}{0} 
\setcounter{section}{-1} 
\setcounter{section}{0} 
\section 
{Introduction} 
In this chapter
we prove the existence
 of  ground states for infinite crystals in $\R^3$
with $d$-dimensional ion lattices
\be\la{Ga3} 
\Ga_d:=\{\x(\n)=\bba_1\n^1+\dots+\bba_d\n^d: \n\in\Z^d  \}, 
\ee 
where $d=1,2,3$ and $\bba_k\in\R^3$ are linearly independent periods. 
A~2D (respectively, 1D lattice) is a mathematical model of a~monomolecular film (a~wire). 
 
 The ions are
considered as classical nonrelativistic particles governed by the Lorentz equations 
neglecting the magnetic field, 
while the electrons are  described by the Schr\"odinger 
equation neglecting the electron spin. 
The scalar potential 
is the solution to the corresponding Poisson equation. 

We consider the crystal with $N$ ions per cell.
Let us denote  by $\rho_j$
the charge density of an ion and by $M_j>0$ its mass, $j=1,...,N$.
Then the coupled equations 
 read 
\beqn 
i\h\dot\psi(\x,t)&=&-\fr{\h^2}{2\cm}\De\psi(\x,t)-e\phi(\x,t)\psi(\x,t), ~~~~~\x\in\R^3,
\la{LPS1g} 
\\ 
\nonumber\\ 
\qquad -\De 
\phi(\x,t)&=&\rho(\x,t):= 
\sum_{j=1}^N\5\sum _{\n\in\Z^d} 
\rho_j(\x-\x(\n)-\x_j(\n,t))-e|\psi(\x,t)|^2, ~~~~~\x\in\R^3,\la{LPS2g} 
\\ 
\nonumber\\ 
M_j\ddot\x_j(\n,t) 
&=&-(\na\phi(\x,t),\rho_j(\x-\x(\n)-\x_j(\n,t))),
\quad \n\in\Z^d,\quad 
j=1,\dots,N. 
\la{LPS3w} 
\eeqn 
Here 
$e>0$ is the elementary charge, 
$\cm$ is the electron mass, 
$\psi(\x,t)$ denotes the 
wave function of the electron field, and 
$\phi(\x,t)$ is the electrostatic  potential 
generated by the ions and the electrons. 
Further, 
$(\cdot,\cdot )$ 
stands for the scalar product in the Hilbert 
space $L^2(\R^3)$. 
All derivatives here and below 
are understood in the sense of distributions. 
The system is  nonlinear and translation invariant, 
i.e., $\psi(\x-\bba,t)$, $\phi(\x-\bba,t)$, 
$\x_j(\n,t)+\bba$ is also a~solution for any $\bba\in\R^3$ . 
 
 
The ground state is
 a~$\Ga_d$-periodic stationary solution 
$\psi^0(\x)e^{-i\om ^0t}$,
$\phi^0(\x)$, 
$\ov\x=(\x^0_1,~\dots,\x^0_N)$  of the system 
(\re{LPS1g})--(\re{LPS3w}): 
\beqn\la{LPS1w} 
\h\om^0\psi^0(\x)&=&-\fr{\h^2}{2\cm}\De\psi^0(\x)-e\phi^0(\x)\psi^0(\x), 
~~~~\x\in T_d,
\\ 
\nonumber\\ 
-\De\phi^0(\x)&=&\rho^0(\x):= 
\si^0(\x)-e|\psi^0(\x)|^2, ~~~~~~\x\in T_d, 
\la{LPS4} 
\\ 
\nonumber\\ 
\la{LPS3g} 
0&=&- \5\langle\na\phi^0(\x), 
\rho^{\rm per}_j(\x-\x^0_j)\rangle, \qquad  j=1,\dots,N. 
\eeqn 
Here, $T_d:=\R^3/\Ga_d$ denotes  the `elementary  cell' of the crystal, 
$\langle\cdot,\cdot\rangle$ 
stands for the  scalar product in the Hilbert space $L^2(T_d)$ 
and its different extensions, and 
\be\la{rrr} 
\si^0(\x):= 
\sum_{j=1}^N\rho^{\rm per}_{j}(\x-\x^0_j),\quad \rho^{\rm per}_j(\x):= 
\sum_{\n\in\Z^d}\rho_j(\x-\x(\n)), 
\ee 
where we assume that the series converge in an appropriate sense. 
More precisely, we will construct a solution to the system (\re{LPS1w})--(\re{LPS3g}) with $\si^0(\x)$ given by 
the first equation of (\re{rrr}) where  $\rho^{\rm per}_j$ satisfy the following condition:
\be\la{roL1}
\mbox{\bf Condition I.}\qquad\rho^{\rm per}_j\in 
L^1(T_d)\cap 
L^2(T_d),\qquad j=1,...,N.\qquad\qquad
\ee 
Obviously, $\rho^{\rm per}_j\in  L^1(T_d)$ if $\rho_j\in  L^1(\R^3)$.
So we consider the case of smeared ions. The case of  point ions we
have considered in \ci{K2015}.
In the cases $d=2$ and $d=1$ we will assume additional conditions  (\re{rous5}) and (\re{rous51})
respectively.

The elementary cell $T_d$ is isomorphic to the 3D torus for $d=3$, 
to the direct product of the 2D torus by $\R$ for $d=2$, 
and to the 
direct product of the 1D torus (circle)  by $\R^2$ for $d=1$. 

The system (\re{LPS1w})--(\re{LPS3g}) is translation invariant similarly to (\re{LPS1g})--(\re{LPS3w}).
Let us note that 
$\om^0$ should be real since
$\rIm\om^0\ne 0$ means an instability of the ground state:
the decay 
as $t\to\infty$
in the case $\rIm\om^0 < 0$ and the explosion  if $\rIm\om^0 > 0$.

Let us denote $\ds Z_j:=\int_{T_d}\rho^{\rm per}_j(\x)d\x/e$. Then 
\be\la{intro}
\int_{T_d}\si^0(\x)d\x=eZ,\qquad Z:=\sum_j Z_j. 
\ee 
The total charge per cell vanishes (cf. \ci{BBL2003}): 
\be\la{neu10} 
\int_{T_d} \rho^0(\x)d\x= 
\int_{T_d} [\si^0(\x) -e|\psi^0(\x)|^2]d\x=0. 
\ee 
For $d=3$  this neutrality condition
follows directly from equation (\re{LPS4}) by integration using  
$\Ga_3$-periodicity of $\phi^0(\x)$. For  $d=1$ and $d=2$ 
we will deduce it  from the finiteness
of energy per cell (see (\ref{neu2}) and (\ref{neu21})).
Equivalently,
the neutrality condition can be written as the normalization
\be\la{neuZ} 
\int_{T_d} |\psi^0(\x)|^2d\x 
=Z. 
\ee 
We allow arbitrary 
$Z_j\in\R$, however we
assume that  $Z>0$: otherwise the theory is trivial.

\smallskip

Let us comment on our approach. 
The neutrality condition (\re{neuZ}) 
defines the submanifold $\mathcal M$ in the space 
$H^1(T_d)\times T_d^N$
of 
space-periodic configurations $(\psi^0,\ov\x^0)$. 
We construct a~ground state as a~minimizer 
over $\mathcal M$
of the energy per cell 
(\re{HamsT}), (\re{HamsT22}), (\re{HamsT221}).

Our techniques in the 
case of 3D lattice is rather standard, and we use it 
as an `Ariadne's thread' to manage the more complicated 
cases of  2D and 1D lattices,  because the corresponding elementary cells 
are unbounded.

Namely, the derivation of the equations \eqref{LPS1w}--\eqref{LPS3g} for the minimizer 
in the cases of  2D and 1D lattices
is not straightforward. The difficulty is that 
the energy per cell   is finite 
only 
on a dense subset of~$\mathcal M$ due to 
the infrared divergence of the Coulomb bond.
In these cases 
we restrict ourselves by one ion per cell, i.e., by $N=1$.
Then $\ov\x^0=\x^0_1$ can be chosen arbitrary
because of the translation invariance of the system  \eqref{LPS1w}--\eqref{LPS3g}.
Respectively, 
now the energy per cell should be minimized over 
$\psi\in M$, where $M$ is the submanifold of $H^1(T_d)$
defined by the neutrality condition (\re{neuZ}).

\smallskip 
The  main  novelties 
of our approach  behind the 
technical proofs for 2D and 1D lattices are as follows:

\smallskip 
 
I. The energy per cell consists of two 
contributions: the kinetic energy, and the 
Coulomb bond. 
Generally, the Coulomb bond for  2D and 1D lattices
is infinite due to 
the infrared divergence
which 
is caused by the continuous spectrum of the Laplace operator
on the corresponding elementary cells.
The spectrum is continuous since 
the elementary cells are unbounded in the case of  
2D and 1D lattices in $R^3$. 
Let us note that the continuous spectrum and the infrared 
singularity also appear in the Schr\"odinger--\allowbreak 
Poisson molecular systems in $\R^3$ studied in 
\ci{Ben02,Kawohl2012,Nier93} where the singularity 
is summable, contrary to the space-periodic  case. 

We indicate suitable conditions (\re{rous5}), (\re{rous51}) which provide 
the finiteness of the Coulomb bond for a dense set of the fields in the case
of 2D and 1D lattice respectively.

Both contributions to the energy per cell (the kinetic energy and the 
Coulomb bond) are nonnegative.
Hence, 
for any minimizing sequence, both contributions 
are bounded. 
The bound for the kinetic energy ensures the compactness in 
each finite region of a~cell 
by the Sobolev embedding theorem.
However, 
this bound cannot prevent the decay of the electron field, i.e., 
its escape to infinity. 
Nevertheless, 
the Coulomb interaction prevents even the partial escape  to infinity,
as we show in Lemma \re{lM2}.
Physically this means that 
the electrostatic potential of the remaining positive charge 
becomes confining. 
 
\smallskip 
 
II. We construct the
solution  to the Poisson equation 
\eqref{LPS4}  
as the contour integral, 
providing the continuity and a bound for the electrostatic potential.
The main difficulty is a verification of the Schr\"odinger 
equation \eqref{LPS1w} for the minimizer. Namely,
the Lagrange method of multipliers is not applicable
because
the energy per cell is infinite outside the submanifold 
$M\subset H^1(T_d)$ due to the infrared divergence
of the  Coulomb bond. 
Moreover, the Coulomb bond is infinite 
for a dense set of $\psi\in  M$. Hence, to differentiate 
the energy functional, we should construct the smooth paths in $ M$ 
lying outside this dense set.

\smallskip 
 
III. Finally, the proof that
$\om^0$ is real
(which is the stability condition for  the ground state)
is not straightforward
for  2D and 1D lattices, 
since the potential $\phi^0(x)$ a~priori can grow at infinity.
The correponding bounds for the potentials are given by 
\eqref{wfi2} and  \eqref{wfi1}.

\smallskip

The minimization strategy ensures the existence of a~ground state for any
 lattice~\eqref{Ga3}. One could expect
 that a stable lattice should provide a local minimum of
 the energy per  cell
 for fixed $d$, $N$ and  functions $\rho_j$, but this is still an open
problem.

\medskip 
 
Let us comment on related works. 
For atomic systems in $\R^3$, a~ground state was constructed by Lieb, Simon  and P.~Lions 
in the case of the Hartree and Hartree--\allowbreak Fock models
\ci{LS1977,Lions1981, Lions1987}, and 
by Nier
for the Schr\"odinger--\allowbreak Poisson  model \ci{Nier93}. 
The Hartree--\allowbreak Fock dynamics 
for molecular  systems in $\R^3$
has been constructed
by Canc\`es and Le Bris \ci{CB}.

A mathematical theory of the stability of  matter 
started from the pioneering works of 
Dyson, Lebowitz, Lenard, Lieb and others 
 for the Schr\"odinger many body model 
\ci{D1967, Lieb2005, LL1972, Lieb2009}; 
see the survey in~\ci{Lemm}. 
Recently, the theory was extended to the 
quantized Maxwell field \ci{LL2005}.

These results and methods were developed 
last two decades
by Blanc, Le Bris, Catto, P. Lions and others 
to justify 
the  thermodynamic limit for the Thomas--Fermi and Hartree--\allowbreak Fock 
models 
with space-periodic  ion arrangement 
\ci{BBL2007,CBL1996,CBL1998,CBL2001} 
and to construct the corresponding space-periodic ground states 
\ci{CBL2002}, 
see the survey and further references in \ci{BL2005}. 
 
Recently, Giuliani, Lebowitz and Lieb have established 
the periodicity of the thermodynamic limit 
in 1D local mean field model 
without the assumption of periodicity of the ion arrangement~\ci{GLL2007}. 
 
Canc\`es and others studied  short-range perturbations 
of the  Hartree--\allowbreak Fock 
model and proved that 
the 
density matrices of the perturbed and  unperturbed 
ground states 
differ by a~compact operator, 
\ci{CL2010,CS2012}. 
 

 \medskip

Let us note that 2D and 1D crystals in $\R^3$ were not considered
previously. 
The space-periodic ground states for 1D and 2D lattices 
 give the model of
 the nanostructures similar to 
 the carbon nanotubes and graphene respectively.

\medskip 
 
The plan of this chapter is as follows. 
In Section 2, we consider the $3$-dimensional lattice.
In Section 3, 
we construct a~ground state, derive equations 
\eqref{LPS1w}--\eqref{LPS3g} and 
study smoothness properties of a~ground state 
for $2$-dimensional lattice. 
In  Section 4, we consider the  $1$-dimensional lattice. 
Finally, in Appendix we construct and estimate the potential for 1D lattice.

\medskip

\setcounter{subsection}{0} 
\setcounter{theorem}{0} 
\setcounter{equation}{0} 
 
\section{3D lattice}

We consider the system  \eqref{LPS1w}--(\re{LPS3g}) 
for the corresponding functions on the torus  $T_3=\R^3/\Ga_3$ 
and with $\x^0_j\5{\rm mod}\5\Ga_3\in T_3$. 
For $s\in\R$, we denote by $H^s$ the complex Sobolev 
space on the torus $T_3$, and for $1\le p\le \infty$, we denote by $L^p$ 
the complex Lebesgue space of  functions on $T_3$. 
 
\subsection{Energy per cell} 
The ground state will be constructed by minimizing the 
energy  in the cell $T_3$. 
To this aim, we will minimize the energy with respect to 
$\ov\x:=(\x_1,\dots,\x_N)\in(T_3)^N$ and 
$\psi\in H^1$ satisfying the neutrality condition 
(\re{neu10}): 
\be\la{neu1} 
\int_{T_3} \rho(\x)d\x=0,~~~~~~~~~~\rho(\x):=\si(\x) -e|\psi(\x)|^2.
\ee 
where we set 
\be\la{rrr2} 
\si(\x):= 
\sum_{j}\rho^{\rm per}_{j}(\x-\x_j), 
\ee 
similarly to (\re{rrr}). Let us note that $\rho\in L^2$ for $\psi\in H^1$
by our condition  (\re{roL1}) since $\psi\in L^6$ by the Sobolev embedding theorem.

We define the energy in the periodic cell 
for $\psi\in H^1$ by 
\be\la{HamsT} 
E(\psi,\ov\x)\!:=\!\! 
\int_{T_3}\Bigl[ 
\fr{\h^2}{2\cm}|\na\psi(\x)|^2 
+ 
\fr12\phi(\x)\rho(\x)\Bigr]d\x ,\qquad \phi(\x):=(-\De)^{-1}\rho, 
\ee 
where 
$(-\De)^{-1}\rho$ is well-defined due to (\re{neu1}). 
Namely, consider the dual lattice 
\be\la{Ga3d} 
\Ga_3^*=\{\bk(\n)=\bb_1 n_1+\bb_2 n_2+\bb_3 n_3: \n=(n_1,n_2,n_3)\in\Z^3  \}, 
\ee 
where 
$\bb_k\bba_{k'}=2\pi\de_{kk'}$. 
Every function $\rho\in L^2$ admits the Fourier representation
\be\la{Fou} 
\rho(\x)=\fr1{\sqrt{|T_3|}}\sum_{\bk\in\Ga_3^*}\hat\rho(\bk) e^{-i\bk\x},\qquad 
\hat\rho(\bk)=\fr1{\sqrt{|T_3|}}\int e^{i\bk\x}\rho(\x)d\x.
\ee 
Respectively, we set 
\be\la{Fou2} 
\phi(\x)= 
(-\De)^{-1}\rho(\x) 
:= \fr1{\sqrt{|T_3|}}\sum_{\bk\in\Ga_3^*\setminus 0}\fr{\hat\rho(\bk)}{\bk^2} 
e^{-i\bk\x}. 
\ee 
This function 
$\phi\in H^2$
and satisfies 
the Poisson equation $-\De\phi=\rho$, 
since $\hat\rho(0)=0$ 
due to the neutrality condition (\re{neu1}).
Finally,
\be\la{Fou3} 
\int_{T_3}\phi(\x)d\x=0. 
\ee 
Now 
it is clear that the energy (\re{HamsT}) is finite for $\psi\in H^1$.
Let us rewrite the energy  as
\be\la{Upsi}
E(\psi,\ov\x)=I_1+I_2, 
\ee
where 
\beqn 
I_1(\psi)&:=&\fr{\h^2}{2\cm}\int_{T_3}|\na\psi(\x)|^2d\x\ge 0, 
\la{I1}\\ 
\nonumber\\ 
I_2(\phi)&:=& 
\fr12\int_{T_3} 
(-\De)^{-1}\rho(\x) 
\cdot\rho(\x) 
d\x=\fr12 \int_{T_3} 
|\na\phi(\x)|^2 
d\x\ge 0. \la{I2}
\eeqn 
The functional (\re{HamsT}) is chosen,  because  
\be\la{ener1} 
\fr{\de E}{\de\x_j}=
- \langle 
(-\De)^{-1}\rho(\x),\na\rho^{\rm per}_j(\x-\x_j)\rangle= 
 \langle\na\phi(\x),\rho^{\rm per}_j(\x-\x_j)\rangle, 
\ee
and the variational derivatives {\it formally} reads 
\be
\fr{\de E}{\de\Psi(\x)}= 
-2\fr{\h^2}{2\cm} 
\De\psi-
2 e 
(-\De)^{-1}\rho(\x)\psi(\x)= 
-2\fr{\h^2}{2\cm} 
\De\psi-
2e 
\phi(\x)\psi(\x).  \la{enerid} 
\ee
The variation in  (\re{enerid})  is taken 
over $\Psi(\x)=(\psi_1(\x),\psi_2(\x))\in L^2(T_3,\R^2)$, 
where 
$\psi_1(\x)=\rRe\psi(x)$ and $\psi_2(\x)=\rIm\psi(x)$. 
Respectively, all the terms in  (\re{enerid})  are identified with 
the corresponding 
$\R^2$-valued distributions.

\subsection{Compactness of minimizing sequence}

Our purpose here is to minimize the energy with respect to 
$$ 
(\psi,\ov\x)\in \mathcal M:=M\times T_3^N, 
$$ 
where $M$ denotes the manifold (cf. (\re{neuZ}))
\be\la{MZ} 
M=\{\psi\in H^1:~ \int_{T_3} |\psi(\x)|^2d\x 
=Z \}. 
\ee 
The energy is bounded from below since   
$
E(\psi,\ov\x)\ge 0 
$ 
by (\re{Upsi})-(\re{I2}).
We choose a minimizing sequence $(\psi^n, \ov\x^n)\in\mathcal M$ such that 
\be\la{min} 
E(\psi^n,\ov\x^n)\to E^0:=\inf_\mathcal M~ 
E(\psi,\ov\x), \qquad n\to\infty. 
\ee 
Our main result for a~3D lattice is the following: 
\bt\la{t3} 
Let condition \eqref{roL1} hold. Then
\medskip\\
i) There exists  $(\psi^0,\ov\x^0)\in \mathcal M$ with
\be\la{U0min} 
E(\psi^0,\ov\x^0)= E^0. 
\ee 
ii) Moreover,  $\psi^0\in H^2$ and satisfies equations \eqref{LPS1w}--\eqref{LPS3g} with $d=3$,
where the  potential $\phi^0\in H^2$ is real, 
 and  $\om^0\in\R$. 

\et 
To prove item i),
let us denote 
\be\la{ron}
\rho^n(\x):=\si^n(\x)-e|\psi^n(\x)|^2,\qquad\qquad \si^n(\x):=\sum_{j}\rho^{\rm per}_{j}(\x-\x_{jn}).
\ee
Now the sequence $\psi^n$ and the corresponding sequence 
$\phi^n:=(-\De)^{-1}\rho^n$ 
are bounded in $H^1$ 
by (\re{Upsi})-(\re{I2}), (\re{Fou3}) and 
(\re{MZ})-(\re{min}). 
Hence, both sequences are precompact in~$L^p$ 
for any $p\in[1,6)$ by the Sobolev embedding theorem \ci{Adams,Sob}. 
Therefore, the sequence $\rho^n$ is precompact in $L^2$ by our assumption (\re{roL1}), and respectively, 
 the sequence $\phi^n$ is precompact in $H^2$.
As the result,
there exist a  subsequence $n'\to\infty$ for which  
\be\la{subs} 
\psi^{n'}\toLp\psi^0, \qquad 
\phi^{n'}\toHd\phi^0, \qquad
\ov\x^{n'}\to \ov\x^0 , \qquad \qquad n'\to\infty
\ee 
with any $p\in[1,6)$.
Respectively,
\be\la{subs1} 
\si^{n'}\toLd \si^0,  \qquad \rho^{n'}\toLd \rho^0, \qquad \qquad n'\to\infty,
\ee 
where $\si^0(\x)$ and $\rho^0(\x)$ are defined by (\re{rrr}) and (\re{LPS4}). 
Hence, the neutrality condition (\re{neu10}) holds, 
$(\psi^0,\ov\x^0)\in \mathcal M$, $\phi^0\in H^2$, and for these limit functions we have
\be\la{phi0}
-\De\phi^0=\rho^0, \qquad \ds\int_{{T_3}}\phi^0(\x)d\x=0.
\ee 
To prove identity (\re{U0min}), we take 
into account that 
$I_1(\psi)$ is lower semicontinuous on $L^2$, while
$I_2(\phi)$ is continuous on $H^2$; i.e., 
\be\la{sem} 
I_1(\psi^0)\le \liminf_{n'\to\infty} I_1(\psi^{n'}), 
\qquad I_2(\phi^0)= \lim_{n'\to\infty} I_2(\phi^{n'}). 
\ee 
These limits, together with (\re{min}),  imply that 
\be\la{sem2} 
E(\psi^0,\ov\x^0) 
=I_1(\psi^0)+I_2(\phi^0) 
\le E^0. 
\ee 
Now (\re{U0min}) follows from the definition of $E^0$, since 
$(\psi^0,\ov\x^0)\in\mathcal M$. Thus Theorem \re{t3} i) is proved.
\medskip

We will prove the item ii) in next sections.

\subsection{Variation of the energy} 

Theorem \re{t3} ii) follows from next proposition.

\bp\la{tgs23} The limit functions 
\eqref{subs}
satisfy equations {\rm \eqref{LPS1w}--\eqref{LPS3g}} with $d=3$
and $\om^0\in\R$.

\ep

Equation (\re{LPS4}) is proved in (\re{phi0}), and the equation
 \eqref{LPS3g} follows from (\re{ener1}) and (\re{U0min}). 
 It remains to prove the Schr\"odinger equation \eqref{LPS1w}. 
Let us denote $\E(\psi):=E(\psi,\ov\x^0)$.
We derive \eqref{LPS1w} in next sections, 
equating the variation of $\E(\cdot)|_{M}$ to zero at $\psi=\psi^0$. 
In this section we calculate the corresponding G\^ateaux variational derivative.

We should work directly on~$M$ introducing an atlas in a~neighborhood 
of $\psi^0$ in $M$.
We define the atlas 
as the stereographic projection from the tangent plane 
$TM(\psi^0)=(\psi^0)^\bot:=\{\psi\in H^1: 
\langle\psi,\psi^0\rangle=0\}$ 
to the sphere (\re{MZ}): 
\be\la{3atlas} 
\psi_\tau= \fr{\psi^0+\tau~~~}{\Vert\psi^0+\tau\Vert_{L^2}} 
\sqrt{Z}, \qquad  \tau\in (\psi^0)^\bot. 
\ee 
Obviously, 
\be\la{3tau} 
\fr d{d\ve}\Big|_{\ve=0} \psi_{\ve\tau}=\tau,\qquad \tau\in (\psi^0)^\bot, 
\ee 
where the derivative exists in $H^1$. We define
the `G\^ateaux derivative' of $\E(\cdot)|_{M}$  as 
\be\la{3Gder} 
D_\tau \E(\psi^0):=\lim_{\ve\to 0}\fr{\E(\psi_{\ve\tau})-\E(\psi^0)}{\ve},
\ee 
if this limit exists. We should restrict the set of allowed tangent vectors~$\tau$. 
 
\bd 
$\ccT^0$ is the space of test functions 
$\tau\in(\psi^0)^\bot\cap C^\infty(T_3)$. 
\ed 
Obviously, $\ccT^0$ is dense in $(\psi^0)^\bot$ in the norm of  $H^1$. 
 Let us rewrite the energy (\re{HamsT}) as
 \be\la{3HamsT22} 
\E(\psi):= 
\int_{T_3}\Bigl[ 
\fr{\h^2}{2\cm}|\na\psi(\x)|^2 
+ 
\fr12 |\Lam\rho(\x)|^2 
\Bigr]d\x, \qquad \rho(\x):=\si(\x) -e|\psi(\x)|^2, 
\ee 
 where $\Lam:=(-\De)^{-1/2}$ is defined similarly to (\re{Fou2}):
 \be\la{Fou2L} 
\Lam\rho(\x) 
:= \fr1{\sqrt{|T_3|}}\sum_{\bk\in\Ga_3^*\setminus 0}\fr{\hat\rho(\bk)}{|\bk|} 
e^{-i\bk\x}\in L^2 \qquad \mbox{\rm for } \qquad  \rho\in L^2. 
\ee

\bl\la{3lvar} 
Let 
$\tau\in\ccT^0$. Then the  derivative \eqref{3Gder} exists, and  {\rm (}cf. \eqref{enerid}{\rm )}, 
\be\la{3Gder2} 
D_\tau \E(\psi^0)=\int_{T_3}  \Big[\fr{\h^2}{2\cm} 
(\na\tau\ov{\na\psi^0}+\na\psi^0\ov{\na \tau}) 
-e {\Lam\rho^0}\Lam(\tau\ov{\psi^0}+\psi^0\ov \tau )\Big]d\x. 
\ee 
\el 
\begin{proof} 
Let us denote $\rho_{\ve\tau}(\x):=\si^0(\x)-e|\psi_{\ve\tau}(\x)|^2$.

\bl\la{3lL2}
 For $\tau\in\ccT^0$ we have
 \be\la{3Gder3} 
D_\tau \Lam\rho:= 
\lim_{\ve\to 0}\fr{\Lam\rho_{\ve\tau}-\Lam\rho^0}{\ve} 
 =e\Lam(\tau\ov{\psi^0}+\psi^0\ov \tau ), 
\ee 
where the limit converges in  $L^2$.

\el
\begin{proof} In the polar coordinates
\be\la{3al} 
\psi_{\ve\tau}=(\psi^0+\ve\tau)\cos\al,\qquad 
\al=\al(\ve)=\arctan\fr{\ve\Vert\tau\Vert_{L^2}}{\Vert\psi^0\Vert_{L^2}}. 
\ee 
Hence,
\beqn\la{3rod} 
\Lam\rho_{\ve\tau}&=&\Lam\si^0 -e\cos^2\al\Lam|\psi^0+\ve\tau|^2 
\nonumber\\
\nonumber\smallskip\\
&=& 
\Lam\rho^0-e\ve\cos^2\al\Lam(\tau\ov{\psi^0}+\psi^0\ov\tau)+ 
e\Lam[\ve^2|\tau|^2\cos^2\al-|\psi^0|^2\sin^2\al]. 
\eeqn 
Here 
$\Lam\rho^0\in L^2$ since 
$\rho^0\in L^2$, and similarly $\Lam [\psi^0\ov\tau]\in L^2$ since 
$\psi^0 \ov \tau\in L^2$.
It remains to 
estimate 
 the last term 
of (\re{3rod}),
\be\la{3lt}
R_\ve:=\Lam[\ve^2|\tau|^2\cos^2\al-|\psi^0|^2\sin^2\al].
\ee
Here $|\psi^0|^2\in L^2$ 
since $\psi^0\in H^1\subset L^6$. Finally,  $|\tau|^2\in L^2$ and $\sin^2\al\sim \ve^2$.
Hence, the convergence (\re{3Gder3}) holds in $L^2$.\end{proof}
\medskip

Now (\re{3Gder2}) follows by differentiation 
in $\ve$ of (\re{3HamsT22}) with $\psi=\psi_{\ve\tau}$ and $\rho=\rho_{\ve\tau}$.
\end{proof}

 \subsection{The variational identity} 
 
Since $\psi^0$ is a minimal point, the G\^ateaux derivative (\re{3Gder2}) 
vanishes: 
\be\la{3GaD} 
\int_{T_2}  \Big[\fr{\h^2}{2\cm} 
(\na\tau\ov{\na\psi^0}+\na\psi^0\ov{\na \tau}) 
-
e {\Lam\rho^0}\Lam(\tau\ov{\psi^0}+\psi^0\ov \tau )
\Big]d\x=0. 
\ee 
Substituting $i\tau$ instead of $\tau$ in this identity 
and subtracting, we obtain 
\be\la{3GaD2} 
-\fr{\h^2}{2\cm}\langle\De\psi^0, \tau\rangle 
- e \langle 
\Lam\rho^0,\Lam(\ov{\psi^0} \tau) \rangle=0. 
\ee 
Next step 
we should evaluate the ``nonlinear'' term.

\bl For the limit functions 
\eqref{subs}--\eqref{subs1} we have 
\be\la{3sp} 
\langle 
\Lam\rho^0,\Lam(\ov{\psi^0} \tau) \rangle= 
\langle\phi^0\psi^0, \tau \rangle,\qquad \tau\in\ccT^0.
\ee

\el 
\begin{proof} 
Let us substitute 
$\rho^0=-\De\phi^0$. 
 Then, by the Parseval--\allowbreak Plancherel identity, 
\be\la{3GaD3} 
\langle 
\Lam\rho^0,\Lam(\ov{\psi^0} \tau) \rangle= \sum_{\bk\in\Ga_3^*\setminus 0}
\fr{\bk^2\hat\phi^0(\bk)}{|\bk|}\cdot
\fr{\widehat{\ov{\psi^0} \tau}(\bk)}{|\bk|} = 
\langle\hat \phi^0,\widehat{\ov{\psi^0} \tau} \rangle
=\langle \phi^0,\ov{\psi^0} \tau \rangle=\langle \phi^0\psi^0, \tau \rangle. 
\ee 
which proves (\re{3sp}).\end{proof} 
\medskip 
 
Using (\re{3sp}), we can rewrite (\re{3GaD2}) as the variational identity (cf.\ (\re{enerid})) 
\be\la{3GaD22} 
\langle -\fr{\h^2}{2\cm}\De\psi^0 
- e 
\phi^0\psi^0,\tau \rangle=0,    \qquad \tau\in\ccT^0. 
\ee 
 
\subsection{The Schr\"odinger equation} 
 
Now we prove the Schr\"odinger equation (\re{LPS1w}) with $d=3$.
\bl\la{lse}
$\psi^0$ is the eigenfunction of the 
Schr\"odinger operator $H=-\fr{\h^2}{2\cm}\De+e\phi^0$:
\be\la{3Hpsi} 
H\psi^0=\lam\psi^0,
\ee 
where 
$\lam\in\R$. 
\el
\begin{proof}
First, $H\psi^0$ is a well-defined distribution since $\phi^0\in H^2\subset C(T_3)$ by (\re{subs}).
Second,
 $\psi^0\ne 0$ since $\psi^0\in M$ and $Z>0$. Hence, there exists a~test function 
$\theta\in C^\infty(T_3)\setminus\ccT^0$, i.e., 
\be\la{3test} 
\langle\psi^0,\theta\rangle\ne 0. 
\ee 
Then 
\be\la{3test2} 
\langle(H-\lam)\psi^0,\theta\rangle= 0. 
\ee 
for an appropriate $\lam\in\C$. However, 
$(H-\lam)\psi^0$ also  annihilates $\ccT^0$ by (\re{3GaD22}), 
hence it  annihilates the whole space 
$C^\infty(T_3)$. This implies (\re{3Hpsi}) in the sense of 
distributions with a $\lam\in\C$. Finally,
the potential is  real, and 
 $\phi^0\in C(T_3)$. Hence,
 $\lam\in\R$.\end{proof}
 \medskip

  This lemma implies  equation (\re{LPS1w})  with $\hbar\om^0=\lam$.
   Hence, $\psi^0\in H^2$ since $\phi^0\in C(T_3)$.
Now Theorem \re{t3} ii) is proved.

\subsection{Smoothness of  ground state} 

We have proved that $\psi^0\in H^2$ under condition (\re{roL1}).
Using the Schr\"odinger equation (\re{3Hpsi})
we can improve further the smoothness
 of $\psi^0$ strengthening  the condition (\re{roL1}).
Namely, let us assume that
\be\la{rpem3} 
\rho_j^{\rm per}\in C^\infty(T_3), \qquad j=1,...,N.
\ee 
Then also
\be\la{rper} 
\si^0(\x):=\sum_{j=1}^N \rho^{\rm per}_j(\x-\x^0_j)\in C^\infty(T_3). 
\ee 
For example, (\re{rpem3}) and (\re{rper}) hold if $\rho_j\in  \cS(\R^3)$,
where $\cS(\R^3)$ is the Schwartz space of test functions.

\bl\la{lpp} 
Let  condition \eqref{rpem3} hold, and 
$\psi^0\in H^2$,  $\phi^0\in H^2$ be a solution to equations 
\eqref{LPS1w}--\eqref{LPS3g} with $d=3$ and some $\x\in T_3^N$. Then 
the functions $\psi^0$ and $\phi^0$ are smooth. 
\el 
\begin{proof} 
First, $\phi^0\psi^0\in H^2$
since  $H^s$ is the algebra for $s>3/2$. 
Hence, equation ~\eqref{LPS1w} implies that 
\be\la{HC2} 
\psi^0\in H^4\subset C^2(T_3). 
\ee 
Now 
$\rho^0:=\si^0+e|\psi^0|^2\in H^4$ by \eqref{rpem3}. Then~\eqref{LPS4} implies that $\phi^0\in H^6\subset C^4(T_3)$.
Hence, $\phi^0\psi^0\in H^4$, $\psi^0\in H^6$, $\rho^0\in H^6$, etc.\end{proof}

\setcounter{equation}{0} 
 
\section{2D lattice} 
 
For simplicity of notation we will consider the 2D lattice $\Ga_2=\Z^2$
 and 
construct a solution to system   \eqref{LPS1w}--\eqref{LPS3g} 
for the corresponding functions on the `cylindrical cell' 
 $T_2:=\R^3/\Ga_2=\T^2\times\R$ with  
 the coordinates $\x=(x_1,x_2,x_3)$, where $(x_1,x_2)\in \T^2$
and $x_3\in\R$. 
Now we denote by $H^s$ the complex Sobolev space on $T_2$, and 
by $L^p$, the  complex Lebesgue space of functions on $T_2$.

We will construct a~ground state by minimizing the  
energy  (\re{HamsT}), 
where the integral is extended 
over $T_2$ instead of $T_3$. 
The neutrality condition 
of type (\re{neu1})
holds for 
$\Ga_2$-periodic states with finite energy, as we show below.

\subsection{The energy per cell} 
We restrict ourselves by $N=1$, so $\ov\x^0=\x^0_1$ can be chosen arbitrary
because of the translation invariance of the system  \eqref{LPS1w}--\eqref{LPS3g}.
For example, we can set $\x^0_1=0$.

The  energy in the cylindrical cell $T_2$
is defined similarly to  (\re{HamsT}), which we rewrite as  (\re{3HamsT22}):
\be\la{HamsT22} 
\E(\psi):= 
\int_{T_2}\Bigl[ 
\fr{\h^2}{2\cm}|\na\psi(\x)|^2 
+ 
\fr12 |\Lam\rho(\x)|^2 
\Bigr]d\x, \qquad \rho(\x):=\si^0(\x) -e|\psi(\x)|^2. 
\ee 
Here $\si^0(\x)$ is defined by 
 (\re{rrr2}) with $N=1$ and $\x^0_1=0$:
\be\la{rrr22} 
\si^0= 
\rho^{\rm per}_1\in L^1\cap L^2
\ee 
according to our condition (\re{roL1}).
Hence, we have 
\be\la{intro2} 
\ds\int_{T_2} 
\si^0(\x)d\x=eZ_1,~~~~~Z_1>0. 
\ee  
Further, $\Lam$ is the operator $(-\De)^{-1/2}$ defined by 
the Fourier transform. Namely, we denote
 $\Ga_2^*=2\pi\Ga_2$, 
and define the Fourier representation for the test functions 
$\vp\in C_0^\infty(T_2)$ by 
\be\la{FR2} 
\vp(\x)= 
\fr 1{\sqrt{2\pi}} \sum_{\bk\in \Ga_2^*} e^{-i(\bk_1x_1+\bk_2x_2)} 
\int_\R e^{-i\xi x_3}\hat \vp(\bk,\xi)d\xi, \qquad \x\in T_2,
\ee 
where 
\be 
\la{FT2} 
\hat 
\vp(\bk,\xi)=F\vp(\bk,\xi)=\fr 1{\sqrt{2\pi}} 
\int_{T_2} e^{i(\bk_1x_1+\bk_2x_2+\xi x_3)} \vp(\x)d\x,
\qquad  (\bk,\xi)\in\Si_2:=\Ga_2^*\times \R.
\ee 
The operator $\Lam$ is defined for $\vp\in L^1\cap L^2$ by 
\be\la{Lam} 
\Lam \vp=F^{-1} \fr{\hat \vp(\bk,\xi)}{\sqrt{\bk^2+\xi^2}} 
\ee 
provided the quotient belongs to $L^2(\Si_2)$. In this case
\be\la{02} 
\hat \vp(0,0)=0. 
\ee 
 Let us note that $\rho\in L^1\cap L^2$ for $\psi\in H^1$
by our condition  (\re{roL1}) since $\psi\in L^p$ with 
$p\in [2,6]$ 
by the Sobolev embedding theorem.
For $\psi\in H^1$ 
with finite energy (\re{HamsT22}) 
we have $\Lam\rho\in L^2(\Si_2)$. Therefore,
(\re{02}) with  $\vp=\rho$ implies 
the neutrality condition 
(\re{neu10}) with $d=2$: 
\be\la{neu2} 
\hat\rho(0,0)=\int_{T_2} \rho(\x)d\x= 
\int_{T_2} [\si^0(\x) -e|\psi(\x)|^2]d\x=0. 
\ee 
In other words, the finiteness of the Coulomb energy 
$\Vert\Lam\rho\Vert^2$ prevents the electron charge from escaping to infinity, 
as mentioned in Introduction. 
Now (\re{intro2}) gives
\be\la{M2} 
\int_{T_2}
|\psi(\x)|^2d\x=Z_1.
\ee 

 \bd 
$M_2$ denotes the set of  
$\psi\in H^1$ satisfying the neutrality condition \eqref{M2}. 
\ed 
It is important that the energy be finite for a nonempty set of $\psi\in H^1$.
To find the corresponding condition, let us rewrite the 
energy (\re{HamsT22})
using the Parseval-Plancherel identity:
\be\la{HamsT22F} 
\E(\psi)= 
\sum 
_{\bk\in \Ga_2^*} 
\fr{\h^2}{2\cm}
\int_\R(\bk^2+\xi^2)|\hat\psi(\bk,\xi)|^2 d\xi
+ 
\fr12 
\sum 
_{\bk\in \Ga_2^*} 
\int_\R
\fr {|\hat\rho(\bk,\xi)|^2} {\bk^2+\xi^2} d\xi. 
\ee 
 Here the first term on the right hand side is finite for all $\psi\in H^1$.
The second term is finite up to the infrared divergence at 
the point $(\bk,\xi)=(0,0)$
since $\rho\in L^2(\Si_2)$ for $\psi\in H^1$.

 We note that (\re{intro2}) can be written as
$\hat\rho^{\rm per}_1(0)-eZ_1=0$. 
We will assume that moreover,
\be\la{rous5} 
\mbox{\bf Condition II.}\qquad
\fr{\hat\rho^{\rm per}_1(0,\xi)-eZ_1}{|\xi|}\in L^2(-1,1). \qquad\qquad\qquad\qquad
\ee 
For example, this condition holds, 
provided that  
\be\la{rous5if} 
\int_{T_2}  |x^3||\rho^{\rm per}_1(\x)|d\x <\infty.
\ee 

\bl\la{lne}
Let conditions \eqref{roL1} and \eqref{rous5} hold, $N=1$ and $\x^0_1\in T_2$. Then 
the energy \eqref{HamsT22F} is finite for a dense set of   $\psi\in H^1$.

\el
\begin{proof} By definition, $\hat \rho(0,\xi)=\hat\rho^{\rm per}_1(0,\xi)-e\hat P(0,\xi)$, where $P(\x):=|\psi(\x)|^2$.
Hence, (\re{rous5}) implies that
the energy (\re{HamsT22F}) is finite for   $\psi\in M_2$ with finite momenta 
$
\ds\int_{T_2}  |x^3|~|\psi(\x)|^2d\x<\infty.
$
\end{proof}

\medskip 
\subsection{Compactness of minimizing sequence} 
 
Similarly to the 3D case, the energy is nonnegative, and 
we choose a minimizing sequence $\psi^n\in  M_2$ 
such that 
\be\la{min2} 
\E(\psi^n)\to \E^0:=\inf_{M_2}~ 
\E(\psi),\qquad n\to\infty. 
\ee 

The second main result of this chapter is the following.
 
\bt\la{tgs2} 
Let conditions \eqref{roL1} and \eqref{rous5} hold, and $N=1$. Then 
\medskip\\
i) There exists $\psi^0\in M_2$ with
\be\la{U0min2} 
\E(\psi^0)= \E^0. 
\ee 
ii) Moreover,  $\psi^0\in H^2_{\rm loc}(T_2)$ and  satisfies equations \eqref{LPS1w}--\eqref{LPS3g} with $d=2$,
where 
the potential  $\phi^0\in H^2_{\rm loc}(T_2)$ is real, $\x^0_1=0$,
 and  $\om^0\in\R$. 
\medskip\\
iii) The following bound holds
\be\la{wfi2} 
|\phi^0(\x)|\le 
C(1+|x_3|)^{1/2},\qquad\qquad \x\in T_2.
\ee

\et
To prove item i),
let us note 
that the sequence $\psi^n$  
is bounded in $H^1$
due to 
(\re{HamsT22}),
(\re{M2}) and (\re{min2}).
Hence,  by the Sobolev embedding theorem \ci{Adams,Sob}, 
the sequence $\psi^n$  
is bounded in $L^p$ with each $p\in[2,6)$ and
compact in $L^p_R:=L^p(T_2(R))$ 
for  any $R>0$, 
where $T_2(R)=\{\x\in T_2:~|x_3|<R\}$. 
Therefore, there exists a~subsequence 
\be\la{subs2} 
\psi^{n'}\toLpR\psi^0,\qquad \rho^{n'}:=\rho^{\rm per}_1-e|\psi^{n'}|^2\toLLR\rho^0,
\qquad 
n'\to\infty,\quad \forall R>0,
\ee 
since 
$\rho^{\rm per}_1\in L^1\cap L^2$ by (\re{roL1}).
Hence, $\psi^0\in H^1\cap L^p$,  and
\be 
\la{pr} 
\rho^0(\x)=\rho^{\rm per}_1(\x)-e|\psi^0(\x)|^2\in L^1\cap L^2. 
\ee 
Next problem is to  check the neutrality condition (\re{M2}) for the 
limit charge density $\rho^0$ since the convergence (\re{subs2}) itself is not sufficient.

\bl\la{lM2} The limit function
$\psi^0\in M_2$, 
and the energy \eqref{HamsT22} for $\psi^0$ is finite. 
\el
\begin{proof}
Let us prove that
\be\la{UU} 
\E(\psi^0)\le \E^0. 
\ee 
Indeed, 
(\re{HamsT22F}) with $\psi=\psi^{n'}$  reads 
\be \la{En}
\E(\psi^{n'}):= 
\Bigl\langle 
\fr{\h^2}{2\cm}|f^{n'}(\bk,\xi)|^2 
+ 
\fr 12 |g^{n'}(\bk,\xi)|^2 
\Bigr\rangle_{\Si_2}, 
\ee 
where $\langle \dots \rangle_{\Si_2}$ stands for 
$\ds\sum_{\bk\in \Ga_2^*}\int_\R~\dots~d\xi$ and 
$$ 
f^{n'}(\bk,\xi):=\sqrt{\bk^2+\xi^2} 
\hat\psi^{n'}(\bk,\xi),\qquad g^{n'}(\bk,\xi):=\fr {\hat\rho^{n'}(\bk,\xi)} 
{\sqrt{\bk^2+\xi^2}}. 
$$ 
The  functions  $\hat\psi^{n'}$ and $\hat\rho^{n'}$
are bounded in $L^2(\Si_2)$, and are 
converging in the sense of distributions due to (\re{subs2}). Hence,
\be\la{tow2} 
\hat\psi^{n'}\tow \hat\psi^0,\qquad\hat\rho^{n'}\tow \hat\rho^0,
\qquad  n'\to\infty.
\ee
Similarly, the  functions $f^{n'}$ and $g^{n'}$ 
are bounded in $L^2(\Si_2)$ by (\re{En}), (\re{min2}), and are
converging in the sense of distributions due to (\re{tow2}).
Therefore, 
\be\la{tow} 
f^{n'}\tow f^0,\qquad g^{n'}\tow g^0,\qquad n'\to\infty. 
\ee 
Hence, for the limit functions, 
$$ 
f^0(\bk,\xi)=\sqrt{\bk^2+\xi^2} 
\hat\psi^0(\bk,\xi), 
\qquad g^0(\bk,\xi)=\fr {\hat\rho^0(\bk,\xi)} 
{\sqrt{\bk^2+\xi^2}}, \qquad\qquad  a.a. \,\,\, (\bk,\xi)\in\Si_2.
$$ 
Therefore, (\re{UU}) holds since  
\be\la{Uro} 
\E(\psi^0)= 
\Bigl\langle 
\fr{\h^2}{2\cm}|f^0(\bk,\xi)|^2 
+ 
\fr 12 |g^0(\bk,\xi)|^2 
\Bigr\rangle_{\Si_2}\le \E^0 
\ee 
by the week convergence (\re{tow}). 
In particular,  
\be\la{Lamro}
\Lam\rho^0\in L^2.
\ee
Therefore, $\hat\rho^0(0,0)=0$ as in (\re{neu2}) since $\rho^0\in L^1$ by (\re{pr}).
Hence,  $\psi^0\in  M_2$.
\end{proof}

Now (\re{UU}) implies 
(\re{U0min2}). Thus Theorem \re{tgs2} i) is proved.

\subsection{The Poisson equation} 
Our aim here is to construct the potential which is the solution to 
the Poisson equation 
(\re{LPS4}) with $d=2$. It suffices to solve the equation 
\be\la{rho2} 
\na\phi^0(\x)=G^0(\x),\qquad\qquad \x\in T_2,
\ee
where 
$G^0(\x):= 
- 
 iF^{-1}\fr{(\bk,\xi)}{\bk^2+\xi^2}\hat\rho^0(\bk,\xi)$ 
  is a~real vector field, 
$G^0\in L^2\otimes \R^3$ by (\re{Lamro}), and $\rot G^0(\x)\equiv 0$.

\bl\la{lP2}
The equation \eqref{rho2} 
admits real  solution $\phi^0\in H^2_{\rm loc}(T_2)$
which is unique up to an additive constant, and satisfies the bound \eqref{wfi2}.

\el
\begin{proof}  The uniqueness up to constant is obvious. If the solution exists, then 
$\phi^0\in H^2_{\rm loc}(T_2)$  by (\re{pr}). Local solutions exist since  $\rot\, G^0(\x)\equiv 0$. However, the 
existence of the global solution is not obvious since the cell $T_2$ is not 1-connected.

We will prove the existence using the following arguments. 
Formally $\phi^0(x)=F^{-1}\fr{\hat\rho^0(\bk,\xi)}{\bk^2+\xi^2}$. However, the last expression is not correctly defined distribution
in the neighborhood of the point $(0,0)$. To avoid this infrared divergence, we split $\hat\rho^0=\hat\rho_1+\hat\rho_2$ where 
\be\la{splitt}
 \hat \rho_1(\bk,\xi)=\left\{
 \ba{ll}
 \hat \rho^0(0,\xi),& \bk=0,~|\xi|<1,\\
 0,& \mbox{\rm otherwise.}
 \ea\right.
\ee
Respectively, 
$G^0=G_1+G_2$, and
the solution $\phi^0=\phi_1+\phi_2$. Obviously,
\be\la{G1}
G_1(\x)=- 
 iF^{-1}\fr{(0,\xi)}{\xi^2}\hat\rho_1(0,\xi)=\e_3g_1(x_3),\qquad \e_3:=(0,0,1),
\ee
and $g_1(x_3)$ is a smooth function. Moreover, 
(\re{pr}) implies that $g_1(x_3)$
 is the real function, and $g_1\in L^2(\R)$ since $G^0\in L^2\otimes \R^3$.
Hence, the solution $\phi_1(\x)=\ds\int_0^{x_3} g_1(s)ds$ is smooth and  continuous, 
and depends on $x_3$ only.
The bound (\re{wfi2}) for $\phi_1$ follows by the Cauchy-Schwartz inequality.

The second solution is given by 
$\phi_2(\x)=F^{-1}\fr{\hat\rho_2(\bk,\xi)}{\bk^2+\xi^2}$, where  $\hat\rho_2\in L^2(\Si_2)$ by (\re{pr}). Moreover, $\hat\rho_2(0,\xi)=0$ for $|\xi|<1$, and hence
$\phi_2\in H^2$.\end{proof}
\medskip
\brs
i) The function $\phi^0(\x)=(1+|x_3|)^{1/2-\ve}$ with $\ve>0$ 
shows that the bound \eqref{wfi2} is exact under the condition $\na\phi^0\in L^2$. 
Note that the potential of uniformly charged plane grows linearly with the distance. 
\medskip\\
ii)
In the Fourier transform, \eqref{rho2} implies that 
\be\la{fm2} 
(\bk,\xi)\hat\phi^0(\bk,\xi)\in L^2(\Si_2)\otimes\C^3. 
\ee 
 \ers

\subsection{Variation of the energy} 

Theorem \re{tgs2} ii) follows from next proposition.

 \bp\la{tgs22} 
The functions $\psi^0$, $\phi^0$ 
satisfy equations {\rm \eqref{LPS1w}--\eqref{LPS3g}} with $d=2$
and $\om^0\in\R$.

\ep

The equation (\re{LPS4}) is
 proved above, and the equation
 \eqref{LPS3g} follows from (\re{ener1}) and (\re{U0min2})
 by the translation invariance of the energy.
 It remains to prove the Schr\"odinger equation \eqref{LPS1w}. 
We are going to derive \eqref{LPS1w}, 
equating the variation of $\E(\psi)|_{M_2}$ to zero at $\psi=\psi^0$. 
 In this section we calculate the corresponding G\^ateaux variational derivative.

 Similarly to (\re{3atlas}),
we define the atlas in a~neighborhood 
of $\psi^0$ in $M_2$ 
as the stereographic projection from the tangent plane 
$TM_2(\psi^0)=(\psi^0)^\bot:=\{\psi\in H^1: 
\langle\psi,\psi^0\rangle=0\}$ 
to the sphere (\re{M2}): 
\be\la{atlas} 
\psi_\tau= \fr{\psi^0+\tau~~~}{\Vert\psi^0+\tau\Vert_{L^2}} 
\sqrt{Z_1}, \qquad  \tau\in (\psi^0)^\bot. 
\ee

\bd 
$\ccT^0$ is the space of test functions 
$\tau\in(\psi^0)^\bot\cap C_0^\infty(T_2)$. 
\ed 
 
Obviously, $\ccT^0$ is dense in $(\psi^0)^\bot$ in the norm of  $H^1$. 
 
\bl\la{lvar} 
Let $\tau\in\ccT^0$. Then 
 \medskip\\
{\rm i)} 
The energy 
$\E(\psi_{\ve\tau})$ is finite for $\ve\in\R$. 
 \medskip\\
{\rm ii)} The  G\^ateaux derivative \eqref{3Gder} exists, and  similarly to  \eqref{3Gder2},
\be\la{Gder2} 
D_\tau \E(\psi^0)=\int_{T_2}  \Big[\fr{\h^2}{2\cm} 
(\na\tau\ov{\na\psi^0}+\na\psi^0\ov{\na \tau}) 
-e {\Lam\rho^0}\Lam(\tau\ov{\psi^0}+\psi^0\ov \tau )\Big]d\x. 
\ee 
\el 
\begin{proof} 
i)  We should prove the bound 
\be\la{HamsT223} 
\E(\psi_{\ve\tau}):= 
\fr{\h^2}{2\cm}\int_{T_2}|\na\psi_{\ve\tau}(\x)|^2d\x 
+ 
\fr12 \int_{T_2}|\Lam\rho_{\ve\tau}(\x)|^2 
d\x<\infty, 
\ee 
where $\rho_{\ve\tau}(\x):=\si^0(\x)-e|\psi_{\ve\tau}(\x)|^2$.
The first integral in (\re{HamsT223}) 
is finite, since $\psi_{\ve\tau}\in H^1$. 

\bl\la{lL2}
 $\Lam\rho_{\ve\tau}\in L^2$ for $\tau\in\ccT^0$ and  $\ve\in\R$, and
 \be\la{Gder3} 
D_\tau \Lam\rho:= 
\lim_{\ve\to 0}\fr{\Lam\rho_{\ve\tau}-\Lam\rho^0}{\ve} 
 =e\Lam(\tau\ov{\psi^0}+\psi^0\ov \tau ), 
\ee 
where the limit converges in  $L^2$.
\el
\begin{proof} 
We use the polar coordinates  (\re{3al}) and the corresponding representation
(\re{3rod}):
\be\la{3rod2} 
\Lam\rho_{\ve\tau}= 
\Lam\rho^0+e\ve\cos^2\al\Lam(\tau\ov{\psi^0}+\psi^0\ov\tau)-
e\Lam[\ve^2|\tau|^2\cos^2\al-|\psi^0|^2\sin^2\al]. 
\ee 
Now $\Lam\rho^0\in L^2$ according to (\re{Lamro}). 
Further, $\Lam [\tau\ov{\psi^0}]\in L^2$ by the following arguments:
\medskip\\
a) $\tau\ov{\psi^0}  \in L^2$, 
\medskip\\
b) $\widehat{\tau\ov{\psi^0}}$ is the smooth function on $\Si_2$,
and 
\medskip\\
c) the orthogonality $\tau\bot\psi^0$ implies
that
\be\la{T0} 
\widehat{\tau\ov{\psi^0}}(0,0)=0. 
\ee 
It remains to estimate 
 the last term 
of (\re{3rod2}),
Let us denote $T(\x):=|\tau(\x)|^2$ and  $P(\x):=|\psi^0(\x)|^2$. Then the last  
term (up to a constant factor) reads
\be\la{lt}
R_\ve(\x):=\Lam[\ve^2 T(\x)
\cos^2\al-P(\x)\sin^2\al].
\ee
\bl\la{l37}
$R_\ve\in L^2$
for $\ve\in\R$, and
\be\la{ltm} 
\Vert R_\ve\Vert_{L^2}=\cO(\ve^2),~~~~~~~\ve\to 0.
\ee
\el
\begin{proof}
{\it i)} It suffices to check that 
\beqn \la{sc} 
&&\fr{\ve^2{\hat T(0,\xi)}\cos^2\al-{\hat P(0,\xi)}\sin^2\al} 
{|\xi|} 
\nonumber\\ 
\nonumber\\ 
&=& 
\fr{(\ve^2{\hat T(0,\xi)}-Z_1\tan^2\al)\cos^2\al}{|\xi|}- 
\fr{({\hat P(0,\xi)}-Z_1)\sin^2\al}{|\xi|} 
\in L^2(-1,1). 
\eeqn 
Let us consider each term of the last line of (\re{sc}) separately.
\medskip

1) The first  quotient
belongs to  $L^2(-1,1)$, since 
\be\la{las3} 
\ve^2\hat T(0,0)-Z_1\tan^2\al=
\int_{T_2} \ve^2|\tau|^2 d\x-Z_1\tan^2\al=0~ 
\ee  
by the definition of $\al$ in (\re{3al}) since $\Vert\psi^0\Vert=\sqrt{Z_1}$. 
 \medskip

2) The second quotient belongs to  $L^2(-1,1)$, since 
\be\la{las} 
\fr{\hat\rho^0}{|\xi|}=\fr{{\hat\rho^{\rm per}_1}}{|\xi|}-e\fr{{\hat P}}{|\xi|}= 
\fr 
{{\hat\rho^{\rm per}_1}+eZ_1}{|\xi|}-
e\fr{{\hat P}-Z_1}{|\xi|}, 
\ee 
where all the functions are taken at the point $(0,\xi)$. 
Here the left-hand side belongs to $L^2(-1,1)$, 
since $\Lam\rho^0\in L^2$, while the first term on the right 
belongs to $L^2(-1,1)$ 
 by our assumption (\re{rous5}). 
\medskip\\
{\it ii)}
 The bound (\re{ltm}) holds for both terms of (\re{sc}) 
by the arguments above
since 
$\tan\al\sim\sin\al\sim\ve$ as $\ve\to 0$.\end{proof}
\medskip

Formula (\re{3rod2}) implies (\re{Gder3}),
where the limit converges in  $L^2$ by (\re{ltm}).
\end{proof}

ii) 
Lemma \re{lL2} implies the bound (\re{HamsT223}). Formula
 (\re{Gder2})  follows by differentiation of  (\re{HamsT223}) in $\ve$.
\end{proof}

\subsection{The variational identity} 
 
Since $\psi^0$ is a minimal point, the G\^ateaux derivative (\re{Gder2}) 
vanishes: 
\be\la{GaD} 
\int_{T_2}  \Big[\fr{\h^2}{2\cm} 
(\na\tau\ov{\na\psi^0}+\na\psi^0\ov{\na \tau}) 
-
e {\Lam\rho^0}\Lam(\tau\ov{\psi^0}+\psi^0\ov \tau )
\Big]d\x=0. 
\ee 
Substituting $i\tau$ instead of $\tau$ in this identity 
and subtracting, we obtain 
\be\la{GaD2} 
-\fr{\h^2}{2\cm}\langle\De\psi^0, \tau\rangle 
- e \langle 
\Lam\rho^0,\Lam(\tau\ov{\psi^0} ) \rangle=0. 
\ee 
Next step 
we should evaluate the ``nonlinear'' term.

\bl For the limit functions \eqref{subs2} we have 
\be\la{spg} 
\langle 
\Lam\rho^0,\Lam(\tau\ov{\psi^0} ) \rangle= 
\langle\phi^0\psi^0, \tau \rangle, \qquad \tau\in\ccT^0,
\ee 
where $\phi^0$ is any potential satisfying \eqref{rho2}.

\el 
\begin{proof} 
First we note that $\Lam\rho^0\in L^2$ by \re{Lamro}), and 
$\Lam(\tau\ov{\psi^0} )\in L^2$ as we have established in the proof of Lemma \re{lL2}.
Moreover, 
$\rho^0=-\De\phi^0$. 
 Then, by the Parseval--\allowbreak Plancherel identity, 
\be\la{GaD3} 
\langle 
\Lam\rho^0,\Lam(\tau\ov{\psi^0} ) \rangle= 
\sum_{\bk\in \Ga_2^*\setminus 0}\int \hat\phi^0(\bk,\xi)\ov{\widehat{\tau\ov{\psi^0} }(\bk,\xi)}d\xi+
\lim_{\ve\to 0+}
\int_{|\xi|>\ve}
\hat\phi^0(0,\xi)
\ov{\widehat{ \tau\ov{\psi^0}}(0,\xi)}d\xi= 
\langle\hat \phi^0,\widehat{\tau\ov{\psi^0} } \rangle, 
\ee 
where $\hat \phi^0$ is the distribution on $\Si_2$.
The last identity holds 
(and the right hand side is well defined)
by (\re{T0}) since $\xi\hat\phi^0(0,\xi)\in L^2(-1,1)$ due to (\re{rho2})
with $G^0\in L^2\otimes\R^3$. Finally, 
\be\la{GaD4} \langle\hat \phi^0,\widehat{\tau\ov{\psi^0} } \rangle 
= \langle\phi^0,\tau\ov{\psi^0} \rangle=\int \phi^0(\x)\ov{\tau}(\x)\psi^0(\x)d\x
\ee
by an obvious extension of the Parseval--\allowbreak Plancherel identity.\end{proof} 
\medskip 
 
Using (\re{spg}), we can rewrite (\re{GaD2}) as the variational identity similar to (\re{3GaD22}):
\be\la{GaD22} 
\langle -\fr{\h^2}{2\cm}\De\psi^0 
- e 
\phi^0\psi^0, \tau \rangle=0,\qquad \tau\in\ccT^0. 
\ee 
 
\subsection{The Schr\"odinger equation}

 Now we prove the Schr\"odinger equation (\re{LPS1w}) with $d=2$.
\bl\la{lse2}
$\psi^0$ is the eigenfunction of the 
Schr\"odinger operator:
\be\la{Hpsi} 
H\psi^0=\lam\psi^0,
\ee 
where 
$\lam\in\R$. 
\el
\begin{proof}
 This equation  with $\lam\in\C$ follows as in Lemma \re{lse}.
It remains to verify that $\lam$ is real. 
Our plan is standard: to multiply 
(\re{Hpsi}) by $\psi^0$ and to integrate. 
{\it Formally}, we would obtain 
\be\la{Hpsi2} 
\langle H\psi^0,\psi^0\rangle = 
\lam\langle\psi^0,\psi^0\rangle. 
\ee 
However, it is not clear that 
the left-hand side is well defined and real since the potential $\phi^0(\x)$
can grow by (\re{wfi2}).

To avoid this problem, we  multiply by a  function
$\psi_\ve\in H^1$ with compact support, where $\ve>0$, and 
$\Vert\psi_\ve-\psi^0\Vert_{H^1}\to  0$ as $\ve\to 0$. Then 
\be\la{Hpsi3} 
\langle H\psi^0,\psi_\ve\rangle = 
\lam\langle\psi^0,\psi_\ve\rangle, 
\ee 
and the right-hand side converges to  the one of  (\re{Hpsi2}) as $\ve\to 0$. 
Hence, the left-hand sides also converge.  
 In  detail, 
\be\la{Hpsi4} 
\langle H\psi^0,\psi_\ve\rangle = 
-\fr{\h^2}{2\cm}\langle \De\psi^0,\psi_\ve\rangle+ 
\langle\phi^0\psi^0,\psi_\ve\rangle. 
\ee 
For the middle term, the limit exists and is real. 
Therefore, identity (\re{Hpsi3}) 
implies that the last term is also converging, and hence
it remains to 
make its  limit  real by 
a~suitable choice of approximations $\psi_\ve$. 
We note that 
\beqn\la{cond2} 
\langle\phi^0\psi^0,\psi_\ve\rangle=\lim_{\de\to 0} 
\langle\phi^0\psi_\de,\psi_\ve\rangle 
=\lim_{\de\to 0} 
\langle\phi^0,\ov\psi_\de\psi_\ve\rangle, 
\eeqn 
since $\phi^0\in H^2_{\rm loc}(T_2)\subset C(T_2)$. 
Hence,  we can set 
\be\la{appr} 
\psi_\ve(\x)=\chi(\ve x_3)\psi^0(\x). 
\ee 
where $\chi$ is a real function from $C_0^\infty(\R^3)$ with $\psi(0)=1$. 
Now the functions $\ov \psi_\de(\x)\psi_\ve(\x)$ are real 
for all $\ve,\de>0$. It remains to note that the potential $\phi^0(\x)$ 
is also real by Lemma \re{lP2}.\end{proof}
 \medskip
 
 This lemma implies  equation (\re{LPS1w}). Therefore, $\psi^0\in H^2_{\rm loc}(T_2)$ since $\phi^0\in C(T_2)$.
 Theorem \re{tgs2} ii) is proved.

\subsection{Smoothness of ground state} 
  We have proved that 
  $\psi^0\in H^2_{\rm loc}(T_2)$ under conditions (\re{roL1}) and (\re{rous5}). 
   Using the Schr\"odinger equation (\re{LPS1w}) 
  we can improve the smoothness
 of $\psi^0$ strengthening  the condition (\re{roL1}).
Namely, let us assume that
\be\la{rpem} 
\rho_1^{\rm per}\in C^\infty(T_2).
\ee 
For example,  (\re{rpem}) holds if $\rho_1\in  \cS(\R^3)$,
where $\cS(\R^3)$ is the Schwartz space of test functions.

\bl\la{lpp2} 
Let condition \eqref{rpem} hold, and 
$\psi^0\in H^2_{\rm loc}(T_2)$, $\phi^0\in H^2_{\rm loc}(T_2)$ is a solution to equations 
\eqref{LPS1w}--\eqref{LPS3g} with $d=2$. Then 
the functions $\psi^0,\phi^0$ are smooth.

\el 
The proof is similar to the one of Lemma \re{lpp}.

 
\setcounter{subsection}{0} 
\setcounter{theorem}{0} 
\setcounter{equation}{0} 
 
\section{1D lattice} 
 
The case of a~one dimensional lattice 
$\Ga_1$ is very similar to the 2D case, 
though some of our  constructions and arguments require  suitable 
modifications. 
For $d=1$
we can assume  $\Ga_1=\Z$ without loss of generality and
construct a solution to system   \eqref{LPS1w}--\eqref{LPS3g} 
for the corresponding functions on the `slab' 
 $T_1:=\R^3/\Ga_1=\T^1\times\R^2$ 
with coordinates $\x=(x_1,x_2,x_3)$, where 
$x_1\in \T^1$, and $(x_2,x_3\in\R^2$. 
Now we denote by $H^s$ the  complex Sobolev space on $T_1$, and 
by $L^p$, the  complex Lebesgue space of  functions on $T_1$.

 \subsection{The existence of ground state}
The existence of the ground state follows
by minimizing the energy  
(\re{HamsT}), where the integral is extended 
over $T_1$ instead of $T_3$. 
The neutrality condition 
of type (\re{neu1})
holds for 
$\Ga_1$-periodic states with finite energy, as for $d=2$. 

Again we restrict ourselves by $N=1$, 
so $\ov\x^0=\x^0_1$ can be chosen arbitrary, and we set $\x^0_1=0$.

The energy in the slab $T_1$
is defined by expression similar to (\re{HamsT22}): 
\be\la{HamsT221} 
\E(\psi):= 
\int_{T_1}\Bigl[ 
\fr{\h^2}{2\cm}|\na\psi(\x)|^2 
+ 
\fr12 |\Lam\rho(\x)|^2 
\Bigr]d\x, \qquad \rho(\x):=\si^0(\x) -e|\psi(\x)|^2. 
\ee 
Here $\si^0=\rho_i^{\rm per}\in L^1\cap L^2$ as in  (\re{rrr22}).
Hence, 
\be\la{intro21} 
\ds\int_{T_1} 
\si^0(\x)d\x=eZ_1,~~~~~Z_1>0. 
\ee  
Now 
the Fourier representation for the test functions 
$\vp(x)\in C_0^\infty(T_1)$
is defined by 
\be\la{FR21} 
\vp(\x)= \fr 1{2\pi} 
\sum_{\bk\in \Ga_1^*} e^{-i\bk x_1} 
\int_{\R^2} e^{-i(\xi_1x_2+\xi_2x_3)}\hat \vp(\bk,\xi)d\xi, 
\ee 
where  $\Ga_1^*=2\pi\Ga_1$
and 
\be 
\la{FT21} 
\hat 
\vp(\bk,\xi)=F\vp(\bk,\xi)=\fr 1{2\pi} 
\int_{T_1} e^{i(\bk x_1+\xi_1x_2+\xi_2x_3)} \vp(\x)d\x,\qquad  (\bk,\xi)\in  \Si_1:=\Ga_1^*\times \R^2. 
\ee 
The operator $\Lam=(-\De)^{1/2}$ is defined for $\vp\in L^1\cap L^2$ by 
the same formula (\re{Lam})
provided the quotient belongs to $L^2(\Si_1)$. This implies 
\be\la{021} 
\hat \vp(0,0)=0. 
\ee 
For $\psi\in H^1$ with finite energy (\re{HamsT221}) 
we have $\Lam\rho\in L^2(\Si_1)$, and 
hence, (\re{021}) with  $\vp=\rho$ implies 
the neutrality condition 
(\re{neu10}) with $d=1$: 
\be\la{neu21} 
\hat\rho(0,0)=\int_{T_1} \rho(\x)d\x= 
\int_{T_1} [\si^0(\x) -e|\psi(\x)|^2]d\x=0. 
\ee 
Now (\re{intro21}) gives
\be\la{M21} 
\int|\psi(\x)|^2d\x=Z_1.
\ee 
Thus, the finiteness of the Coulomb energy 
$\Vert\Lam\rho\Vert^2$ prevents the electron charge from escaping to infinity, 
as in 2D case.

Finally, the Fourier transform $F:\psi\mapsto\hat\psi$ 
is a unitary operator from $L^2(T_1)$ to $L^2(\Si_1)$. Hence, 
energy (\re{HamsT22}) reads 
\be\la{HamsT22F1} 
\E(\psi)= 
\sum 
_{\bk\in \Ga_1^*} 
\int_{\R^2}\Bigl[ 
\fr{\h^2}{2\cm}(\bk^2+\xi^2)|\hat\psi(\bk,\xi)|^2 
+ 
\fr12 
\fr {|\hat\rho(\bk,\xi)|^2} {\bk^2+\xi^2} 
\Bigr]d\xi. 
\ee 
 
\bd 
$M_1$ denotes the set of  
$\psi\in H^1$ satisfying the neutrality condition (\re{M21}). 
\ed 
 
We note that 
(\re{intro21}) can be written as
$\hat\rho^{\rm per}_1(0)-eZ_1=0$. 
We assume  moreover, 
\be\la{rous51} 
\mbox{\bf Condition III.}\qquad
\fr{\hat\rho^{\rm per}_1(0,\xi)-eZ_1}{|\xi|}\in L^2(D),\qquad
D:=\{\xi\in\R^2: |\xi|\le 1\} 
\qquad
\ee 
similarly to (\re{rous5}).
For example, this condition holds, 
provided that  
\be\la{rous5if1} 
\int_{\R^3}(1+|x_2|+|x_3|)|\rho_1(\x)|d\x <\infty.
\ee 
The third main result of this chapter is the following.
 
 \bt\la{tgs1} 
Let conditions \eqref{roL1} and \eqref{rous51} hold,  and $N=1$. Then 
\medskip\\
i) There exists 
$\psi^0\in M_1$ with
\be\la{U0min1} 
\E(\psi^0)=\inf_{\psi\in M_1}\E(\psi).
\ee 
ii) Moreover, $\psi^0\in H^2_{\rm loc}(T_1)$ and
satisfies equations \eqref{LPS1w}--\eqref{LPS3g} with $d=1$,
where the
potential $\phi^0\in H^2_{\rm loc}(T_1)$ is real, $\x^0_1=0$, 
 and $\om^0\in\R$. 
\medskip\\
iii) The following bound holds
\be\la{wfi1} 
|\phi^0(\x)|\le 
C(1+|x_2|+|x_3|)^{1/2},\qquad\qquad \x\in T_2.
\ee 
\et

The proof is similar to the one of Theorem \re{tgs2}.
As in 2D case, we obtain 
$\psi^0\in M_1$ as a minimizer for the energy (\re{HamsT221}).
The  potential 
$\phi^0$ can be constructed by a modification of  Lemma \re{lP2}, see Appendix below.
 \medskip

 Finally, next lemma follows similarly to Lemma \re{lpp}.

 \bl\la{lpp1} 
 The functions $\psi^0,\phi^0$ are smooth 
under condition 
\be\la{rpem1} 
\rho_1^{\rm per}\in C^\infty(T_1).
\ee 

\el

 \subsection{The bound  for the potential} 
 We start with obvious modifications of the proof of Lemma \re{lP2}.
 Namely, 
 the potential $\phi^0(\x)$ for the 1D lattice 
 satisfies the equation of type (\re{rho2}) with 
 \be\la{G0}
 G^0:= 
-  iF^{-1}\fr{(\bk,\xi)}{\bk^2+\xi^2}\hat\rho^0(\bk,\xi)
 \in L^2(T_1), \qquad \rot G^0(\x)\equiv 0.
 \ee
 We use  the splitting of type  (\re{splitt}), and
 respectively, the solution splits as $\phi^0=\phi_1+\phi_2$. 
The second solution 
$\phi_2\in H^2$  as in the proof of Lemma \re{lP2}.
Hence, $\phi_2$
is bounded continuous function on $T_1$ by the 
Sobolev embedding theorem.
 
 On the other hand, the analysis of the first solution needs some 
 modifications.
Now $G_1(\x)=g_1(x_2,x_3)\in L^2(\R^2)\otimes \R^2$
 is the real vector field,
and $\supp \hat g_1\subset \{\xi\in\R^2: |\xi|\le 1\}$.
Therefore,
$g_1$ is the smooth function, and 
\be\la{Dphi1}
 \De \phi_1=\na\cdot g_1\in L^2(\R^2), \qquad \rot g_1(\x)\equiv 0.
 \ee
  Respectively, the solution to 
$\na \phi_1= g_1$
is given by 
the contour integral
\be\la{pott} 
\phi_1(\x)=\int_0^\x g_1(\y)d\y+C,\qquad\qquad \x\in \R^2,
\ee 
which does not depend on the path in $\R^2$. This solution is real and smooth.
\medskip

We still should prove the estimate (\re{wfi1}).
We will 
deduce it from the corresponding 
estimate 'in the mean'. 
Let us denote the circle $B:=\{\x\in\R^2:|\x|<1\}$.

\bl\la{lpot}
For any unit vector 
$\e\in\R^2$ 
\be\la{esm}
\Vert\phi_1 \Vert_{L^2(B+\e R)}\le C(1+R)^{1/2},\qquad \qquad R>0.
\ee
\el
\begin{proof}  First, (\re{pott}) implies that 
\be\la{pot2} 
\phi_1(\x+\e R)-\phi_1(\x)=\int_0^R g_1(\x+t\e)dt, \qquad\qquad \x\in \R^2
\ee 
for any $R\in\R$. 
Now
the Cauchy-Schwartz inequality implies that
\be\la{pot23} 
|\phi_1(\x+\e R)|^2\le C_1+2R\int_0^R |g_1(\x+t\e)|^2dt, \qquad\qquad \x\in B
\ee 
since  the function $\phi_1$ is bounded in $B$.
Finally, averaging over $\x\in B$, we get 
\be\la{pot24} 
\int_B|\phi_1(\x+\e R)|^2d\x\le C_1|B|+2R\int_0^R \int_B|g_1(\x+t\e)|^2d\x \5 dt\le C_1+C_2R\Vert g_1\Vert_{L^2(\R^2)}^2.
\ee 
Hence, 
 (\re{esm})
is proved.\end{proof}
\medskip

Now  (\re{wfi1}) follows from the Sobolev embedding theorem:
\be\la{sapr}
\max_{\x\in B+\e R}|\phi_1(\x)|\le C_3   \Vert\phi_1 \Vert_{H^2(B+\e R)}\le 
C_4[ \Vert\De\phi_1 \Vert_{L^2(B+\e R)}+ \Vert\phi_1 \Vert_{L^2(B+\e R)})]\le C(1+R)^{1/2}
\ee
since $\De \phi_1\in L^2(\R^2)$ by (\re{Dphi1}). 
\medskip

\br
Our estimate  (\re{wfi1}) seems to be  far from optimal since
 the potential of uniformly charged line grows logarithmically  with the distance,
One could expect an optimal estimate 
$$
|\phi^0(\x)|\le C[\log(2+|x_2|+|x_3|)]^{1/2}
$$
in the case $\na\phi^0\in L^2$
due to the example
$
\phi(\x)=[\log(2+|x_2|+|x_3|)]^{1/2-\ve}
$
with $\na \phi(\x)\in L^2$ for  $\ve>0$.

\er

\chapter{Linear Stability}

\centerline{Abstract}
\medskip

In this chapter
 we consider the
Schr\"odinger--Poisson--Newton equations for 
infinite
crystals
with one ion per cell.
We linearize this dynamics at the periodic minimizers of energy per cell (such minimizers were constructed in previous chapter)
and
introduce he corresponding `Jellium' and the `Wiener' conditions on the  ion charge densities that ensures
the stability of the linearized dynamics.
The conditions are 
suitable (but not identical) versions  of the 
conditions (\re{Wai}) and (\ref{W1}) respectively, introduced in the context
of finite crystals.

Our main result
is the {\it energy positivity}
for the
Bloch generators
of the
linearized dynamics  under 
these conditions.

The Bloch generators  are
nonselfadjoint (and even nonsymmetric) Hamiltonian operators.
We diagonalize these generators using our theory of
spectral resolution of the Hamiltonian ope\-rators {\it with positive definite energy}
\ci{KK2014a,KK2014b}.
The stability of the linearized crystal dynamics is established using this  spectral resolution which 
is a special 
version of the Gohberg--Krein--Langer theory of selfadjoint 
operators in the Hilbert spaces with indefinite metric \ci{GK,KL1963,L1981}.

In the next chapter, we estabilsh the dispersive decay in suitable norms 
for the linearised dynamics.

\section{Introduction}
In this chapter, we analyse 
the dynamic stability of a~crystal periodic minimizer
of energy per cell
in linear approximation for the simplest Schr\"o\-din\-ger--Poisson model.
The periodic minimizers for this model were constructed in 
the previous chapter.

We consider infinite crystals with one ion per cell.
The electron cloud is described by  the one-particle Schr\"odinger equation;  the ions are looked upon as  classical particles
that corresponds to the Born--Oppenheimer  approximation.
The ions interact with the electron cloud via
the scalar potential, which  is a~solution to the corresponding Poisson equation.

This model does not respect the Pauli exclusion principle for electrons.
Nevertheless, it provides a convenient framework
to introduce suitable functional tools that might be instrumental for physically
more realistic models (the Thomas--Fermi, Hartree--Fock, and second quantized models).
In particular, we find a novel 
stability criterion \eqref{4-W1}, \eqref{4-Wai}.
\medskip

We denote by
$\sigma(x)\in L^1(\R^3)$ the charge density of one ion,
\be\la{4-ro+}
\int_{\R^3} \sigma(x)dx=eZ>0,
\ee
where $e>0$ is the elementary charge.
We assume througout this chapter that
\be\la{4-L123i}
\langle x\rangle^4\si\in L^2(\R^3),\qquad (\De-1)\si\in L^1(\R^3).
\ee
We consider the cubic
lattice  $\Ga= \Z^3$ for the simplicity of notations. 
Let $\psi(x,t)$ be the wave function of the electron field,
$q(n,t)$ denote the ions displacements,
and
$\phi(x)$ be the electrostatic  potential generated by the ions and electrons.
We assume that $\hbar=c=\cm=1$, where $c$ is the speed of light and $\cm$ is the electron mass.
The coupled Schr\"odinger--Poisson--Newton equations  read 
\beqn\la{4-LPS1}
i\dot\psi(x,t)\!\!&=&\!\!-\fr12\De\psi(x,t)-e\phi(x,t)\psi(x,t),\qquad x\in\R^3,
\\
\nonumber\\
-\De\phi(x,t)\!\!&=&\!\!\rho(x,t):=\sum_n\sigma(x-n-q(n,t))-e|\psi(x,t)|^2,\qquad x\in\R^3,
\la{4-LPS2}
\\
\nonumber\\
M\ddot q(n,t)
\!\!&=&\!\!-\langle\nab\phi(x,t),\sigma(x-n-q(n,t))\rangle,
\qquad n\in\Z^3.
\la{4-LPS3}
\eeqn
Here the
brackets
 stand for the Hermitian scalar product in the Hilbert
space $L^2(\R^3)$ and for its various extensions, the  series (\ref{4-LPS2}) converges in
a suitable sense, and $M>0$.
All the derivatives here and below are understood in the sense of distributions.
These equations can be  {\it formally} written as a~Hamiltonian system
\be\la{4-HSi}
i\dot \psi(x,t)=\pa_{\ov \psi(x)}\cH,
~~~
\dot q(n,t)=\pa_{p(n)}\cH,
~~~
\dot p(n,t)=-\!\pa_{q(n)}\cH,
\ee
where $\pa_{\ov z}:=\fr12(\pa_{z_1}+i\pa_{z_2})$ with $z_1=\rRe z$ and $z_2=\rIm z$, and 
the {\it formal} Hamiltonian  functional is
\be\la{4-Hfor}
 \cH(\psi,q,p)=\fr12\int_{\R^3}[|\na\psi(x)|^2+\rho(x)G\rho(x)]dx+\sum_{n} \fr{p^2(n)}{2M},
\ee
where
$ q:=(q(n): ~n\in\Z^3)$, $p:=(p(n): ~n\in\Z^3)$,  $\rho(x)$ is defined
similarly to
(\ref{4-LPS2}), and $G:=(-\De)^{-1}$.
\medskip

The system  (\ref{4-LPS1})--(\ref{4-LPS3})  admits
ground states that is 
$\Gamma$-periodic solutions  of type
\be\la{4-gr}
\psi^0(x)e^{-i\om^0 t}~,~~~ \phi^0(x)~,~~~~q^0(n)=q^0~~{\rm and}~~p^0(n)=0~~~~
{\rm for}~~ n\in\Z^3,
\ee
which are minimizers of  energy per cell (\ref{HamsT}):
\be\la{4-HamsT} 
E(\psi)\!:=\!\! 
\int_{\Pi}\Bigl[ 
\fr{\h^2}{2\cm}|\na\psi(\x)|^2 
+ 
\fr12\phi(\x)\rho(\x)\Bigr]d\x ,\qquad \phi(\x):=(-\De)^{-1}\rho, 
\ee
where 
$(-\De)^{-1}\rho$ is defined by (\ref{Fou2}), and
 $\Pi:=[0,1]^3$.

Such periodic minimizers
were constructed in previous chapter  for general lattice with several ions per cell. In  our case
the ion position  $q^0\in\R^3$ can be chosen arbitrarily, and we set
$q^0=0$ everywhere below. The corresponding ion charge density reads
\be\la{4-ro+2}
\si^0(x):=
\sum_{n}\si(x-n).
\ee
We will see that the periodic minimizers can be stable 
depending on the choice of the ion density $\sigma$. 
However,
we only study very special densities 
$\sigma$ satisfying some conditions discussed below.
Our first basic assumption coincides with (\ref{Wai}),
\be\la{4-Wai}
 ~~~~~~~\mbox{\bf The Jellium Condition:}~~~~~~~~~~~~ \ti\si(2\pi m)=0,\quad m\in\Z^3\setminus 0.
~~~~~~~~~~~~~~~~~~~~~~~~~~~~~~~~~~~~~~~~~~~~~~~~~~~~~~~~~
\ee
This condition  immediately implies that
the periodized ions charge density (\ref{4-ro+2})
is a positive constant everywhere in space, and
\be\la{persi}
\si^0(x)\equiv eZ,\qquad x\in\R^3
\ee
by (\ref{4-ro+}).
This identity  implies that
under condition (\ref{4-Wai})  there exist the ground states
\be\la{4-ppo}
\psi^0(x)\equiv \sqrt{Z} e^{i\theta},\qquad
\phi^0(x)\equiv 0,\qquad\om_0=0
\ee
with  $\theta\in[0,2\pi]$,
which are stationary solutions of
 (\ref{4-LPS1})--(\ref{4-LPS3}) 
 and correspond to the minimal zero 
energy  per  cell  (\ref{4-HamsT})
since $\rho(x)\equiv 0$,
so, 
the ion and electron densities cancel each other.
The simplest example of such a 
$\sigma$ is a constant over the unit cell of a given lattice, which is what physicists 
usually call Jellium \cite{GV2005}.
Here we study this model in the rigorous context of the Schr\"odinger-Poisson equations.  
\medskip

In this chapter, we prove
the stability of  the
{\it formal linearization}
of the nonlinear  system (\ref{4-LPS1})--(\ref{4-LPS3})
at the periodic minimizer (\ref{4-ppo}).
The system is $U(1)$-invariant, that 
is for any solution $(\psi, q,p)$,
the function
$(e^{i\theta}\psi, q,p)$
with $\theta\in\R$
 is also a solution. Hence,
it suffices to consider only 
 the case $\theta=0$.
Substituting
\be\la{4-lin2i}
 \psi(x,t)=\psi_0(x)+\Psi(x,t)
\ee
into the nonlinear equations (\ref{4-LPS1}), (\ref{4-LPS3})
with $\phi(x,t)=G\rho(x,t)$,
we {\it formally} obtain  the linearized equations 
\be\la{4-LPS1Li}
\left\{
\!\!\!\!\!
\ba{rcl}
\!\!\!\! &\!\!\!\!&\!\!\!\!i\dot\Psi(x,t)=-\fr12\De\Psi(x,t)\Psi(x,t)-e\psi_0(x) G\rho_1(x,t),\quad x\in\R^3\\\\
\!\!\!\! &\!\!\!\!&\!\!\!\!M\dot q(n,t)=p(n,t),
\qquad \dot p(n,t)=-\langle\nab G\rho_1(t), \sigma(x-n)\rangle,
\qquad n\in\Z^3 
\ea\right|.
\ee
Here
$\rho_1(x,t)$ is the linearized charge density
\be\la{4-ro1i}
 \rho_1(x,t)=-\sum_n q(n,t)\cdot\na
 \si(x-n) -2e\psi_0(x)\,\, \rRe\Psi(x,t).
\ee
The system (\ref{4-LPS1Li}) is linear over~$\R$, but it is not complex linear.
This is due to the  last term in (\ref{4-ro1i}),  which appears from
the linearization of the term  $|\psi|^2=\psi\ov\psi$ in (\ref{4-LPS2}).
However, we need the complex linearity for the application of the spectral theory.
That is  why we will consider below the complexification  of system (\ref{4-LPS1Li})
by writing it in the variables
$\Psi_1(x,t):=\rRe\Psi(x,t),\Psi_2(x,t):=\rIm \Psi(x,t)$.
Then (\ref{4-LPS1Li})  can be written as
\be\la{4-JDi}
\dot Y(t)=AY(t),\qquad
 A=\left(\ba{ccrl}
 0   &  H^0 &  0  & 0\medskip\\
-H^0-2e^2\psi_0 G\psi_0 &   0  & -S  & 0\\
        0                         &           0                        &   0    &   M^{-1}\\
     -2S^{\5*}                   &              0            &  -T    &  0\\
\ea\right),
\ee
where we denote
$Y(t)=(\Psi_1(\cdot,t),\Psi_2(\cdot,t),q(\cdot,t),p(\cdot,t))$,
$H^0:=-\fr12\De$,
the operators $S$ and $T$ correspond to
matrices
(\ref{4-S}) and (\ref{4-T}), respectively.
The Hamiltonian representation (\ref{4-HSi}) implies that
(cf. (\ref{E''}) and (\ref{E''N}))
\be\la{4-AJB}
A=JB,\qquad B=
\left(\ba{cccl}
 2H^0+4e^2\psi_0 G\psi_0 & 0 & 2S & 0
 \medskip\\
 0 & 2H^0  &0 & 0\medskip\\
 2S^{\5*}  &    0  &   T    & 0  \\
      0      &    0            &   0    &  M^{-1} \\
\ea\right),\qquad
J=\left(\ba{cccc}
0  & \fr12 &  0 & 0\\
-\fr12 & 0 &  0 & 0\\
0 & 0 &  0 & 1\\
0 & 0 & -1 & 0
\ea\right).
\ee
We show that the generator $A$ is densely defined
in the Hilbert space
\be\la{4-cX}
{\cX^0}:=L^2(\R^3)\oplus L^2(\R^3)\oplus\R^3\oplus\R^3.
\ee

The linear system (\ref{4-JDi}) is Hamiltonian with the Hamiltonian functional which is the quadratic form
\be\la{Ham0}
\cH^0(Y)=\fr12\langle Y,BY\rangle=
\fr12\int[|\na\Psi(x)|^2+\rho_1(x)G\rho_1(x)]dx+
\sum_n\fr {p^2(n)}{2M},\qquad Y=(\Psi_1,\Psi_2,q,p).
\ee

Our main result is the stability of the linearized system (\ref{4-JDi}):
for any initial state of finite energy there exists a unique  global solution which is bounded in the energy norm.
\medskip

Obviously, 
the generator $A$
commutes with  translations by vectors from $\Ga$.
Hence, the  equation
\eqref{4-JDi} can be reduced with the help of the Fourier--Bloch--\allowbreak Gelfand--\allowbreak Zak
transform
(\ref{4-YPi})
to  equations
with the corresponding Bloch generators  $\ti A(\theta)=J\ti B(\theta)$, given by (\ref{4-tiA}),
which depend on the parameter $\theta$ from the Brillouin zone
 $\Pi^*:=[0,2\pi]^3$, and
\be\la{4-hess2i}
 \ti B(\theta) =\!\left(\!\ba{cccl}
 2\ti H^0(\theta)+4e^2\psi_0 \ti G(\theta)\psi_0&  0 & 2\ti S(\theta)&0
 \medskip\\
 0 &2\ti H^0(\theta)&  0 &0
\medskip\\
 2\ti S^{\5*}(\theta) &   0
 &   \hat T(\theta)             & 0  \\
      0    &    0            &   0 &  M^{-1} \\
\ea\!\right),\qquad
\theta\in\Pi^*
\setminus\Ga^*,
\ee
 where $\Ga^*:=2\pi\Z^3$, and
$\ti H^0(\theta):=-\fr12(\na-i\theta)^2$.
Further, $\ti G(\theta)$ is the inverse of the operator
$(i\na+\theta)^2:H^2(T^3)\to L^2(T^3)$. Finally, $\ti S(\theta)$ and $\hat T(\theta)$
are defined, respectively, by (\ref{4-tiHS}) and  (\ref{4-K3}).

The operator $\ti B(\theta)$ is selfadjoint in the Hilbert space  ${\ti\cX^0}(T^3)$ with the domain
 ${\ti\cX^2}(T^3)$, where we denote
\be\la{4-XTs}
{\ti\cX^s}(T^3):=H^s(T^3)\oplus H^s(T^3)\oplus \C^3\oplus \C^3,\qquad  T^3:=\R^3/\Ga
\ee
for $s\in\R$; its spectrum is discrete.
 However,
the operator $A$ is not selfadjoint and even not symmetric in $\ti\cX^0$ -- this a typical situation in the linearization of $U(1)$-invariant
nonlinear equations  \ci[Appendix B]{KK2014a}.
Respectively,  the  Bloch generators
$\ti A(\theta)$ are not  selfadjoint in ${\ti\cX^0}(T^3)$

The main crux here  is that
we cannot apply the von Neumann
spectral theorem to the nonselfadjoint generators $A$ and $\ti A(\theta)$.
We solve this problem by applying our spectral theory of abstract Hamiltonian operators
with positive energy \ci{KK2014a,KK2014b}.
This is why  we need
the positivity of the energy operator $\ti B(\theta)$: for $\ti Y\in{\ti\cX^2}(T^3)$
\be\la{4-Hpos2k}
\cE(\theta,\ti Y):=\langle \ti Y, \ti B(\theta)\ti Y\rangle_{{\ti\cX^0}(T^3)}\ge \vka(\theta)\Vert \ti Y\Vert_{{\ti\cX^1}(T^3)}^2, \quad \mbox{\rm where }\,\, \vka(\theta)>0\,\,\,\, \mbox{\rm  for a.e.}~~\theta\in
\Pi^*\setminus\Ga^*
\ee
and the brackets  denote the scalar product in ${\ti\cX^0}(T^3)$. Equivalently,
\be\la{4-Hpos2}
\ti B_0(\theta):=\inf\,\,\spec \ti B(\theta)>0\qquad \mbox{\rm  for a.e.}~~\theta\in
\Pi^*\setminus\Ga^*.
\ee
The main result of this chapter is  the positivity
\eqref{4-Hpos2}
for the ions charge  densities
 $\si$ satisfying the Jellium condition (\ref{4-Wai}) and
 also
\be\la{4-W1}
\mbox{\bf The Wiener Condition:}~~~~~~  \Si(\theta):=\sum_m\Big[
 \fr{\xi\otimes\xi}{|\xi|^2}|\ti\si(\xi)|^2\Big]_{\xi=2\pi m+\theta}>0~
 \quad  ~~{\rm for\ a.e.}~~\theta\in \Pi^*\setminus\Ga^*,~~~~~~
\ee
where the series converges by (\ref{4-L123i}).
Equivalently,
\be\la{4-W12}
\Si_0(\theta)>0~~~~~~~~~
 \quad  ~~{\rm for\ a.e.}~~\theta\in \Pi^*\setminus\Ga^*,~~~~~~~~~~~~~~
 ~~~~~~~~~~~~~~~~~~~~~~~~
\ee
where  $\Si_0(\theta)$
 is the minimal eigenvalue of the matrix
$\Si(\theta)$.
This condition is
 an analog of the Fermi Golden Rule for crystals.
It is easy to construct examples of densities
 $\si(x)$ satisfying conditions \eqref{4-W1} and \eqref{4-Wai}.

\begin{example}\la{4-ex}
 \eqref{4-W1} holds for
 $\si$ satisfying (\ref{4-L123i})
if
\be\la{4-Wex}
\ti\si(\xi)\ne 0\qquad ~~\mbox{\rm for \ a.e.}~~\xi\in\R^3.
\ee

\end{example}

\begin{example}\la{4-ex2}
Let us define  the function $s(x)$ for $x\in\R$ by its Fourier transform
$\ti s(\xi):=\ds\fr{2\sin\ds\fr\xi2}{\xi}e^{-\xi^2}$, and set
\be\la{4-exW}
\si(x):=eZ s(x_1)s(x_2)s(x_3),\qquad x\in\R^3.
\ee
Then $\si(x)$ is a holomorphic function of $x\in\C^3$ satisfying conditions
\eqref{4-W1},  \eqref{4-Wai}, (\ref{4-ro+}), (\ref{4-L123i} ), and besides,
\be\la{4-rd}
 |\pa^\al \si(x)|\le C(\al,a)e^{-a|x|},~~~~x\in\R^3,
\ee
for any $a>0$ and $\al$ by the Paley--Wiener theorem.
\end{example}

 We prove the positivity \eqref{4-Hpos2}  in Sections \ref{sLD}--\ref{sEP}.
 In Section \ref{4-sred}
 we apply the positivity 
to give a meaning to the associated linearized dynamics.

\medskip

We prove \eqref{4-Hpos2} with
\be\la{4-vkat}
\ti B_0(\theta)\ge
\ve d^4(\theta)\Si_0(\theta),\qquad
\theta\in\Pi^*\setminus\Ga^*,
\ee
where $\ve>0$ is sufficiently small and  $d(\theta):=\dist(\theta,\Ga^*)$.
This implies that $\spec B\subset [0,\infty)$.
Moreover, we show in Theorem \ref{4-tpose} ii) that
\be\la{4-vkat2}
\inf\spec\ti B_0(\theta)\le \Si_0(\theta),\qquad
\theta\in\Pi^*\setminus\Ga^*.
\ee
This inequality implies that $0\in \spec B$. Indeed,
the conditions (\ref{4-W1}) and  \eqref{4-Wai} imply
that $\Si(\theta)$ is a continuous $\Ga^*$-periodic function, which admits
the asymptotics
\be\la{4-ask}
\Si(\theta)\sim \fr{\theta\otimes\theta}{|\theta|^2}\ti\si(0)+\cO(|\theta|^2),\qquad \theta\to 0.
\ee
However,
the matrix $\theta\otimes\theta$ is degenerate, and hence,   $\Si_0(\theta)\to 0$ as
$\theta\to 0$
 by the asymptotics (\ref{4-ask}).
Therefore, the positivity (\ref{4-Hpos2})
breaks down   at  $\theta\in\Ga^*\cap \Pi^*$ by (\ref{4-vkat2}).
Examples \ref{4-ex} and \ref{4-ex2} demonstrate that the positivity  can also break down at
some  other points and submanifolds of $\Pi^*$ that depend on the ion charge density $\si$.
\medskip

Let us comment on our approach.
The structure of the periodic minimizer (\ref{4-ppo}) under condition
\eqref{4-Wai}  seems trivial. However, even in this case the proof
of the positivity \eqref{4-Hpos2} is not straightforward, since the operators
$\ti S(\theta)$ and $\hat T(\theta)$
in  $\ti B(\theta)$
depend on the fuctional parameter $\si$.
Our proof of
\eqref{4-Hpos2}
 relies on i) a novel
factorization (\ref{4-fact}) of the
matrix elements of  $\ti B(\theta)$, and ii)
Sylvester-type arguments for matrix operators (see Remark \ref{4-rS}).

We show that the condition \eqref{4-W1} is necessary for the positivity \eqref{4-Hpos2}.
We expect that the  condition  \eqref{4-Wai} is also necessary for the positivity \eqref{4-Hpos2},
however, this is still an open challenging problem.

Finally, the positivity (\ref{4-Hpos2}) allows us to construct the spectral resolution
of $\ti A(\theta)$, which results in the stability for
the linearized dynamics \eqref{4-JDi}. The spectral resolution
is constructed with application of our
spectral theory of abstract Hamiltonian operators \ci{KK2014a,KK2014b}.

In concluzion,
all our methods and results extend obviously to 
 equations (\ref{4-LPS1})--(\ref{4-LPS3})  in the case of
general lattice 
\begin{equation}\la{4-gG}
\Ga=\{n_1a_1+n_2a_2+n_3a_3: (n_1,n_2,n_3)\in\Z^3\},
\end{equation}
where the generators $a_k\in\R^3$ are linearly independent. In this case
 the condition \eqref{4-Wai} becomes 
 \begin{equation}\la{4-Waig}
\ti\si(\ga^*)=0,\quad \ga^*\in\Ga^*\setminus 0,
\end{equation}
where $\Ga^*$ denotes the dual lattice, i.e.,
$
\Ga^*=\{m_1b_1+m_2b_2+m_3b_3: (m_1,m_2,m_3)\in\Z^3\}
$
with $\langle a_k,b_j\rangle=2\pi\de_{kj}$. 
The condition \eqref{4-Waig} 
claryfies the relation between the properties 
of the ions and the resulting crystal geometry.

\br
{\rm
Conditions \eqref{4-Wai}, \eqref{4-Waig} seem 
to be rather restrictive. On the other hand, the distinction between the
ions and electron field is not too sharp, since each ion
contains in itself a number of bonding electrons. Physically,
the ion charge density $\si(x)$ might vary
during the process of the crystal formation
due to interaction with the electron field. Respectively,
one could expect that identities \eqref{4-Wai}, \eqref{4-Waig}  may result from
this process. 
}
\er

This chapter is organized as follows.
In Section 2  we recall our result \ci{K2014} on the existence of a periodic minimizer
In Sections 3--5 we study the Hamiltonian structure
of the linearized dynamics and find a~bound of the energy from below.
In Section 6 we calculate the generator of the
linearized dynamics
in the  Fourier--Bloch representation.
In
Section 7 we prove  the  positivity of the energy.
In Section 8 we apply this positivity to
the stability of the linearized dynamics.
Finally, in  Sections 9 and 10 we establish small charge asymptotics of the periodic minimizer
and  construct examples of negative energy.
Some technical calculations are carried out in Appendices.

\setcounter{equation}{0}
\section{Linearized dynamics}\la{sLD}
Let us calculate the entries of the matrix operator
(\ref{4-JDi})
under conditions (\ref{4-L123i}).
For $f(x)\in C_0^\infty(\R^3)$ the Fourier transform is defined by
\be\la{4-Fu2}
f(x)=\fr 1{(2\pi)^3}\int_{\R^3} e^{-i\xi x}\ti f(\xi)d\xi,\qquad x\in \R^3;
\qquad
\ti f(\xi)=\int_{\R^3} e^{i\xi x}f(x)dx,\qquad \xi\in \R^3.
\ee
The  conditions \eqref{4-L123i} imply that
\be\la{4-L123}
(\De-1)\ti\si\in L^2(\R^3),\qquad \langle\xi\rangle^2\ti\si(\xi)\le \const.
\ee
The system (\ref{4-LPS1Li}) and (\ref{4-ro1i}) imply that
the operator-matrix $A$
 is given by  \eqref{4-JDi}, where
$ S$  denotes the operator with the `matrix' similar to (\ref{S}) and (\ref{SN}) 
\be\la{4-S}
 S(x,n):=e\psi_0 G\na\si(x-n):~~n\in\Z^3,~x\in\R^3.
\ee
Finally, $T$ is the real matrix similar to (\ref{T})
\be\la{4-T}
T(n,n'):=-\ds\langle  G\na\otimes\na\si(x-n'),  \si(x-n) \rangle.
\ee
The operators
$G: L^2(\R^3)\to L^2(\R^3)$ and $S:l^2:=l^2(\Z^3)\otimes \C^3\to L^2(\R^3)$
are not bounded due to the `infrared divergence', see Remark 
\ref{4-cAA*} i).
In the next section, we will construct a dense domain for all these operators.

On the other hand, the
operator $T$  is bounded in view of the following lemma.
Denote by $\Pi$  the primitive cell
\be\la{4-cellPi}
\Pi:=\{(x_1,x_2,x_3):0\le x_k\le 1,~k=1,2,3\}.
\ee
Let us define the Fourier transform on $l^2$ as
\be\la{4-Fu}
\hat q(\theta)=\sum_{n\in\Z^3} e^{in\theta}q(n)\qquad~~{\rm for\  a.e.}~~ \theta\in \Pi^*;\qquad
q(n)=\fr 1{|\Pi^*|}\int_{\Pi^*}   e^{-in\theta}\hat q(\theta)d\theta,~~n\in\Z^3,
\ee
where  $\Pi^*=2\pi \Pi$ denotes the primitive cell of the lattice $\Ga^*$,  the series converging in $L^2(\Pi^*)$.
\begin{lemma}\la{4-lT}
Let conditions (\ref{4-L123i})  hold. Then
the operator $T$ is bounded in  $l^2$.
\end{lemma}
\begin{proof}
The operator $T$ reads as the convolution $Tq(n)=\sum T(n-n')q(n')$,
where
\be\la{4-Kn}
T(n)=-\langle G\na\otimes \na\si(x),   \si(x-n) \rangle.
\ee
By the Fourier transform (\ref{4-Fu}), the convolution operator $T$ becomes the
multiplication,
\be\la{4-K22}
\widehat{T q}(\theta)=\hat T(\theta) \hat q(\theta  )\qquad ~~{\rm for\ a.e.}~~\theta\in\Pi^*\setminus\Ga^*.
\ee
By the Bessel-Parseval identity it suffices to check that
the `symbol' $\hat T(\theta)$ is a bounded function. This follows 
from (\ref{4-T})
by direct
calculation. First, we apply the Parseval identity
\beqn\la{4-K3}
\!\!\!\!\!\!\!\!\!\!\!\!\!\!\!\!\hat T(\theta)\!\!\!\!&\!\!=\!\!&\!\!\!
-\sum_n e^{in\theta}\langle  G\na\otimes \na\si(x),\si(x-n)\rangle
=
\fr1{(2\pi)^3}\sum_n e^{in\theta}\langle  \fr{\xi\otimes\xi}{|\xi|^2}
\ti\si(\xi),\ti\si(\xi)e^{in\xi}\rangle
\nonumber\\
\nonumber\\
\!\!\!\!\!\!\!\!\!\!\!\!\!\!\!\!\!\!\!\!&\!\!=\!\!&\!\!\!
\fr 1{(2\pi)^{3}}\langle
\fr{\xi\otimes\xi}{|\xi|^2}\ti\si(\xi), \ti\si(\xi)
\sum_n
e^{in(\theta+\xi)}\rangle
=
\sum_m
\Big[\fr{\xi\otimes\xi}{|\xi|^2}|\ti\si(\xi)|^2\Big]_{\xi=2\pi m-\theta}=\Si(\theta)~,
\quad\theta\in\Pi^*\setminus\Ga^*,
\eeqn
since the  last sum over $n$ equals $\ds|\Pi^*|\sum_m \de(\theta+\xi-2\pi m)$
by the Poisson summation formula \ci{Her1}. Finally, $|\ti\si(\xi)|\le C\langle
\xi\rangle^{-2}$ by (\ref{4-L123}). Hence,
\be\la{4-K4}
\Vert\hat T_1(\theta)\Vert\le \sum_m|
\ti\si(2\pi m-\theta)|^2\le C^2\sum_m\langle m\rangle^{-4}<\infty.
\ee
\end{proof}


\setcounter{equation}{0}
\section{The Hamiltonian structure and the domain}

In this section we study  the domain  of the generator $A$ given by (\ref{4-JDi}) and (\ref{4-AJB}).

\begin{definition}\la{4-dD}
i) $\cS_+:= \cup_{\ve>0} \cS_\ve$, where $\cS_\ve$  is the space of functions
$\Psi\in\cS(\R^3)$ whose Fourier transforms  $\hat\Psi(\xi)$ vanish in the $\ve$-neighborhood of the lattice $\Ga^*$,
\medskip\\
ii) $l_c$ is the space of sequences $q(n)\in R^3$ such that
$q(n)=0$, $n>N$ for some $N$.
\medskip\\
iii)
$  \cD:=\{Y=(\Psi_1,\Psi_2,q,p):  \Psi_1,\Psi_2\in \cS_+,~~~ q, p\in l_c \}.$
\end{definition}
Obviously, $\cD$ is dense in ${\cX^0}$.

\begin{theorem}\la{4-tHam}
Let conditions \eqref{4-L123i}  hold. Then
$B\cD\subset {\cX^0}$ and
$B$ is a symmetric operator on
the domain $\cD$.
\end{theorem}
\begin{proof}
Formally the matrix  (\ref{4-AJB}) is symmetric.
The following lemma implies that
$B$ is defined on $\cD$.
\begin{lemma}\la{4-lDB}
i) $H^0\Psi\in L^2(\R^3)$ for $\Psi\in \cS_+$.
\medskip\\
ii)
$G \Psi\in L^2(\R^3)$  and $S^*\Psi\in l^2$  for
$\Psi\in\cS_+$.
\medskip\\
iii) $S q\in L^2(\R^3)$ for $q\in l_c$.
\end{lemma}
\begin{proof}
i) $H^0\Psi(x):=-\fr12 \De\Psi(x)\in L^2(\R^3)$.
\medskip\\
ii) Given a fixed $\Psi\in\cS_+$, we have $\Psi\in\cS_\ve$ with some $\ve>0$, and
\be\la{4-Qp}
 G\Psi=F^{-1}\fr{\ti\Psi(\xi)}{|\xi|^2},
\ee
where $F$ stands for the Fourier transform.
Hence,
 $G\Psi\in L^2(\R^3)$.
\medskip\\
We now consider $S^*\Psi$. Applying (\ref{4-S}), (\ref{4-ppo}), and the Parseval identity, we get
for $\Psi\in\cS_\ve$
\beqn\la{4-conv2}
[S^*\Psi](n)&=&e\psi_0\int \Psi(x)G\na\si(x-n)dx
=e\psi_0\langle\Psi(x),G\na\si(x-n)\rangle
\nonumber\\
\nonumber\\
&=&\frac{ie\psi_0}{(2\pi)^{3}}\int_{|\xi|>\ve}\ti\Psi(\xi)
\fr{\xi\ov{\ti\si}(\xi)e^{-in\xi}}{|\xi|^2}d\xi.
\eeqn
Note that
$\pa^\beta\ti\Psi
\in \cS(\R^3)$
for all $\beta$.
Moreover, $\pa^\beta\ti\si\in L^2(\R^3)$ for $|\beta|\le 2$ by (\ref{4-L123}).
Hence, integrating by parts twice, we obtain
\be\la{4-conv4}
 |[S^*\Psi](n)|\le C\langle n\rangle^{-2},
\ee
which implies that $S^*\Psi\in l^2$.
\medskip\\
iii) Let us  check that $Sq\in L^2(\R^3)$ for $q\in l_c$.
Namely,
\beqn\la{4-SF}
\widetilde{Sq}(\xi)=
e \psi_0 F_{x\to\xi}\sum_n G\na\si(x-n)q(n)
=e\psi_0 \sum_n \fr{\xi\ov{\ti\si}(\xi)e^{-in\xi}}{|\xi|^2}q(n).
\eeqn
Hence, $\widetilde{Sq}\in L^2(\R^3)$ by (\ref{4-L123}).

\end{proof}

This lemma implies that $BY\in{\cX^0}$ for $Y\in \cD$.
The symmetry of $B$ on $\cD$ is evident from (\ref{4-AJB}).
Theorem \ref{4-tHam} is proved. \end{proof}

\bc
The operator $B$
with the domain $\cD$
is nonnegative by (\ref{Ham0}), and hence, 
it admits  a canonical  extension 
 to a selfadjoint operator in $\cX^0$
 by the Friedrichs  theorem \ci{RS2}.
\ec

\brs\la{4-cAA*}
{\rm
i) The `infrared singularity' at $\xi=0$ of the integrands \eqref{4-Qp}, \eqref{4-conv2} 
demonstrates that all operators
$G:L^2(\R^3)\to L^2(\R^3)$, $S^*:L^2(\R^3)\to l^2$, and $S:l^2\to L^2(\R^3)$  are unbounded.
\smallskip\\
ii)
The proof of
Theorem \ref{4-tHam} shows that
$A\cD\subset\cX^0$, and also $A^*\cD\subset\cX^0$, where
the `formal adjoint' $A^*$ is  defined by the identity
\be\la{4-A*}
\langle AY_1,Y_2\rangle=\langle Y_1,A^*Y_2\rangle,\qquad Y_1,Y_2\in \cD.
\ee
}
\ers

\setcounter{equation}{0}
\section{Generator in  the Fourier--Bloch transform}
We reduce the operators  $A$ and $B$ with the help of
the  Fourier--Bloch--Gelfand--Zak transform \ci{DK2005,PST,RS4}.

\subsection{The discrete Fourier transform}
Let us consider a vector
$  Y=(\Psi_1, \Psi_2, q, p)\in {\cX^0}$ and denote
\be\la{4-Yn2}
 Y(n)=(\Psi_1(n,\cdot),\Psi_2(n,\cdot), q(n),p(n))~,~~~~~~n\in\Z^3,
\ee
where
\be\la{4-YP}
\Psi_j(n,y)=
\Psi_j(n+y)~~{\rm for\ a.e.}~~y\in \Pi,\qquad j=1,2.
\ee
Obviously, $Y(n)$ with different $n\in\Z^3$ are orthogonal vectors in ${\cX^0}$, and besides,
\be\la{4-ort}
Y=\sum_n Y(n),
\ee
where the sum converges in ${\cX^0}$. The norms in ${\cY^0}$ and ${\cY^1}$ can be represented as
\be\la{4-nV}
 \Vert Y\Vert_{\cX^0}^2=\sum_{n\in\Z^3}\Vert Y(n)\Vert_{{\cX^0}(\Pi)}^2,\quad\quad
 \Vert Y\Vert_{\cX^1}^2=\sum_{n\in\Z^3}\Vert Y(n)\Vert_{{\cX^1}(\Pi)}^2,
\ee
where
\be\la{4-KV}
 {\cX^0}(\Pi):=L^2(\Pi)\oplus L^2(\Pi) \oplus \C^3\oplus \C^3,\quad\quad
 {\cX^1}(\Pi):=H^1(\Pi)\oplus H^1(\Pi) \oplus \C^3\oplus \C^3.
\ee
Further,
the periodic minimizer (\ref{4-gr}) is invariant with respect to translations of the lattice $\Ga$,
and hence the operator  $A$ commutes with these translations. Namely, (\ref{S}) implies that
\be\la{4-S2}
 S (x,n)=S(x-n,0),
\ee
 Similarly,  (\ref{T})  implies
that $T$ commutes with  translations of $\Ga$.
Hence,
$A$ can be reduced by the discrete Fourier transform
\be\la{4-F}
 \hat Y(\theta)=F_{n\to \theta}Y(n):=\sum\limits_{n\in \Z^3}e^{in\theta}Y(n)
 =(\hat{\Psi}_1(\theta,\cdot),\hat{\Psi}_2(\theta,\cdot),\hat {q}(\theta),\hat{p}(\theta))
 \quad~~{\rm for \ a.e.}~~\theta\in \R^3,
\ee
where
\be\la{4-FPsi}
 \hat \Psi_j(\theta,y)=\sum\limits_{n\in \Z^3}e^{in\theta}\Psi_j(n+y)
 \quad~~{\rm for \ a.e.}~~\theta\in \R^3,\quad ~~{\rm a.e.}~~y\in\R^3.
\ee
The function $\hat Y(\theta)$ is $\Ga^*$-periodic in $\theta$.
The series (\ref{4-F}) converges in $L^2(\Pi^*,{\cX^0}(\Pi))$, since the series
(\ref{4-ort}) converges in ${\cX^0}$.
The inversion formula  is given by
\be\la{4-FI}
  Y(n)=|\Pi^*|^{-1}\int_{\Pi^*}e^{-in\theta}\hat Y(\theta)d\theta
\ee
(cf. (\ref{4-Fu})).
The Parseval--Plancherel identity gives
\be\la{4-PP}
 \Vert Y\Vert_{\cX^1}^2=|\Pi^*|^{-1}\Vert\hat Y\Vert_{L^2(\Pi^*,{\cX^1}(\Pi))}^2,
\qquad
 \Vert Y\Vert_{\cX^0}^2=|\Pi^*|^{-1}\Vert\hat Y\Vert_{L^2(\Pi^*,{\cX^0}(\Pi))}^2.
\ee
The functions $\hat {\Psi}_j(\theta,y)$ are  $\Ga$-quasiperiodic in $y$; i.e.,
\be\la{4-qp}
 \hat{\Psi}_j(\theta,y+m)=e^{-im\theta}\hat{\Psi}_j
(\theta,y), \qquad m\in\Z^3.
\ee

\subsection{Generator in the discrete Fourier transform}
Let us consider $Y\in\cD$ and calculate
the Fourier transform (\ref{4-F}) for $AY$ given by (\ref{4-JDi}).
Using  (\ref{4-T}), (\ref{4-conv2}), (\ref{4-S2}), 
we obtain
\be\la{4-CPFh}
 \widehat{AY}(\theta)=\hat A(\theta){\hat Y}(\theta)\qquad
~~{\rm for \ a.e.}~~\theta\in\R^3\setminus \Ga^*,
\ee
where $\hat A(\theta)$ is a  $\Ga^*$-periodic operator function,
\be\la{4-tiAh}
 \hat A(\theta)=\left(\!\ba{cccl}
 0 & H^0 & 0 &  0\medskip\\
 -H^0-2e^2\psi_0\hat G(\theta)\psi_0  & 0  &
 \hat S(\theta) &  0\\
 0  &  0 &  0 & M^{-1}\\
 -2\hat S^{\5*}(\theta)&0  &-\hat T(\theta)&0\\
\ea\!\right)
\ee
by \eqref{4-JDi}.
Here
\be\la{4-qp4}
\hat G(\theta)\hat\Psi(\theta,y)=\sum_m\fr{\check\Psi(\theta,m)}{(2\pi m+\theta)^2}
e^{-i2\pi m y}\qquad~~{\rm for \ a.e.}~~
\theta\in\R^3\setminus \Ga^*,
\ee
where
\be\la{psich}
\check\Psi(\theta,m)=\int_{T^3}e^{i2\pi m x}\hat\Psi(\theta,x)dx/
\ee
This expression  is well-defined for
$\Psi\in\cS_\ve$, since
\be\la{4-qp5}
\check\Psi(\theta,m)=\ti\Psi(2\pi m+\theta)=0
\quad {\rm for}\quad|2\pi m+\theta|<\ve.
\ee

\begin{lemma}\label{SST}
Let (\ref{4-L123i})
 hold.
Then the operator $\hat S(\theta)$ acts as follows:
\begin{equation}\label{4-S-act}
 \widehat {S q}(\theta)
 =\hat S(\theta)\hat q(\theta),
 \quad {\rm where}\quad \hat S(\theta)=e\psi_0\hat G(\theta)\na\hat\sigma(\theta,y).
\end{equation}
\end{lemma}
\begin{proof}
For $x=y+n$  equations \eqref{4-ro+2} and \eqref{4-S} imply
\beqn\nonumber
S q(y+n)&=&e\psi_0\sum_m G\na\sigma^0(m,y+n)q(m)\\
\nonumber
&=&e\psi_0\sum_mG\na\sigma(y+n-m)q(m).
\eeqn
Applying  the Fourier transform (\ref{4-F}), we obtain \eqref{4-S-act}.
\end{proof}
Furthermore, $\hat S^{\5*}(\theta)$ in (\ref{4-tiAh}) is the corresponding adjoint operator,
and   $\hat T(\theta)$ is the operator matrix expressed by (\ref{4-K3}).
Note that $\hat S(\theta)$, $\hat S^{\5*}(\theta)$ and  $\hat T(\theta)$
are finite-rank operators.

\subsection{Generator in the Bloch transform}

\begin{definition}\la{4-FGLZ}
The Bloch transform of $Y\in{\cX^0}$ is defined as
\be\la{4-YPi}
 \ti Y(\theta)=[\cF Y](\theta):= \cM(\theta)\hat Y(\theta):
 =(\ti{\Psi}_1(\theta,\cdot),\ti{\Psi}_2(\theta,\cdot),\hat{ q}(\theta),\hat p(\theta))
 \qquad~~{\rm for  a.e.}~~\theta\in\R^3,
\ee
where $\ti{\Psi}_j(\theta,y)=M(\theta)\hat{\Psi}_j:=e^{i\theta y}\hat{\Psi}_j(\theta,y)$
are  $\Ga$-periodic functions in $y\in\R^3$.
\end{definition}
Now the Parseval-Plancherel identities (\ref{4-PP}) read
\be\la{4-PPt}
 \Vert Y\Vert_{\cX^1}^2= |\Pi^*|^{-1}\Vert\ti Y\Vert_{L^2(\Pi^*,{\ti\cX^1}(T^3))}^2,
 \qquad
 \Vert Y\Vert_{\cX^0}^2 = |\Pi^*|^{-1}\Vert\ti Y\Vert_{L^2(\Pi^*,{\ti\cX^0}(T^3))}^2.
\ee
Hence, $\cF:{\cX^0}\to  L^2(\Pi^*,{\cX^0}(T^3))$ is an isomorphism.
The inversion is given by
\be\la{4-FZI}
Y(n)=|\Pi^*|^{-1}\int_{\Pi^*}e^{-in\theta}\cM(-\theta)\ti Y(\theta)d\theta,\qquad n\in\Z^3.
\ee
Finally, the  above calculations can be summarised as follows:
(\ref{4-CPFh}) implies that, for $Y\in\cD$,
\be\la{4-CPF}
\widetilde{AY}(\theta)=
\ti A(\theta) \ti Y(\theta)\quad~~{\rm for \ a.e.}~~\theta\in  \Pi^*\setminus\Ga^*.
\ee
Here
\be\la{4-tiA}
\ti A(\theta)\!=\!\cM(\theta)\hat A(\theta)\cM(-\theta)
 \!=\!\!\left(\!\!\!
 \ba{cccl}
0 &\!\!\ti H^0(\theta) & 0  & 0\medskip\\
 -\ti H^0(\theta)-2e^2\psi_0\ti G(\theta)\psi_0&0 &~\ti S(\theta) & 0\\
 0  & 0&\!0&\!M^{-1}\\
 -2\ti S^*(\theta)\!\!&  0 &-\hat T(\theta)&\!0\\
\ea\!\!\!\!\right),
\ee
where
\beqn\la{4-tiHS}
\ti S(\theta)&:=& M(\theta)\hat S(\theta)=
e\psi_0\ti G(\theta)\na\ti\sigma^0(\theta),\\
\nonumber\\
\ti H^0(\theta)&:=&M(\theta) H^0M(-\theta)=
-\fr 12(\na-i\theta)^2,\la{4-tiH0}\\
\nonumber\\
\la{4-tiH1}
\ti G(\theta)&:=&M(\theta) \hat G(\theta) M(-\theta)=(i\na+\theta)^{-2}.
\eeqn
Formula (\ref{4-CPF}) is obtained for $Y\in\cD$.
Respectively, the operator (\ref{4-tiA}) is considered on the space
$\cD(T^3):=C^\infty(T^3)\oplus C^\infty(T^3)\oplus\C^3\oplus\C^3$
up to now. The operator (\ref{4-tiA})
 extends uniquely
to the  continuous
 operator $\ti\cX^2(T^3)\to \ti\cX^0(T^3)$ for $\theta\in \Pi^*\setminus\Ga^*$.
We keep below the notation (\ref{4-tiA})--(\ref{4-tiH1}) for this extension.

\begin{remark}\la{4-rb} \rm
The operators $\ti G(\theta):L^2(T^3)\to H^2(T^3)$ are bounded for
$\theta\in \Pi^*\setminus\Ga^*$; however
$\Vert \ti G(\theta)\Vert \sim d^{-2}(\theta)$, where $d(\theta):=\dist(\theta,\Ga^*)$.

\end{remark}

\begin{lemma}\la{4-lB} Let conditions \eqref{4-L123i}  hold. Then
the operator  $\ti A(\theta)$ admits the representation
\be\la{4-Has}
\ti A(\theta)=J\ti B(\theta),~~~~~~~~~~\theta\in \Pi^*\setminus\Ga^*,
\ee
where $\ti B(\theta)$ is the selfadjoint operator (\ref{4-hess2i})
in ${\ti\cX^0}(T^3)$
with the domain ${\ti\cX^2}(T^3)$.

\end{lemma}
\begin{proof}
The representation (\ref{4-Has}) follows from  (\ref{4-AJB}).
The operator $\ti B(\theta) $
is symmetric on the domain  $\ti\cX^2(T^3)$.
Moreover,  all operators in (\ref{4-hess2i}), except for $\ti H^0(\theta)$, are bounded.
Finally,  $\ti H^0(\theta)$
is  selfadjoint in $L^2(T^3)$ with the domain $H^2(T^3)$.
Hence, $\ti B(\theta)$ is also selfadjoint on  the domain 
$\ti\cX^2(T^3)$.
\end{proof}

\setcounter{equation}{0}
\section{The positivity  of energy}\la{sEP}
Here we prove the  positivity \eqref{4-Hpos2}  under conditions \eqref{4-W1} and \eqref{4-Wai}.
Recall that
\be\la{4-TT}
\hat T(\theta)=\Si(\theta),\qquad\theta\in\Pi^*\setminus\Ga^*
\ee
by (\ref{4-K3}).
This matrix  is a continuous function of $\theta\in\Pi^*\setminus\Ga^*$.
Let us denote the {\it open set}
\be\la{4-Pi+}
\Pi^*_+:=\{\theta\in \Pi^*\setminus\Ga^*: \Si(\theta)>0 \}.
\ee
Then the Wiener condition  \eqref{4-W1} means that $|\Pi^*_+|=|\Pi^*|$.
The main result of this chapter is the following theorem.

\bt\la{4-tpose}
Let  conditions \eqref{4-Wai} and \eqref{4-L123i} hold.
Then 
\medskip\\
i) The Wiener  condition  \eqref{4-W1} is necessary
and sufficient for the positivity  \eqref{4-Hpos2k}.
\smallskip\\
ii) The bound (\ref{4-vkat2}) holds.
\smallskip\\
iii)  Bound (\ref{4-vkat}) holds with  sufficiently small $\ve>0$ under the Wiener condition  \eqref{4-W1}.

\et
\begin{proof}
First, let us check that the Wiener condition \eqref{4-W1} is necessary.
Namely, 
for $\ti Y=(0,0,\hat q,0)\in {\ti\cX^1}(T^3)$
the inequality   \eqref{4-Hpos2k} implies that
\be\la{4-Hpos0}
 \cE(\theta,\ti Y)
=\hat q\hat T(\theta)\hat q
\ge \vka(\theta)|\hat q|^2\quad~~\mbox{\rm for \ a.e.}~~\theta\in\Pi^*\setminus\Ga^*.
\ee
Now (\ref{4-TT}) gives
\be\la{4-Hpos0b}
 \cE(\theta,\ti Y)
=
\hat q\Si(\theta)\hat q
\ge \vka(\theta)|\hat q|^2.
\ee
Hence, the  condition  \eqref{4-W1}
is necessary
for the positivity  \eqref{4-Hpos2k}.
Moreover, \eqref{4-Hpos0b}  implies \eqref{4-vkat2}.
\medskip

It remains to show  that the Wiener condition \eqref{4-W1} together with \eqref{4-Wai}
is sufficient for the bound \eqref{4-vkat}.
Substituting (\ref{4-ro1i}) into (\ref{Ham0}), we obtain 
for  $Y=(\Psi_1,\Psi_2,q,p)\in \cD$  
\be\la{4-bb2}
\langle Y,BY \rangle=2\sum_{j=1}^2\langle \Psi_j,H^0\Psi_j \rangle
\!+\!
\langle
2ef\Psi_1+gq,
2ef\Psi_1+gq
\rangle,
\ee
where 
 \be\la{4-fg}
f\Psi_1(x):=e\psi_0\sqrt{G}\Psi_1,\qquad  gq(x)={\sqrt{G}}\sum_n\na\si(x-n)q(n).
\ee
The operators (\ref{4-fg}) commute with the $\Ga$-translations,
and therefore similarly to (\ref{4-CPF})
and (\ref{4-tiHS})--(\ref{4-tiH1})
\be\la{4-fact}
\widetilde {f\Psi_1}(\theta):=e\psi_0\sqrt{\ti G(\theta)}\ti\Psi_1(\theta),\qquad
\widetilde{gq}(\theta)={\sqrt{\ti G(\theta)}}(\na-i\theta)\ti\si(\cdot,\theta)\ti q(\theta).
\ee
Hence,
 the (second) Parseval identity (\ref{4-PPt})  gives
\be\la{4-bb2-2}
|\Pi^*|\langle Y,BY \rangle=2\sum_{j=1}^2( \ti\Psi_j,\ti H^0\ti \Psi_j )
\!+\!
(
2e\ti f\ti \Psi_1+\ti g\ti q,
2e\ti f\ti \Psi_1+\ti g\ti q
),
\ee
where the brackets $(\cdot,\cdot)$ denote the 
 scalar product in $L^2(\Pi^*,\ti\cX^0(\Pi))$.
 Hence,  
 \be\la{4-bb2t}
\cE(\theta,\ti Y):=
 \langle \ti  Y,\ti B(\theta)\ti  Y \rangle_{{\ti\cX^0}(T^3)}\!=\!E(\theta, \ti\Psi_1,\hat q)+
 2\langle  \ti\Psi_2,\ti H^0(\theta) \ti\Psi_2 \rangle_{L^2(T^3)}
+ \hat pM^{-1} \hat p
\ee
for  $\ti Y=(\ti\Psi_1, \ti\Psi_2, \hat q, \hat p)\in \ti\cX^2(T^3)$, where
\be\la{4-sq}
E(\theta, \ti\Psi_1,\hat q)
:=
2\langle  \ti\Psi_1,\ti H^0(\theta) \ti\Psi_1 \rangle_{L^2(T^3)}+
\langle 2\ti f(\theta)   \ti\Psi_1 + \ti g(\theta) \hat q, ~2\ti f(\theta)   \ti\Psi_1+\ti g(\theta) \hat q \rangle_{L^2(T^3)}.
\ee
Recall that $\ti H^0(\theta)=
-\ds\fr12(\na-i\theta)^2$ by (\ref{4-tiH0}), so
the eigenvalues of $\ti H^0(\theta)$ equal to $\ds\fr12|2\pi m-\theta|^2$, where
$m\in \Z^3$. Therefore,
$\ti H^0(\theta)$ is positive definite: for $j=1,2$
\be\la{4-H0d}
\langle\ti\Psi_j, \ti H^0(\theta)\ti\Psi_j \rangle\ge \fr12 d^2(\theta)\Vert  \ti\Psi_j\Vert_{H^1(T^3)}^2~,
\qquad
\theta\in\Pi^*\setminus \Ga^*.
\ee
It remains to prove the following lemma.

\bl\la{4-p}
Let conditions of Theorem \ref{4-tpose} hold. Then for any $\theta\in\Pi^*_+ $ there exists $\ve_1>0$
such that
\be\la{4-ub2}
E(\theta,  \ti\Psi_1,\hat q)
\ge
\fr12 d^2(\theta)\Vert\ti\Psi_1\Vert_{H^1(T^3)}^2+\ve_1 d^4(\theta) \Si_0(\theta)|\hat q|^2.
\ee

\el
\begin{proof}
Let us denote
\be\la{4-abb}
\beta_{11}:=
\langle 2\ti f(\theta)   \ti\Psi_1,
2\ti f(\theta)   \ti\Psi_1\rangle_{L^2(T^3)},
\quad
\beta_{12}:=
 \langle 2\ti f(\theta)   \ti\Psi_1, \ti g(\theta) \hat q \rangle_{L^2(T^3)},
\quad
\beta_{22}:=
\langle \ti g(\theta) \hat q, \ti g(\theta) \hat q \rangle_{L^2(T^3)}.
\ee
Then we can write the quadratic form (\ref{4-sq}) as
\be\la{4-sq2}
E=2\al+\beta,
\ee
where $\al:=\langle  \ti\Psi_1,\ti H^0(\theta) \ti\Psi_1 \rangle_{L^2(T^3)}\ge 0$ and
\be\la{4-sq22}
\beta:=
\beta_{11}+2\rRe\beta_{12}+\beta_{22}=
\langle 2\ti f(\theta)   \ti\Psi_1 + \ti g(\theta) \hat q, ~2\ti f(\theta)   \ti\Psi_1+\ti g(\theta) \hat q \rangle_{L^2(T^3)}
\ge 0.
\ee
By (\ref{4-H0d})
it suffices to prove the estimate
\be\la{4-sq23}
E\ge \al+\ve d^4(\theta) \beta_{22},
\ee
since
\be\la{4-b22}
\beta_{22}=
\hat q\hat T_1(\theta)\hat q= \hat q\Si(\theta)\hat q
\ee
by (\ref{4-fact}) and (\ref{4-TT}).
To prove (\ref{4-sq23}), we first note that 
\be\la{4-abb3}
\al\ge \ve_2d^4(\theta)
\beta_{11},\qquad \theta\in\Pi^*\setminus\Ga^*,
\ee
where $\ve_2>0$.
Indeed, using (\ref{4-tiH1})  and applying (\ref{4-H0d}), we obtain
that
\be\la{4-b23}
\beta_{11}=4\langle\ti\Psi_1, \ti G(\theta)\ti\Psi_1\rangle
\le \fr C{d^2(\theta)}\Vert\ti\Psi_1\Vert^2_{L^2(T^3)}
\le\fr{C_1}{d^4(\theta)}\al,
\ee
Now  (\ref{4-sq2}) gives that
\be\la{4-abb4}
E\ge\al+
(1+ \ve_2 d^4(\theta))
\beta_{11}+2\rRe\beta_{12}+\beta_{22},\qquad \theta\in\Pi^*\setminus\Ga^*.
\ee
On the other hand, the Cauchy--Schwarz inequality implies that
\be\la{4-abb5}
|\beta_{12}|\le \beta_{11}^{1/2}\beta_{22}^{1/2}\le \fr12[\ga\beta_{11}+\fr1\ga\beta_{22}]
\ee
for any $\ga>0$.
Hence, (\ref{4-abb4}) implies that
\be\la{4-abb7}
b\ge\al+
(1+ \ve_2 d^4(\theta)-\ga)
\beta_{11}+(1-\fr1{\ga})\beta_{22},\qquad \theta\in\Pi^*\setminus\Ga^*.
\ee
Choosing $\ga= 1+ \ve_2 d^4(\theta)$,
we obtain (\ref{4-sq23}).
\end{proof}

At last, formula (\ref{4-bb2t}) and estimates (\ref{4-H0d}), (\ref{4-ub2}) imply
(\ref{4-vkat}) with sufficiently small $\ve>0$.\end{proof}

\bc\la{4-cK}
Bound  (\ref{4-vkat}) implies that
(\ref{4-Hpos2k}) holds with
\be\la{4-com}
\inf_{\theta\in K}\vka(\theta)>0
\ee
for any compact subset $K\subset\Pi^*_+$.
\ec

\br\la{4-rS} \rm
Lemma \ref{4-p} and its proof were inspired by the Sylvester
criterion for the positivity of $2\times 2$ matrices.
Namely, in notation (\ref{4-abb}) for the matrix $\beta=(\beta_{ij})$
we have $\beta_{11}\ge 0$, $\beta_{22}> 0$. Furthermore,
the matrix $\beta\ge 0$, since it corresponds to the perfect square, and hence
$\operatorname{det} \beta\ge 0$. Therefore, the Sylvester
criterion implies that
\be\la{4-alb}
\beta_+:=\left(
\ba{cc}
\al+\beta_{11}&\beta_{12}\\
\beta_{21}& \beta_{22}
\ea
\right)>0
\ee
since $\al+\beta_{11}> 0$, $\beta_{22}> 0$ and $\det\beta_+=\al\beta_{22}+\det\beta>0$.
These arguments are behind
 our estimates (\ref{4-abb4})--(\ref{4-abb7}), which give (\ref{4-sq23}).
\er

\setcounter{equation}{0}
\section{Weak solutions and linear stability}\la{4-sred}


We introduce weak solutions and prove the  linear stability of the dynamics
 \eqref{4-JDi} assuming (\ref{4-L123i}), (\ref{4-W1}) and (\ref{4-Wai}). Then
the  real periodic minimizer is given by
 (\ref{4-ppo}) with $\al=0$,
 and (\ref{4-Hpos2k}) and (\ref{4-vkat}) hold by Theorem \ref{4-tpose}.

\subsection{Weak solutions }
Let us define solutions $Y(t)\in C(\R,\cX^1)$
to \eqref{4-JDi} in the sense of vector-valued distributions of $t\in\R$.
Let us recall that $A^*V\in{\cX^0}$ for $V\in\cD$ by Corollary \ref{4-cAA*}.
We call
$Y(t)$ a weak solution to \eqref{4-JDi} if, for every $V\in\cD$,
\be\la{4-ws}
\langle Y(t)-Y(0),V\rangle= \int_0^t\langle Y(s),A^*V \rangle ds,\qquad t\in\R.
\ee
Equivalently,
by  the  Parseval--Plancherel identity,
\be\la{4-wsF}
\int_{\Pi^*}\langle \ti Y(\theta,t)-\ti Y(\theta,0),\ti V(\theta)\rangle_{{\ti\cX^0}(T^3)} d\theta= \int_0^t\Big[\int_{\Pi^*}
\langle \ti Y(\theta,s),\ti A^*(\theta)\ti V(\theta) \rangle_{{\ti\cX^0}(T^3)} d\theta \Big]ds
\ee
Fubini's theorem implies that
\be\la{4-FF}
\ti Y(\theta,\cdot)\in L^1_{\rm loc} (\R,\ti\cX^1(T^3))\qquad\mbox{for a.e.}\quad \theta\in \Pi^*,
\ee
and
(\ref{4-wsF}) is equivalent to
\be\la{4-wsF2}
\int_{\Pi^*}\langle \ti Y(\theta,t)-\ti Y(\theta,0),\ti V(\theta)\rangle_{{\ti\cX^0}(T^3)} d\theta= \int_{\Pi^*}\Big[\int_0^t
\langle \ti Y(\theta,s),\ti A^*(\theta)\ti V(\theta) \rangle_{{\ti\cX^0}(T^3)}ds\Big] d\theta.
\ee
Equivalently,
\be\la{4-wsF3}
 \langle \ti Y(\theta, t)-\ti Y(\theta, 0),\ti V\rangle_{{\ti\cX^0}(T^3)} \!=\! \int_0^t
\langle \ti Y(\theta,s),\ti A^*(\theta)\ti V \rangle_{{\ti\cX^0}(T^3)}  ds,\qquad t\in\R,\qquad~~
\ti V\!\in\! \cD(T^3)
\ee
for a.e. $\theta\in  \Pi^*\setminus \Ga^*$.
Formally,
\be\la{4-JD}
\dot {\ti Y}(\theta, t)= \ti A(\theta) \ti Y(\theta,t),\qquad t\in\R
\ee
for a.e. $\theta\in  \Pi^*\setminus \Ga^*$
in the sense of vector-valued distributions.

\subsection{Reduction to mild solution}
We reduce  (\ref{4-JD})  to an equation with
a selfadjoint generator by using \eqref{4-Hpos2k}  and our methods \ci{KK2014a,KK2014b}.
By  \eqref{4-Hpos2k} and (\ref{4-vkat})  the operator
$
\ti\Lam(\theta):=\ti B^{1/2}(\theta)>0
$
is invertible in  ${\ti\cX^0}(T^3)$  for   $\theta\in\Pi^*_+$ and
\be\la{4-nVet}
 \Vert \ti\Lam^{-1}(\theta)Z\Vert_{{\ti\cX^1}(T^3)}\le \fr 1{\sqrt{ \vka(\theta)}} \Vert Z\Vert_{{\ti\cX^0}(T^3)},
 \quad Z\in{\ti\cX^0}(T^3),\qquad  \theta\in\Pi^*_+.
\ee
Hence, $\ti A(\theta)=J\ti B(\theta)$ and $\ti A^*(\theta)=-\ti B(\theta)J$ are
also invertible in  ${\ti\cX^0}(T^3)$. Therefore,  \eqref{4-wsF3} can be rewritten as
\be\la{4-wsF32}
 \langle \ti Y(\theta, t)-\ti Y(\theta, 0),(\ti A^*(\theta))^{-1}\ti W\rangle_{{\ti\cX^0}(T^3)}\!=\! \int_0^t
\langle \ti Y(\theta,s),\ti W \rangle_{{\ti\cX^0}(T^3)}  ds,\qquad t\in\R,\qquad~~
\ti W\!\in\! \ti A^*(\theta)\cD(T^3)
\ee
for a.e. $\theta\in\Pi^*_+$.
\bl
The linear space $\ti A^*(\theta)\cD(T^3)$ is dense in ${\ti\cX^0}(T^3)$.

\el
\begin{proof} First, $\ti A^*(\theta)\cD(T^3)=\ti B(\theta)\cD(T^3)$, since $J\cD(T^3)=\cD(T^3)$.
Second, $\ti B(\theta)$, which is defined on $\cD(T^3)$, extends to
an invertible  selfadjoint operator in ${\ti\cX^0}(T^3)$
with the domain  ${\ti\cX^2}(T^3)$ and  $\Ran \ti B(\theta)={\ti\cX^0}(T^3)$.
\end{proof}

As a corollary,  \eqref{4-wsF32} is equivalent to the `mild solution' identity
\be\la{4-wsF33}
 \ti A^{-1}(\theta) [\ti Y(\theta, t)-\ti Y(\theta, 0)]
 =\! \int_0^t
 \ti Y(\theta,s)  ds,\qquad t\in\R\qquad \mbox{for a.e. }\theta\in \Pi^*_+.
\ee
\subsection{Reduction to selfadjoint generator}
Now we can apply our approach \ci{KK2014a} to reduce
the Hamiltonian system
(\ref{4-JDi})
to the dynamics with a selfadjoint generator.
Namely, 
   (\ref{4-FF}) implies that
\be\la{4-ZY}
Z(\theta,\cdot):=\ti\Lam(\theta)\ti Y(\theta,\cdot)\in L^1_{\rm loc} (\R,\ti\cX^0(T^3))
\qquad\mbox{for a.e.}\quad \theta\in \Pi^*.
\ee
Hence, applying $\ti\Lam(\theta)$ to the both sides of
 \eqref{4-wsF33}, we obtain the equivalent equation
 \be\la{4-wsF34}
  \ti K^{-1}(\theta) [\ti Z(\theta, t)-\ti Z(\theta, 0)]
 =-i \int_0^t
 \ti Z(\theta,s)  ds,\qquad t\in\R\qquad\mbox{for a.e.}\quad \theta\in \Pi^*,
\ee
 where $\ti K(\theta):=i\ti\Lam(\theta) J\ti\Lam(\theta)$, since $\ti A^{-1}(\theta)=\ti\Lam^{-2}(\theta)J^{-1}$.
Formally,
\be\la{4-CPF3}
 \dot {\ti Z}(\theta,t)=-i\ti K(\theta)\ti Z(\theta,t),\quad t\in\R\qquad\mbox{for a.e.}\quad \theta\in \Pi^*
\ee
in the sense of vector-valued distributions.
Now the problem  is that the domain of $\ti K(\theta)$ is unknown if the ion density
$\si(x)$ is not sufficiently smooth, so we cannot use the PDO techniques.
The following lemma plays a~key role in our approach (cf.\ Lemma 2.1 of \ci{KK2014a}).
\begin{lemma}\la{4-lH0}
i) $\ti K(\theta)$ is a selfadjoint operator in ${\ti\cX^0}(T^3)$ with a dense domain  $D_\theta=D(\ti K(\theta))\subset {\ti\cX^1}(T^3)$
for every $\theta\in\Pi^*_+$.
\medskip\\
ii) The eigenvectors of  $\ti K(\theta)$ form a complete set in ${\ti\cX^0}(T^3)$.
\end{lemma}

\begin{proof}
i) The operator $\ti K(\theta)$ is injective. On the other hand, $\Ran\5\ti\Lambda(\theta)={\ti\cX^0}(T^3)$, and
$J:{\ti\cX^0}(T^3)\to{\ti\cX^0}(T^3)$ is a bounded invertible operator.
Hence, $\Ran\5\ti  K(\theta)={\ti\cX^0}(T^3)$. Consider the inverse operator
\begin{equation}\la{4-G}
 \ti R(\theta):=\ti K^{-1}(\theta)=i\ti\Lam^{-1}(\theta) J^{-1}\ti\Lam^{-1}(\theta).
\end{equation}
This operator is selfadjoint, since it is bounded and symmetric. Hence,
$\Ran\5 \ti K(\theta)=D(\ti R(\theta))={\ti\cX^0}(T^3)$.
Therefore, $\ti K(\theta)=\ti R^{-1}(\theta)$ is a densely defined selfadjoint operator
by Theorem 13.11, (b) of \cite{Rudin}:
$$
\ti K^*(\theta)=\ti K(\theta)~, \quad D(\ti K(\theta))=\Ran\5 \ti R(\theta)\subset \
\Ran \5\ti\Lam^{-1}(\theta)\subset {\ti\cX^1}(T^3)
$$
where the last inclusion follows by (\ref{4-nVet}).
\medskip\\
ii) (\ref{4-nVet}) implies that $\ti\Lam^{-1}(\theta)$ is a compact operator in
${\ti\cX^0}(T^3)$ by the Sobolev embedding theorem. Hence, $\ti K^{-1}(\theta)$ is also
compact operator in  ${\ti\cX^0}(T^3)$ by (\ref{4-G}).
\end{proof}

This lemma implies that the formula
\be\la{4-ZH0}
 \ti Z(\theta,t)=e^{-i \ti K(\theta) t}\ti Z(\theta,0) \in   C_b(\R,{\ti\cX^0}(T^3))
 \ee
gives a unique solution to (\ref{4-CPF3}) for each $\theta\in \Pi^*_+$
and every $\ti Z(\theta,0)\in{\ti\cX^0}(T^3)$.
Indeed, it suffices to expand $Z(\theta,t)$ in the eigenvectors of $\ti K(\theta)$
and to note that (\ref{4-wsF34}) gives ordinary differential equations for each component.
Now we can prove the well posedness of the Cauchy problem
for  equation \eqref{4-JD} with any $\theta\in\Pi^*_+$.

\bt\la{4-tdt}
Let  conditions \eqref{4-W1},
\eqref{4-Wai} and \eqref{4-L123i} hold,
 and $\theta\in\Pi^*_+$. Then,
for every initial state $\ti Y(\theta,0)\in {\ti\cX^1}(T^3)$, there exists a
 unique solution $\ti Y(\theta,\cdot)\in C_b(\R,{\ti\cX^1}(T^3))$ to equation \eqref{4-JD}
in the sense of~\eqref{4-wsF3}. Besides,
\be\la{4-econ}
 \langle\ti \Lam(\theta) \ti Y(\theta,t),
\ti \Lam(\theta)\ti Y(\theta,t)\rangle_{{\ti\cX^0}(T^3)}=C(\theta),
\qquad t\in\R.
\ee
\et
\begin{proof} First, we note that
$\ti Z(\theta,0):=\ti\Lam(\theta)\ti Y(\theta,0)\in {\ti\cX^0}(T^3)$. Hence,
(\ref{4-ZH0})
and (\ref{4-nVet})
imply that
\be\la{4-YH0}
  \ti Y(\theta,t)=\ti\Lam^{-1}(\theta) e^{-i K(\theta) t} \ti Z(\theta,0) \in   C_b(\R,{\ti\cX^1}(T^3))
\ee
is the unique solution to \eqref{4-JD}. Finally,
$$
\langle\ti \Lam(\theta)\ti  Y(\theta,t),\ti \Lam(\theta)\ti Y(\theta,t)\rangle_{{\ti\cX^0}(T^3)}=
\langle\ti Z(\theta,t),\ti Z(\theta,t)\rangle_{{\ti\cX^0}(T^3)}=C(\theta),\qquad t\in\R,
$$
 since  $e^{-i K(\theta) t}$  is  the unitary group
in ${\ti\cX^0}(T^3)$.
\end{proof}

\subsection{Linear stability in the energy space}

Thus, we have constructed $\ti  Y(\theta,t)$ uniquely 
for a.e.\ $\theta\in\Pi^*_+$. However,  (\ref{4-econ})
does not imply that there exists the corresponding
$Y(t)\in \cX^1$, since $\ti\Lam(\theta)$ can degenerate at
some points $\theta\in\Pi^*\setminus\Pi^*_+$. In particular,
it degenerates at
$\theta=0$ due to (\ref{4-vkat2}) and (\ref{4-ask}).
Thus, we need another phase space to construct solutions to (\ref{4-ws}).
Let us denote
$$
\cD_0:=\{Y\in{\cX^1}: \ti Y(\theta)=0 ~~\mbox{in a neighborhood of}~~\Ga^*\}.
$$
Lemma \ref{4-lB} implies that
$\ti \Lam(\theta)\ti Y(\theta)\in L^2(\Pi^*_+,{\ti\cX^0}(T^3))$
for
 $Y\in\cD_0$. Moreover, Theorem \ref{4-tpose} shows that
 \be\la{4-nW}
\Vert Y \Vert_\cW:=\Vert \ti\Lam(\theta)\ti Y(\theta) \Vert_{L^2(\Pi^*_+,{\ti\cX^0}(T^3))}=
\langle Y,BY\rangle
> 0,\qquad
Y\in\cD_0\setminus 0
\ee
 under conditions \eqref{4-W1}, \eqref{4-Wai} and (\ref{4-L123i}).
Hence, $\Vert Y \Vert_\cW$ is a~norm on $\cD_0$.

\bd\la{4-dW}
The Hilbert space $\cW$ is the completion of $\cD_0$ in the norm $\Vert Y \Vert_\cW$.

\ed
Formally, we have $\Vert Y\Vert_\cW=\langle Y,BY\rangle^{1/2}$.
By Corollary \ref{4-cK},
the Fourier--Bloch transform (\ref{4-YPi}) extends to the isomorphism
\be\la{4-FBW}
\cF:\cW\to \ti\cW:=\{\ti Y(\cdot)\in L^2_{\rm loc}(\Pi^*_+,{\ti\cX^0}(T^3)):
\Vert \ti\Lam(\theta)\ti Y(\theta) \Vert_{L^2(\Pi^*_+,{\ti\cX^0}(T^3))}<\infty\}.
\ee
Hence, we can extend the definition of weak solutions
(\ref{4-ws})
to  $Y(t)\in C(\R,\cW)$
by identity  \eqref{4-ws} with $V\in\cD$ such that
\be\la{4-VV}
\operatorname{supp}\ti V(\theta)\subset \Pi^*_+.
\ee

Theorem  \ref{4-tdt} has the following corollary, which is one of  main results of this chapter.

\bc\la{4-cd}
Let all conditions of Theorem \ref{4-tdt} hold. Then,
for every initial state $Y(0)\in \cW$, there exists a unique weak
 solution $Y(\cdot)\in C_b(\R,\cW)$ to equation \eqref{4-JDi}, the energy norm being conserved{\rm :}
\be\la{4-econ2}
 \Vert  Y(t)\Vert_\cW=\const,\qquad t\in\R.
\ee
The solution is given by formula \eqref{4-YH0}:
\be\la{4-tZ0}
 Y(t)=\cF^{-1}\ti\Lam^{-1}(\theta) e^{-i \ti K(\theta) t} \ti\Lam(\theta)\ti Y(\theta,0) \in   C_b(\R,\cW).
\ee
The energy conservation (\ref{4-econ2}) follows from  (\ref{4-econ}) by integration over $\theta\in\Pi^*_+$.

\ec
This means that the linearized dynamics
\eqref{4-JDi}
is stable in the `energy space' $\cW$: a~global
solution exists and is unique
for each initial state of finite energy, and the `energy norm' is constant in time.


\chapter{Dispersive decay}

\centerline{Abstract}
\medskip

We establish the dispersive decay for the linearised system 
(\ref{4-LPS1Li})
 assuming  that the ion charge density $\si(x)$
 sasisfies 
 i)  the Wiener and Jellium conditions
 (\ref{4-Wai}) and  (\ref{4-W1}), and 
 ii) it decays exponential  at infinity.
The corresponding examples are given.

The dispersion relations are introduced via
spectral resolution for the non-selfadjoint Hamiltonian generator 
in the Bloch representation
using the positivity of the energy  (\ref{4-Hpos2k}).

The main result
of this chapter
 is 
 the dispersive decay
in the weighted Sobolev norms
for solutions with initial states from the space of  continuous spectrum of the Hamiltonian generator.
We also prove
the absence of singular spectrum and 
limiting absorption principle.
The multiplicity of every eigenvalue is shown to be infinite.

The proofs rely on novel exact bounds and compactness
for the inversion of the Bloch generators and on 
uniform asymptotics for the dispersion relations
using the energy positivity (\ref{4-Hpos2k}).
We also use the  theory of analytic sets.

\section{Introduction}
We establish the dispersive decay for the linearised system 
(\ref{4-JDi}):
\be\la{JDi}
\dot Y(t)=AY(t),\qquad
 A=\left(\ba{ccrl}
 0   &  H^0 &  0  & 0\medskip\\
-H^0-2e^2\psi_0 G\psi_0 &   0  & -S  & 0\\
        0                         &           0                        &   0    &   M^{-1}\\
     -2S^{\5*}                   &              0            &  -T    &  0\\
\ea\right),
\ee
where we denote
$Y(t)=(\Psi_1(\cdot,t),\Psi_2(\cdot,t),q(\cdot,t),p(\cdot,t))$,
$H^0:=-\fr12\De$,
the operators $S$ and $T$ correspond to
matrices
(\ref{4-S}) and (\ref{4-T}), respectively.
We keep all notations of the previous chapter
and
 assume  that the ion charge density $\si(x)$
 satisfies 
 the Wiener and Jellium conditions
 (\ref{4-Wai}) and  (\ref{4-W1}).

Our main result is 
 the dispersive decay
in the weighted Sobolev norms
for solutions with initial states from the space of  continuous spectrum of the Hamiltonian generator.
We also prove
the absence of singular spectrum and 
limiting absorption principle.
The multiplicity of every eigenvalue is shown to be infinite.

For the proof we use the formula for solutions (\ref{4-tZ0}) and
the spectral resolution  of the selfadjoint operator $\ti K(\theta)$.
We establish that its spectrum is discrete by 
the energy positivity (\ref{4-Hpos2k}).
 The dispersion relations $\om_n(\theta)$ are introduced as
 the eigenvalues of $\ti K(\theta)$. 
 The key role is played 
 by their analyticity in a complex neighborhood   of the torus $T^3$
 which allows us to apply  the  theory of analytic sets \ci{S1989}.


As in previous chapter, we assume
(\ref{ro+}) and 
\be\la{L123i}
(\De-1)\si\in L^1(\R^3)
\ee
which provides a suitable decay for the Fourier transform of $\si$.
In particular, the  series (\re{W1}) are converging.
Moreover, we  assume
 the exponential decay of the ion charge density
\be\la{rd}
 |\si(x)|\le Ce^{-\ve|x|},~~~~x\in\R^3,
\ee
where $\ve>0$.
The cubic
lattice  $\Ga= \Z^3$ is chosen for the simplicity of notations.
 The results \ci{KKpl2015} imply that the energy operator $B$ is densely defined 
and is selfadjoint on the Hilbert space
\be\la{cX0}
\cX^0(\R^3):=L^2(\R^3)\oplus L^2(\R^3)\oplus l^2(\Z^3)\oplus l^2(\Z^3).
\ee
Our main goal is  to show the dispersive decay of solutions to (\re{JDi})
in the weighted norms
\be\la{cXw}
\Vert X\Vert_\al:=\Vert \langle x\rangle^\al\,\, \Psi_1(x)\Vert_{L^2(\R^3)}+
\Vert \langle x\rangle^\al \,\, \Psi_2(x)\Vert_{L^2(\R^3)}+
\Vert \langle n\rangle^\al \, q(n)\Vert_{l^2(\Z^3)}+
\Vert \langle n\rangle^\al \, p(n)\Vert_{l^2(\Z^3)}
\ee
with $\al<0$ for $X=(\Psi_1,\Psi_2,q,p)\in\cX^0(\R^3)$.
Recall that  $\cW$  denotes the completion of the space of functions 
$Y\in\cX^1(\R^3):=H^1(\R^3)\oplus H^1(\R^3)\oplus l^2(\Z^3)\oplus l^2(\Z^3)$ with finite norm
\be\la{cW}
\Vert Y\Vert_\cW:= \Vert\Lam Y\Vert_{\cX^0(\R^3)},\qquad \Lam:= B^{1/2}>0.
\ee
By Coollary \ref{4-cd},
  for any $Y(0)\in\cW$ there exists a unique  {\it weak solution}  $Y(t)\in C(\R,\cW)$ to (\re{JDi}).
 The main result of the this chapter is the following  theorem.

\bt\la{5-tmg}
 Let  conditions \eqref{L123i},  \eqref{rd},  \eqref{W1}, and \eqref{Wai}  hold.
Then every solution $Y(t)\in C(\R,\cW)$ to (\re{JDi}) splits as follows
\be\la{FY53}
 Y(t)=\sum_1^N Y_k e^{-i\om^*_k t}+Y_c(t).
\ee
Here   $Y_k\in\cW$ and
$N\le\infty$,  
and the sum is defined by (\ref{discr}) in the case $N=\infty$. Moreover,
\be\la{Minf}
|\om^*_k|\to\infty,\qquad k\to\infty,
\ee
if $N=\infty$. The remainder $Y_c(t)$ decays in the weighted norms
(\ref{cXw}): for any $\al<-3/2$,
\be\la{FY5}
\Vert \Lam Y_c(t)\Vert_\al\to 0,~~~~~~|t|\to\infty.
\ee
\et
This theorem  means the linear  asymptotics stability of the ground state (\re{ppo})
when $N=0$.
\medskip

Let us recall basic constructions of previous chapter
relying on      the Bloch transform (\ref{4-YPi}).
Namely, the generator  $A$
commutes with  translations by vectors from $\Ga$.
Hence, the  equation
\eqref{JDi} can be reduced using the
Bloch
transform
$Y(t)\mapsto\ti Y(\cdot,t)\in L^2(\Pi^*, \cX^0(T^3))$, 
where
$T^3:=\R^3/\Ga$ is the periodic cell, and
\be\la{cX03}
\cX^s(T^3):= H^s(T^3)\oplus H^s(T^3)\oplus \C^3\oplus \C^3,\qquad s\in\R.
\ee
In the Bloch transform equation  (\re{JDi})  reads as (\ref{4-JD}),
\be\la{Ble}
\dot{\ti Y}(\theta,t)=\ti A(\theta)\ti Y(\theta,t)\quad\mbox{\rm for \,\,a.e.}\,\,\theta\in\Pi^*\setminus\Ga^*,
\qquad t\in\R,
\ee
where $\ti Y(\cdot,t)\in \cX^0(T^3)$.
The Hamiltonian representation (\re{4-AJB}) implies that
\be\la{tiAJB}
 \ti A(\theta)=J\ti B(\theta),\qquad \theta\in\Pi^*\setminus \Ga^*,
 \ee
  where the Bloch energy operators $\ti B(\theta)$
are selfadjoint in  $\cX^0(T^3)$.
The main crux here is that
the generator  $\ti A(\theta)$ is not selfadjoint 
and even is not symmetric in the Hilbert space $\cX^0(T^3)$.
Hence we cannot diagonalize it
using the von Neumann spectral theorem.
Thus, even an introduction of the `dispersion relations' $\om_k(\theta)$, which are the eigenvalues of $\ti A(\theta)$, 
is a nontrivial problem in our situation.
Let us recall the definition (\ref{PI+}),
\be\la{PI+}
\Pi^*_+:=\{ \theta\in \Pi^*\setminus \Ga^*:\Si(\theta)>0 \}. 
\ee
This is an open set of
the complete Lebesgue measure, i.e.,  $\mes(\Pi^*\setminus\Pi^*_+ )=0$, by (\re{W1}).
The key role in our approach is played by the positivity
(\ref{4-Hpos2k}) proved in the previous chapter
under the Jellium and the Wiener conditions, 
\be\la{5-Hpos2}
\langle \ti Y, \ti B(\theta)\ti Y\rangle\ge \vka(\theta)\Vert \ti Y\Vert_{{\cX^1}({T^3})}^2,
\quad\ti Y\in{\cX^1}({T^3}),\quad \theta\in\Pi^*_+,
\ee
with $\vka(\theta)>0$,  the brackets  denoting the scalar product in ${\cX^0}({T^3})$.
In this chapter,  we
use this positivity
to show
that the eigenvectors of $\ti A(\theta)$ span the Hilbert space
$\cX^0({T^3})$ by
our spectral theory of the Hamiltonian operators with positive energy \ci{KK2014a,KK2014b}.
 Namely, setting $\ti\Lam(\theta):=\ti B^{1/2}(\theta)$, we obtain 
 from (\ref{tiAJB})
 that
 \be\la{AK2}
 \ti A(\theta)=-i\ti \Lam^{-1}(\theta) \ti K(\theta)\ti \Lam(\theta),\qquad\theta\in\Pi^*_+,
  \ee
  where  $\ti K(\theta)=\ti\Lam(\theta) iJ\ti\Lam(\theta)$
   is a selfadjoint operator  in    ${\cX^0}({T^3})$
   by Lemma \ref{4-lH0}.
  Hence,
  all solutions to (\re{Ble}) admit the representation
\be\la{solexp}
 \ti Y(\theta, t)=\ti\Lam^{-1}(\theta) e^{-i \ti K(\theta)t}
\ti\Lam(\theta)\ti Y(\theta,0),\quad t\in\R,
\quad \theta\in\Pi^*_+.
\ee
Moreover, (\ref{4-nVet}) implies that 
\be\la{nVet}
 \Vert \ti\Lam^{-1}(\theta)Z\Vert_{{\ti\cX^1}(T^3)}\le \fr 1{\sqrt{ \vka(\theta)}} \Vert Z\Vert_{{\ti\cX^0}(T^3)},
 \quad Z\in{\ti\cX^0}(T^3),\qquad  \theta\in\Pi^*_+.
\ee
Hence,
(\ref{4-G}) implies that $\ti K^{-1}(\theta)=i\ti\Lam^{-1}(\theta) J^{-1}\ti\Lam^{-1}(\theta)$ and
\be\la{qfB23}
\Vert
\ti K^{-1}(\theta) \ti Z \Vert_{\cX^1({T^3})}^2
\le \fr C{\vka(\theta)}\Vert \ti Z\Vert_{\cX^0({T^3})}^2,\qquad \ti Z\in \cX^0({T^3}).
\ee

We prove that the spectrum of $\ti K(\theta)$ is discrete and obtain the lower estimate for the eigenvalues
 $\om_k(\theta)$ which are 
 also the eigenvalues of  $\ti A(\theta)$
and are called 
 the dispersion relations (or the 
Floquet eigenvalues).

\medskip

Further, we represent
the solution $Y(t)$ as the inversion of the Bloch transform  (\re{solexp}).
This inversion is the series of  oscillatory  integrals
with the phase functions $\om_k(\theta)$.
Using the decay (\re{rd}) we  show that
\medskip\\
i)
$\om_k(\theta)$ are piecewise real-analytic in $\theta\in\Pi^*\setminus \Ga^*$
for every $k$;
\medskip\\
ii)
If $\om_k(\theta)\not\equiv\const$, then the set
\be\la{deg}
\{\theta\in\Pi^*\setminus \Ga^*: \na\om_k(\theta)= 0,\,\,
\det\,{\rm Hess}\,\om_k(\theta)= 0\}
\ee
has the Lebesgue measure zero;
\medskip\\
iii) In the case $N=\infty$
the limit (\re{Minf}) holds
for the constant dispersion relations $\om_k(\theta)\equiv \om^*_k$.
\medskip

These properties of the phase functions provide
the asymptotics  (\re{FY53}) and  (\re{FY5}).
Finally, we establish the
absence  of singular spectrum and
the limiting absorption principle for the selfadjoint operator $K:=i\Lam A\Lam^{-1}$.
\smallskip

Note that
all our methods and results extend obviously to
  the case of a~general lattice (\ref{4-gG})
\medskip

Let us comment on previous results 
on the dispersive decay for space-periodic 
dynamical systems.
\smallskip

The first results on the dispersive decay
$\sim t^{-1}$
were obtained by Firsova \ci{Fir1996}
for 1D  Schr\"odinger equation with space-periodic
potential  
for finite band case. The proofs
 rely on Korotyaev's results \ci{Kor1991}
on stationary
points of the dispersion relations.

The decay
$\sim t^{-\ve}$ with a small $\ve>0$
for the 1D Schr\"odinger equation
with an infinite band potential
was  established by Cuccagna
\ci{Cuc2008s}.
This decay was applied to the asymptotic  stability of standing waves
in presence of
small nonlinear perturbations \ci{Cuc2006}.

The absense of constant  dispersion relations
for the periodic
Schr\"odinger equations was  established
by Thomas \ci{Thom1973}, see also Lemma~2~(c) of  \ci{RS4}, p.~308.

Recently Prill \ci{Prill2014} proved the decay $\sim t^{-p}$ with  $p=3/2$ and $p=1/2$
(under distinct assumptions)
for
the 1D Klein-Gordon equation  with
a periodic Lam\'e potential and its short range perturbations.

The dispersive decay
for the periodic Schr\"odinger and Klein-Gordon  equations
in higher dimensions $n\ge 2$ was not obtained
previously.
\smallskip

This chapter is organized as follows.
In Section 2 we recall some formulas from~\ci{KKpl2015} for
the Bloch representation.
In Section 3 we introduce the dispersion relations and prove their properties.
In Section 4 we prove the asymptotics (\re{FY53}), (\re{FY5}), and
in Section 5 we justify the limiting absorption principle.
In Appendix A we collect some formulas from \ci{KKpl2015}
which we need in our calculations.
\





\setcounter{equation}{0}
\section{Dispersion relations}\la{dr}
  Here we establich the properties of the eigenvalues of  $\ti K(\theta)$ which play the key role in the proof of the 
   dispersive decay.
Lemma \re{4-lH0} implies the spectral resolution
\be\la{Hsr}
 \ti K(\theta)=\sum_{k=1}^\infty \om_k(\theta)P_k(\theta),~~~~~~\theta\in \Pi^*_+,
\ee
where  $\om_k(\theta)$
 are the eigenvalues
(dispersion relations)
counted with their multiplicities,
$$
|\om_1(\theta)|\le |\om_2(\theta)|\le\dots,
$$
and
$P_k(\theta)$ are the corresponding orthogonal projections.

\bl \la{lQ}
Let   conditions
 \eqref{L123i} and    \eqref{W1},  \eqref{Wai} hold and
$Q$ be a compact subset of $\Pi^*_+$. Then
\be\la{omk1}
|\om_k(\theta)|\ge \ve(Q)k^{2/3},\qquad k\ge 1,\qquad\theta\in Q,
\ee
where $\ve(Q)>0$.
\el
\begin{proof}
The key role in the proof of (\re{omk1}) is played by the estimate \ci[(7.23)]{KKpl2015}:
\be\la{com}
b(Q):=
\inf_{\theta\in Q}\vka(\theta)>0
\ee
for any compact subset $Q\subset\Pi^*_+$.  The  expansion (\re{Hsr})  implies that 
\be\la{Hsr2}
 |\ti K^{-1}(\theta)|=\sum_{k=1}^\infty |\om_k(\theta)|^{-1}P_k(\theta),~~~~~~\theta\in \Pi^*_+.
\ee
Moreover,  by duality we have from estimate (\re{qfB23})  
\be\la{qfB24}
\Vert
\ti K^{-1}(\theta) \ti Z \Vert_{\cX^0({T^3})}^2\le \fr C{b(Q)}\Vert \ti Z\Vert_{\cX^{-1}({T^3})}^2,\qquad \ti Z\in \cX^0({T^3}),\quad \theta\in Q
\ee
due to (\re{com}), since the operator $\ti K^{-1}(\theta)$ is selfadjoint.
At last, the norm in the right-hand side of (\ref{qfB24}) can be written
as $\Vert   g \ti Z\Vert_{\cX^{0}({T^3})}$, where
\be\la{g}
g=\left(\ba{cccc}
(-\De+1)^{-1/2}&0&0&0\\
0&(-\De+1)^{-1/2}&0&0\\
0&0&1&0\\
0&0&0&1
\ea
\right)
\ee
is the positive selfadjoint operator in $\cX^{0}({T^3})$. Now (\re{qfB24}) gives that
\be\la{qfB25}
\Vert
|\ti K^{-1}(\theta) | \ti Z \Vert_{\cX^0({T^3})}\le C(Q)\Vert g \ti Z\Vert_{\cX^{0}({T^3})},\qquad \ti Z\in \cX^0({T^3}),\quad \theta\in Q.
\ee
Hence,
the Rayleigh-Courant-Fisher theorem  (\ci[Theorem 1, p.110]{A}
and \ci[Theorem XIII.1, p.91]{RS4}) implies that
\be\la{RCF}
|\om_k(\theta)|^{-1}\le C(Q) g_k,\qquad k\ge 1,\quad  \theta\in Q,
\ee
 where $g_1\ge g_2\ge\dots$ are the eigenvalues of $g$ counted with their multiplicities.
 Therefore, (\re{omk1}) holds, since $g_k\le Ck^{-2/3}$.
 The last inequality is obvious, in as much as $k\le \#(n\in\Z^3:n^2+1 \le g_k^{-1})\le C_1 g_k^{-3/2}$.
  \end{proof}
Further we use the exponential decay of the ion charge density (\re{rd}).
Example \ref{4-ex2} gives the 
 densities $\si$ satisfying
 all conditions of Theorem \re{5-tmg}:
(\re{L123i}), (\re{rd})  and the Wiener and Jellium conditions
(\re{W1}), (\re{Wai}).
 The decay (\re{rd}) implies that
 the function  $\ti\si(\theta,y)$ is analytic with respect to~$\theta$ in the complex tube
\[
 \Pi^*_\ve:=\{\theta\in [\Pi^*\setminus\Ga^*]\oplus i\R^3:~|\rIm\theta|<\ve\}.
\]
Hence, the finite rank operators $\ti S(\theta)$
and $\ti T(\theta)$ defined in (\ref{4-tiHS}) -- (\ref{4-tiH1})
are also analytic in $\theta\in \Pi^*_\ve$. Therefore, $\ti K(\theta)$
is real-analytic on  $\Pi^*_+$.
Denote the set
\be\la{spec}
\cR:=\{(\theta,\om): \,\,\theta\in  \Pi^*_+, \,\,\om\in\spec\ti K(\theta)\}.
\ee
The eigenvalues $\om_k(\theta)$ and  the projections $P_k(\theta)$ become 
 single-valued functions on $\cR$: for $R=(\theta,\om_k(\theta))$
 \be\la{func}
 \theta(R):=\theta,\qquad
 \om(R):=\om_k(\theta),\qquad
  P(R):=P_k(\theta).
 \ee
 These functions are
  continuous on the manifold
  $\cR$ endowed with natural topology  by the incluzion $\cR\subset \Pi^*\times\R$.
 They are
  piecewise analytic on $\cR$
 by the following lemma,
 which extends \ci[Lemma 1.1]{S1989} from the Schr\"odinger equation with 
 periodic potential to the system (\re{JDi}).
\begin{lemma}\la{lom}
Let  all conditions of Theorem \re{5-tmg} hold.
Then
for every point $R^*=(\theta^*,\om^*)\in\cR$  there exists a neighborhood
$U=U(R^*)\subset \cR$ with its projection $V=V(R^*)$ onto $\Pi^*_+$\, ,
and  a critical subset
$\cC=\cC(R^*)\subset \Pi^*_\ve$, which is a finite union of analytic   submanifolds
of positive complex codimension in $\Pi^*_\ve$,
with the following properties:
\medskip\\
i) For any point
$R=(\theta,\om) \in U$
we have $\om(R):=\om\in \spec \ti K(\theta)$.  
\medskip\\
ii) For  any point $\theta' \in V\setminus \cC$  
there exists a neihborhood $W=W(\theta')\subset V\setminus \cC$ such that
$R=(\theta,\om) \in U$ with $\theta\in  W$ if and only if 
$\om=\om_l(\theta)$
with some $l=1,..., L=L(R^*)$. 
\medskip\\
iii)  The eigenvalues $\om_l(\cdot)$ and the corresponding projections  $P_l(\cdot)$ are  real-analytic functions  on   $W$
and admit an analytic continuation outside $\cC$ in a complex neighborhood of $\theta^*$ in $\Pi^*_\ve$.
\medskip\\
iv) For each $l=1,..., L(R^*)$, either
\be\la{Ck}
  \na \om_{l}(\theta)\ne 0,~~~~\theta\in W,
\ee
or
\be\la{omc}
  \om_{l}(\theta)\equiv \om^*,~~~~~~\theta\in W.
\ee
v) If  (\re{omc}) holds with some $l$ for a point $R^*=(\theta^*,\om^*)$, 
then 
the constant eigenvalue
also exists for
 $(\theta,\om^*)$ with any
$\theta\in\Pi^*_+$.
\end{lemma}
\begin{proof}
Let us set $r:=\dist (\om^*, \spec \ti K(\theta^*)\setminus\om^*)>0$. Then
\be\la{Rp1}
 P(\theta)=-\fr1{2\pi i}
 \int_{|\om-\om^*|=r/2} [\ti K(\theta)-\om]^{-1}d\om
\ee
is a finite-rank 
Riesz 
 projection, which is analytic in a complex neighborhood of $\theta^*$.
Its range $\Ran P(\theta)$ is invariant under $\ti K(\theta)$, and hence the bifurcated  from
$\om^*$ eigenvalues of $\ti K(\theta)$ coincide with the roots of the characteristic equation
\be\la{Rp2}
 \det[m(\theta)-\om]=0,
\ee
where
$m(\theta):=\ti K(\theta)|_{\Ran P(\theta)}$.
The coefficients of this polynomial are analytic
functions of $\theta$
in a complex neighborhood of $\theta^*$, and hence
i)--iv) follow
by the arguments from the proof of  Lemma 1.1 of \ci{S1989}.

Finally,  v) follows from the fact that the set of the corresponding $\theta\in\Pi^*_+$ is closed and open at the same time by the analyticity
of each $\om_l(\theta)$ in a connected open region of  $\Pi^*_\ve\setminus\cC$.
\end{proof}

\begin{definition}\la{dOm*}
 $\Om^*$ is the set of all 
 $\om^*$ which are 
constant eigenvalues (\re{omc}) at least for one point
$R^*\in\cR$. 
\end{definition}


\setcounter{equation}{0}
\section{Dispersion decay}
Here we prove our main Theorem \re{5-tmg}.
Recall that
$\Lam: \cW\to\cX^1(\R^3)$  is an isomorphism by  the definition (\re{cW}), and hence,
 it suffices to check the corresponding asymptotics for $Z(t):=\Lam Y(t)\in C(\R,\cX^0(\R^3))$:
\be\la{FY52}
 Z(t)=\sum_1^N Z_k e^{-i\om^*_k t}+Z^c(t);
\qquad \Vert Z^c(t)\Vert_\al \to 0~,\quad |t|\to\infty,
\ee
where $Z_k\in\cX^0(\R^3)$ and $ \al<-3/2$.
Using the inversion formula  (\re{4-FZI})
and the representation 
 (\re{4-ZH0}) for $\ti Z(\theta,t)$,
we obtain the corresponding `cell representation'
\be\la{FY2}
 Z(n,t)=|\Pi^*|^{-1}\int_{\Pi^*} e^{-in\theta}\cM(-\theta)e^{-i\ti K(\theta)t}
 \ti Z(\theta,0)d\theta,\quad n\in\Z^3.
\ee
The spectral resolution
 (\re{Hsr}) implies that
\be\la{FY22}
 Z(n,t)=|\Pi^*|^{-1}\int_{\Pi^*}e^{-in\theta}\cM(-\theta)
 [\sum_k e^{-i\om_k(\theta)t}P_k(\theta)]\ti Z(\theta,0)d\theta,\qquad n\in\Z^3.
\ee
Equivalently,
\be\la{FY222}
 Z(n,t)=|\Pi^*|^{-1}\int_{\cR}e^{-in\theta}\cM(-\theta)
 e^{-i\om t} P(\theta,\omega)\ti Z(\theta,0)  d\theta,\qquad n\in\Z^3, 
\ee
where  $\theta$, $\om$ and the projection $P(\theta,\omega)$ are the single-valued continuous functions
(\re{func}) on $\cR$. We denote by $d\theta$ is the corresponding differential form on $\cR$.
The integral is well defined by Lemma \re{lom}. 
\medskip

\subsection{Discrete spectral component}  
We define the series of oscillating terms of (\re{FY52}) by its cell representation
\be\la{discr}
\sum_k Z_k (n)e^{-i\om^*_k t}=
|\Pi^*|^{-1}\int_{\{(\theta,\om)\in\cR:\om\in \Om^*\}}e^{-in\theta}\cM(-\theta)
 e^{-i\om t} P(\theta,\omega)\ti Z(\theta,0)  d\theta,\qquad n\in\Z^3.
\ee
Obviously,  (\re{Minf}) follows from (\re{omk1}).
\medskip

\subsection{Continuous spectral component} 
It remains to prove the decay (\re{FY52}) for the remainder corresponding to the cell representation
\be\la{FY222d}
 Z^c(n,t)=|\Pi^*|^{-1}\int_{\cV}e^{-in\theta}\cM(-\theta)
 e^{-i\om t} P(\theta,\omega)\ti Z(\theta,0)  d\theta,\qquad n\in\Z^3,
\ee
where the integration  spreads over the set
$\cV:=\{(\theta,\om)\in\cR:\om\not\in \Om^*\}$.
For every $\nu>0$, we split $Z^c(t)=Z^\nu_-(t)+Z^\nu_+(t)$, where
\beqn
 Z^\nu_-(n,t)&=&|\Pi^*|^{-1}\ds\int_{\cV^\nu_-} e^{-in\theta}\cM(-\theta)
  e^{-i\om t}P(\theta,\om)\ti Z(\theta,0)d\theta,
 \la{Znu}\\
 \nonumber\\
 \la{Rnu}
  Z^\nu_+(n,t)&=&|\Pi^*|^{-1}\ds\int_{\cV^\nu_+}e^{-in\theta}\cM(-\theta)
  e^{-i\om t}P(\theta,\om)\ti Z(\theta,0)d\theta.
\eeqn
Here $\cV^\nu_-:=\{(\theta,\om)\in\cV:  |\om|\le\nu\}$ and
$\cV^\nu_+:=\{(\theta,\om)\in\cV:  |\om| >\nu\}$.
\medskip\\
  {\bf  High energy  component.}
By  (\re{4-nV}) and the
Parseval--Plancherel theorem
\be\la{Zn}
\Vert  Z^\nu_+(t)\Vert_{\cX^0(\R^3)}^2=\sum_{n\in\Z^3} \Vert  Z^\nu_+(n,t)\Vert_{\cX^0(\Pi)}^2=
 |\Pi^*|^{-1}\int_{\cV^\nu_+}
\Vert P(\theta,\om)\ti Z(\theta,0)\Vert_{\cX^0({T^3})}^2 d\theta.
 \ee
 According to definition (\re{cW})
the condition $Y(0)\in\cW$ means that $Z=\Lam Y(0)\in\cX^0(\R^3)$.
Hence, the Parseval--Plancherel identity gives
\be\la{5-Wm}
\Vert Z(0)\Vert^2_{\cX^0(\R^3)}=|\Pi^*|^{-1} \int_{\Pi^*} \Vert\ti Z(\theta,0)\Vert_{\cX^0({T^3})}^2 d\theta<\infty.
\ee
Therefore,    (\re{Zn}) implies that
\be\la{Zn2}
\Vert  Z^\nu_+(t)\Vert_{\cX^0(\R^3)}
\to 0, \qquad \nu\to\infty
\ee
uniformly in $t\in\R$
by the $\si$-additivity
since $\cap_{\nu>0} \cV_+^\nu=\emptyset $. 
\medskip\\
{\bf Low energy component.}
It remains  to prove the decay (\re{FY52}) for $Z^\nu_-(t)$ corresponding to 
the cell representation
$Z^\nu_-(n,t)$.
The weighted norms  (\re{cXw}) are equivalent to the modified norms 
\be\la{eqn}
|\!|\!| Z|\!|\!|_\al ^2:=\sum_{n\in\Z^3}(1+|n|)^{2\al} \Vert Z(n) \Vert_{\cX^0(\Pi)}^2,\qquad  Z\in\cX^0(\R^3),
\ee
where $Z(n)$
 are defined as in (\re{4-Yn2}), (\re{4-YP}). Hence, 
 the decay  (\re{FY52}) for $Z^c(t)$ is equivalent to
 \be\la{dZc}
 \sum_{n\in \Z^3}(1+|n|)^{2\al} \Vert Z^c(n,t) \Vert_{\cX^0(\Pi)}^2\to 0,
 \quad t\to\infty.
\ee
It suffices to check that every norm
$\Vert Z^\nu_-(n,t) \Vert_{\cX^0(\Pi)}$ decays to zero as $t\to\infty$, since $\al<-3/2$ and
\be\la{dZc2}
 \sum_{n\in \Z^3} \Vert Z^\nu_-(n,t) \Vert_{\cX^0(\Pi)}^2=\Vert  Z^\nu_-(t) \Vert_{\cX^0}^2=\const, \qquad t\in\R
 \ee
 by (\re{4-nV}) and formula of type (\re{4-ZH0}) for $\ti Z^\nu_-(\theta,t)$.
 \medskip\\
{\bf Reduction to a compact set and partition of unity.} 
Consider an open precompact subset $Q\subset\Pi^*_+$ such that the Lebesgue measure of
$\Pi^*_+\setminus Q$ is sufficiently small, and denote
$\hat Q^\nu:=\{R=(\theta,\om)\in\cV^\nu_-: \theta\in Q\}$.
Then
the $\cX^0(\Pi)$-norm of the integral  of type (\re{Znu})  over
$\cV^\nu_-\setminus \hat Q^\nu$ is small uniformly in $t\in\R$ by (\re{5-Wm}).
Hence,
it remains to prove the decay for
\be\la{Zn3}
Z^\nu_Q(n,t):=|\Pi^*|^{-1}\int_{\hat Q^\nu}
e^{-in\theta}\cM(-\theta)
  e^{-i\om t}P(\theta,\omega)\ti Z(\theta,0)d\theta.
\ee
The  asymptotics (\re{omk1}), which are
uniform in $\theta\in Q$, imply that the set $\hat Q^\nu$
is open and precompact in $\cR$. 
Neglecting an arbitrarily  small term we can assume that $Q$ does not intersect a small neighborhood of the
critical submanifold
 $\cC_j\subset V(\theta_j)$ for every $j$.
Hence, we can cover $\hat Q^\nu$
 by a finite number of neighborhoods
$W(R_j)$ from Lemma \re{lom} 
with $R_j=(\theta_j,\om_j)\in\ov{\hat Q^\nu}$.
Then there exists a 
partition of unity
$\chi_j\in C(\cR)$   with $\supp\chi_j\subset W(R_j)$:
\be\la{pu}
 \sum_j\chi_j(R)=1,\quad R=(\theta,\omega)\in \hat Q^\nu.
 \ee
 Hence, (\re{Zn3}) becomes the finite sum
\be\la{FY223}
 Z^\nu_{jl}(n,t)=
 \sum_{j,l}
 |\Pi^*|^{-1}
 \int_{W(R_j)}
 e^{-in\theta}
 \chi_j(\theta,\om_{jl}(\theta))\cM(-\theta)e^{-i\om_{jl}(\theta)t}P_{jl}(\theta)\ti Z(\theta,0)d\theta,
\ee
where the functions $\om_{jl}$ and projections $P_{jl}$ are constructed in Lemma \re{lom}. 
Note, that  all constant dispersion relations (\re{omc}) are excluded from the integration (\re{FY223}), and hence,
the remaining nonconstant dispersion relations $\om_{jl}(\theta)$ satisfy (\re{Ck}).
 Let us approximate
\medskip\\
i) $\chi_j(\theta,\om_{jl}(\theta))$ by $\chi_{jl}(\cdot)\in C_0^\infty(W(R_j))$ and
\medskip\\
ii) $P_{jl}(\theta)\ti Z(\theta,0)$
by some  functions $D_{jl}(\theta)\in C^\infty(W(R_j),\cX^0({T^3}))$  in the norm of $L^2(Q,\cX^0({T^3}))$.
\medskip\\
Then the corresponding error in (\re{FY223})  is small
in the norm $\cX^0(\Pi)$ uniformly in $n\in\Z^3$ and $t\in\R$.
Finally,  (\re{Ck}) implies by a partial integration the decay of the integrals  (\re{FY223})
with $\mu_{jl}(\theta)D_{jl}(\theta)$ instead of $\chi_j(\theta,\om_{jl}(\theta))P_{jl}(\theta)\ti Z(\theta,0)$.
Theorem \re{5-tmg} is proved.

\setcounter{equation}{0}
\section{Spectral properties of the selfadjoint generator}
Here we study
spectral properties of the 
operator $K:=\cF^{-1}\ti K\cF$, where $\ti K$ denotes the operator of multiplication by $\ti K(\theta)$ 
in the Hilbert space
 $L^2(\Pi^*,\cX^0({T^3}))$. 
  \bl
$K$  is  a selfadjoint operator in  $\cX^0(\R^3)$ with a dense domain $D(K)$. 
 \el
 \begin{proof}
 Lemma \re{4-lH0} ii) implies that the operator  of multiplication by $\ti K^{-1}(\theta)$ is  selfadjoint and injective in $L^2(\Pi^*,\cX^0({T^3}))$.
 Hence,  its inverse   is densely defined  selfadjoint operator  in $L^2(\Pi^*,\cX^0({T^3}))$ by Theorem 13.11 (b) of \ci{Rudin}.
 \end{proof}
\bc
The Hamiltonian generator $A$ from (\re{JDi})
admits the representation
\be\la{AK2f}
 AY=-i\Lam^{-1}K \Lam Y,\qquad Y\in\Lam^{-1} D(K),
  \ee
where $\Lam:=\cF^{-1}\ti\Lam(\theta):\cW\to\cX^0(\R^3)$ is the isomorphism.
\ec

By (\re{Hsr}),
\be\la{FY24}
 KZ(n)=|\Pi^*|^{-1}\int_{\Pi^*_+}e^{-in\theta}\cM(-\theta)
 \sum_k \om_k(\theta) P_k(\theta)\ti Z(\theta)d\theta,~~~~~~~n\in\Z^3
\ee
for any $Z\in \cX^0(\R^3)$.
Similarly
\be\la{FY23}
 Z(n)=|\Pi^*|^{-1}\int_{\Pi^*_+}e^{-in\theta}\cM(-\theta)
 \sum_k P_k(\theta)\ti Z(\theta)d\theta,~~~~~~~n\in\Z^3.
\ee
Therefore,
\be\la{FY25}
(K-\om)Z(n)=|\Pi^*|^{-1}\int_{\Pi^*_+}e^{-in\theta}\cM(-\theta)
\sum_k (\om_k(\theta)-\om) P_k(\theta)\ti Z(\theta)d\theta.
\ee
Hence, the discrete spectrum $\si_p(K)$ consists of constant dispersion relations.

\begin{lemma}\la{leig}
Let all conditions of Theorem \re{5-tmg} hold. Then
\medskip\\
i)
$\si_p(K)=\Om^*$.
\medskip\\
ii) The multiplicity of every  eigenvalue is  infinite.
\end{lemma}
\begin{proof}
Let $\om^*\in\Om^*$ is a constant eigenvalue  (\re{omc}) corresponding to a point $R^*=(\theta^*,\om^*)\in\cR$.
Let us take
any    $Z\in\cX^0(\R^3)$ with the Bloch transform $\ti Z(\theta) \in\Ran P(\om^*)$ for
$\theta\in V(\theta^*)$ and $\ti Z(\theta)\equiv 0$ for $\theta\not\in V(\theta^*)$.
Then 
(\re{omc}) and (\re{FY25})  imply that $(K-\om^*)Z=0$. Obvioulsy, the space of such $Z$ is infinite dimensional.
\medskip\\
Conversely, let $(K-\om^*)Z=0$ for some $Z\in\cX^0(\R^3)$, and $\ti Z(\theta^*)\ne 0$. Then (\re{FY25}) implies (\re{omc}) with some $l=1,...,L(\theta^*,\om^*)$.
\end{proof}

Let us show   that the continuous  spectrum of $K$ is absolutely continuous.
First,
(\re{FY25}) implies that
 the resolvent  $R_K(\om):=(K-\om)^{-1}$ for $\rIm\om\ne 0$ is given by
\be\la{FY26}
R_K(\om)Z(n)=|\Pi^*|^{-1}\int_{\Pi^*_+}e^{-in\theta}\cM(-\theta)
\sum_k (\om_k(\theta)-\om)^{-1} P_k(\theta)\ti Z(\theta)d\theta,\qquad Z\in\cX^0(\R^3).
\ee
Denote by $\cX_d$ the space  of discrete spectrum of $K$.
\bl\la{lss}
Let all conditions of Theorem \re{5-tmg} hold. Then
the singular spectrum of $K$ is empty.
\el
\begin{proof}
This follows by  Theorem XIII.20 of  \ci{RS4}. Namely, it suffices to check
the corresponding  criterion
\be\la{crit}
\sup_{0<\ve<1}
\int_a^b|\rIm\langle Z, R_K(\om+i\ve) Z\rangle |^pd\om<\infty
\ee
with any $a,b\in\R$ and some $p>1$ for a dense set of $Z\in \cX_d^\bot$.
For example, for  the linear span of vectors  $Z\in\cX^0(\R^3)$ with the Bloch transform
\be\la{YV}
\ti Z(\theta)=P_l(\theta)D(\theta),~~~~~~~~
 D\in C_0^\infty(W, \cX^0({T^3})),
\ee
as constructed in Lemma \re{lom} for each $R^*=(\theta^*,\om^*)\in\cR$,
where $P_l(\theta)$ is  the projection
corresponding to an eigenvalue $\om_l(\theta)$ satisfying   (\re{Ck}).
It suffices to check (\re{crit})
only for the vectors of type (\re{YV}). Applying Sokhotski-Plemelj's formula, we obtain  for these vectors
\beqn\la{crit2}
\rIm\langle Z, R_K(\om+i\ve) Z\rangle
&=&\,\,\,\int_W\rIm\langle P_l(\theta)D(\theta), (\om-\om_l(\theta)-i\ve)^{-1} P_l(\theta)D(\theta)\rangle_{\cX({T^3})}d\theta
\nonumber\\
&\to\!\!&\!\!-\!\pi\!\!\int_{\om_l(\theta)=\om}\fr{\langle P_l(\theta)D(\theta),P_l(\theta)D(\theta)\rangle_{\cX({T^3})}}
{|\na\om_l(\theta)|}d\theta,\qquad \ve\to 0+,
\eeqn
which implies (\re{crit}) with any $p\ge 1$.
\end{proof}

In concluzion, let us prove the Limiting Absorption Principle. Let us denote by $\cX_\al$ the Hilbert space
of functions with the finite norm (\re{eqn}).

\begin{lemma}\la{llap}
Let all conditions of Theorem \re{5-tmg} hold,
and let $Z\in\cX_d^\bot$ be
a finite linear combination of
the vectors
with the Bloch transform
of type   (\re{YV}).
Then   for any
$\om\in \R$ and
$\al<-7/2$
\be\la{lap}
R_K(\om\pm i\ve)Z\toYs R_K(\om\pm i0)Z,~~~~~~\ve\to +0.
\ee
\end{lemma}
\begin{proof}
It suffices to prove
(\re{lap}) for every vector  of type   (\re{YV}).
By (\re{FY22})
the corresponding solution $Z(t)$ with $Z(0)=Z$ reads
\[
 Z(n,t)=|\Pi^*|^{-1}\int_W e^{-in\theta}
 \cM(-\theta)e^{-i\om_l(\theta)t}P_l(\theta)\ti Z(\theta)d\theta,\quad n\in\Z^3
\]
The partial integration shows   the time-decay
\[
 \Vert Z(n,t)\Vert_{\cX^0(\Pi)}\le C(1+|n|)^2(1+|t|)^{-2}.
\]
Hence,
\[
 \Vert Z(t)\Vert_{\cX_\al}\le C(1+|t|)^{-2}.
\]
Now the convergence (\re{lap}) follows from the integral representation
\[
R_K(\om\pm i\ve)Z=\int_0^{\pm \infty} e^{(i\om\mp\ve)t}Z(t)dt.
\]
\end{proof}














\end{document}